\documentclass[english,12pt]{article}
\usepackage[latin1]{inputenc}
\usepackage{babel}
\usepackage{ngerman}
\usepackage{amsmath,amssymb, mathrsfs, dsfont, amsthm, color}
\usepackage[margin=1in]{geometry}
\usepackage{tikz}

\usepackage{amsxtra}
\usepackage{stmaryrd}

\renewenvironment{thebibliography}[1]{%
\begin{oldthebibliography}{#1}%
\setlength{\baselineskip}{.9em}
\linespread{.9}
\small
\setlength{\parskip}{0ex}%
\setlength{\itemsep}{.1em}%
}%
{%
\end{oldthebibliography}%
}

\usepackage{amsfonts}
\usepackage{amsthm}
\usepackage{amsmath}
\usepackage{amssymb}
\usepackage{bbm}
\usepackage{dsfont}
\newtheorem{thm}{Theorem}[section]
\newtheorem{defi}[thm]{Definition}
\newtheorem{prop}[thm]{Proposition}

\newtheorem{lemma}[thm]{Lemma}

\theoremstyle{definition}
\newtheorem{remark}[thm]{Remark}

\newtheorem*{claim}{Claim}

\newcommand{\bt}{\begin{thm}}
\newcommand{\et}{\end{thm}}
\newcommand{\br}{\begin{remark}}
\newcommand{\er}{\end{remark}}
\newcommand{\bl}{\begin{lemma}}
\newcommand{\el}{\end{lemma}}
\newcommand{\bp}{\begin{proof}}
\newcommand{\ep}{\end{proof}}
\newcommand{\bal}{\begin{align*}}
\newcommand{\eal}{\end{align*}}
\newcommand{\bi}{\begin{itemize}}
\newcommand{\be}{\begin{equation}}
\newcommand{\ee}{\end{equation}}
\newcommand{\bea}{\begin{eqnarray}}
\newcommand{\eea}{\end{eqnarray}}
\newcommand{\ba}{\begin{align*}}
\newcommand{\ea}{\end{align*}}
\newcommand{\ei}{\end{itemize}}
\newcommand{\bc}{\begin{claim}}
\newcommand{\ec}{\end{claim}}

\DeclareMathOperator{\sol}{sol}
\DeclareMathOperator{\conv}{conv}

\DeclareMathOperator{\Var}{Var}

\newcommand{\R}{\mathbb{R}}

\newcommand{\N}{\mathbb{N}}

\newcommand{\F}{\mathcal{F}}
\newcommand{\cF}{\mathcal{F}}

\newcommand{\cA}{\mathcal{A}}
\newcommand{\Om}{\Omega}
\newcommand{\om}{\omega}

\newcommand{\cS}{\mathcal{S}}
\newcommand{\cZ}{\mathcal{Z}}

\newcommand{\tvp}{\widetilde\varphi}
\newcommand{\hvp}{\widehat\varphi}

\newcommand{\hy}{\widehat{y}}

\newcommand{\vp}{\varphi}
\newcommand{\vr}{\varrho}
\newcommand{\ve}{\varepsilon}

\newcommand{\omt}{(\om,t)}
\newcommand{\OmT}{\Om\times[0,T]}

\newcommand{\cD}{\mathcal{D}}

\newcommand{\tS}{\widetilde S}
\newcommand{\tQ}{\widetilde Q}
\newcommand{\tZ}{\widetilde Z}
\newcommand{\tV}{\widetilde V}
\newcommand{\tM}{\widetilde M}
\newcommand{\hS}{\widehat S}

\newcommand{\hY}{\widehat Y}

\newcommand{\sint}{\stackrel{\mbox{\tiny$\bullet$}}{}}
\newcommand{\hh}{\widehat h}

\newcommand{\cC}{\mathcal{C}}
\newcommand{\cY}{\mathcal{Y}}
\newcommand{\cB}{\mathcal{B}}

\numberwithin{equation}{section}
\newcommand{\Inf}{\inf\limits}
\newcommand{\Lim}{\lim\limits}

\begin{document}
\title{Duality Theory for Portfolio Optimisation under Transaction Costs}
\author{Christoph Czichowsky\footnote{Department of Mathematics, London School of Economics and Political Science, Columbia House, Houghton Street, London WC2A 2AE, UK, {\tt c.czichowsky@lse.ac.uk}. Financial support by the Swiss National Science Foundation (SNF) under grant PBEZP2\_137313 is gratefully acknowledged.} 
\hspace{20pt}Walter Schachermayer\footnote{Fakult\"at f\"ur Mathematik, Universit\"at Wien, Oskar-Morgenstern-Platz 1, A-1090 Wien, {\tt walter.schachermayer@univie.ac.at}. Partially supported by the Austrian Science Fund (FWF) under grant P25815, the European Research Council (ERC) under grant FA506041 and by the Vienna Science and Technology Fund (WWTF) under grant MA09-003.}
}
\date{\today}
\maketitle

\begin{abstract}
\noindent
For portfolio optimisation under proportional transaction costs, we provide a duality theory for general c\`adl\`ag price processes. In this setting, we prove the existence of a dual optimiser as well as a shadow price process in a generalised sense. This shadow price is defined via a ``sandwiched'' process consisting of a predictable and an optional strong supermartingale and pertains to all strategies which remain solvent under transaction costs. We provide examples showing that in the present general setting the shadow price process has to be of this generalised form.

\end{abstract}
\noindent
\textbf{MSC 2010 Subject Classification:} optimiser\newline
\vspace{-0.2cm}\newline
\vspace{-0.2cm}\newline


\section{Introduction}
Utility maximisation under transaction costs is a classical problem in mathematical finance and essentially as old as its frictionless counterpart. A basic question in this context is whether or not it actually makes a difference, if one considers this problem with or without transaction costs, after passing to an appropriate shadow price process. In this paper, we develop a general duality theory for utility maximisation under transaction costs that allows us to investigate this question. Moreover, we provide examples that illustrate the new phenomena that arise due to the transaction costs in our general framework and cannot be observed in frictionless financial markets.

The prototype of such a duality theory has been derived by Cvitanic and Karatzas in their seminal paper \cite{CK96}. For this, Cvitanic and Karatzas used the density processes of \emph{consistent price systems} introduced by Jouini and Kallal \cite{JK95} as dual variables. These are two-dimensional processes $Z=(Z^0_t, Z^1_t)_{0 \leq t \leq T}$ that consist of the density process $Z^0=(Z^0_t)_{0 \leq t \leq T}$ of an equivalent local martingale measure $Q$ for a price process $\widetilde{S}=(\widetilde{S}_t)_{0 \leq t \leq T}:=\frac{Z^1}{Z^0}$ evolving in the bid-ask spread $[(1-\lambda) S,S]$. Requiring that $\widetilde{S}$ is a local martingale under $Q$ is tantamount to the product $Z^1=Z^0\widetilde{S}$ being a local martingale. Under transaction costs these processes play a similar role as equivalent local martingale measures in the frictionless theory. Assuming that the optimiser to the dual problem exists as a local martingale, denoted by $\widehat{Y}=(\widehat{Y}^0_t, \widehat{Y}^1_t)_{0 \leq t \leq T},$ Cvitanic and Karatzas showed in an It\^o process model that the duality theory applies. It is ``folklore'' that the ratio $\widehat{S}:= \frac{\widehat{Y}^1}{\widehat{Y}^0}$ is then a so-called \emph{shadow price process}. This is a price process evolving within the bid-ask spread such that frictionless trading (i.e.~ trading without transaction costs) for this price process yields the same optimal trading strategy and utility as in the original problem under transaction costs. This implies in particular that the optimal trading strategy under transaction costs only buys stocks, if the ratio $\widehat{S}= \frac{\widehat{Y}^1}{\widehat{Y}^0}$ is at the (higher) ask price $S$, and only sells stocks, if it is at the (lower) bid price $(1-\lambda)S$.

In this paper, we establish a duality theory pertaining to general strictly positive c\`adl\`ag price processes $S=(S_t)_{0 \leq t \leq T}$. Without imposing unnecessary regularity assumptions we want to show that the problem of maximising utility from terminal wealth allows for a primal and a dual optimiser, related via the usual first moment conditions, and that the dual optimiser can be interpreted as a {\it shadow price} $\widehat{S}=\frac{\widehat{Y}^1}{\widehat{Y}^0}.$ To do so we have to interpret the notion of a shadow price $\widehat{S}$ in a rather general sense. In particular, it will turn out that $\widehat{S}$ may fail to be c\`adl\`ag (right continuous with left limits) so that we are forced to leave the classical framework of semimartingale theory.

To motivate the new phenomena arising in the present framework of general c\`adl\`ag price processes $S$, we indicate the ideas of two illuminating examples presented in Section  \ref{sec:ex} below. There the c\`adl\`ag stock price process $S=(S_t)_{0 \leq t \leq 1}$ is defined in such a way that it has a jump happening at a predictable stopping time $\tau$, say $\tau = \frac{1}{2}$. You may interpret $\tau$, e.g., as the time of a (previously announced) speech of the chair person of the ECB. Consider the log-optimal investor holding $(\hvp^0_t)_{0 \leq t \leq 1}$ units of cash and $(\hvp^1_t)_{0 \leq t \leq 1}$ units of the stock $S$. The process $S$ in Example \ref{ex1:pp} is designed in such a way that the holdings in stock $\hvp^1_t$ are increasing, for $0 < t \leq \frac{1}{2}$. The reason is that the stock price $S$ is sufficiently favourable for the investor during this period. If there is a shadow price $\widehat{S}$, this process must therefore satisfy $\widehat{S}_t=S_t$ for $0 \leq t < \frac{1}{2}$. Indeed, it is the basic feature of a shadow price that $\widehat{S}_t = S_t$ holds true when the optimising agent {\it buys} stock, while $\widehat{S}_t=(1-\lambda)S_t$ holds true when she {\it sells} stock.

At time $\tau=\frac{1}{2}$ it may happen that the news revealed during the speech are sufficiently negative to cause the agent to immediately {\it sell} stock, so that a shadow price process $\widehat{S}$ has to satisfy $\widehat{S}_{\frac{1}{2}}=(1-\lambda)S_{\frac{1}{2}},$ on a set of positive measure. Immediately, after time $\tau=\frac{1}{2}$ the situation of Example \ref{ex1:pp} quickly improves again for the log-optimising agent so that $\hvp^1_t$ {\it increases} again for $t > \frac{1}{2}$, implying that $\widehat{S}_t = S_t$, for $ t > \frac{1}{2}.$ 

The bottom line is that a shadow price $\widehat{S}$, if it exists in this example, {\it must have} a left as well as a right jump at time $t=\frac{1}{2}$ with positive probability. In particular $\widehat{S}$ cannot be given by the quotient $\frac{\widehat{Y}^1}{\widehat{Y}^0}$ of two local martingales $(\widehat{Y}^0, \widehat{Y}^1)$, as local martingales are c\`adl\`ag. In fact, $\widehat{S}$ cannot be a semimartingale.

Here is the way out of this difficulty. There is the classical notion of an {\it optional strong supermartingale} introduced by Mertens \cite{M72}, which allows for processes which are only optional and may very well have non-trivial left as well as right jumps. It turns out that this notion is tailor-made to replace the usual notion of a c\`adl\`ag supermartingale in the present situation and allows us to establish the existence of a dual optimiser $(\widehat{Y}^0, \widehat{Y}^1)$ in the class of optional strong supermartingales. In particular, it yields a candidate shadow price $\widehat{S}$ defined via $\widehat{S}=\frac{\widehat{Y}^1}{\widehat{Y}^0}$ as the quotient of two optional strong supermartingales $\widehat{Y}^0$ and $\widehat{Y}^1$. 

Actually, the phenomenon revealed by Example \ref{ex1:pp} is not yet the end of the story. In Example \ref{Ex2:S} we construct a variant displaying an even more delicate issue. Fixing again $\tau=\frac{1}{2},$ this example is designed in such a way that, with positive probability, the optimal strategy $\hvp$ sells stock at all times $t < \frac{1}{2}$ and also sells stock at all times $t \geq \frac{1}{2}$. Just ``immediately before'' time $t = \frac{1}{2}$, which is mathematically described by considering the left limit $S_{\frac{1}{2}-}$, it {\it buys} stock. Therefore a shadow price $\widehat{S}$, provided it exists, would have to satisfy $\widehat{S}_t=(1-\lambda) S_t,$ for $t < \frac{1}{2}$ as well as for $t \geq \frac{1}{2}$, while for $t=\frac{1}{2}$ we have $\widehat{S}_{t-}=S_{t-}$. Such a process $\widehat{S}$ cannot exist as the above properties do not make sense. The way out of this difficulty is to pass to {\it two} ``sandwiched'' processes $(\widehat{S}^p, \widehat{S})$ where $\widehat{S}$ is a quotient of optional strong supermartingales $(\widehat{Y}^0, \widehat{Y}^1)$ as above, while $\widehat{S}^p$ is a quotient of two {\it predictable strong supermartingales} $(\widehat{Y}^{0,p}, \widehat{Y}^{1,p})$, another classical notion from the general theory of stochastic processes \cite{CG79}. The process $\widehat{S}^p$ pertains to the left limits of $S$ and describes the buying or selling of the agent ``immediately before'' predictable stopping times. This turns out to be the final step of the complications. Using the notion of a {\it ``sandwiched shadow price process''} $\widehat{\mathcal{S}}:=(\widehat{S}^p, \widehat{S})$ as above we are able to characterise the dual optimiser as a shadow price and to prove positive results.  

Here is a verbal description of Theorem \ref{mt3} which is one of the main positive results of this paper. Under general hypotheses on an $\mathbb{R}_+$-valued c\`adl\`ag price processes $S=(S_t)_{0 \leq t \leq T}$, transaction costs $\lambda \in(0, 1)$, and a utility function $U$, there is a primal optimiser $\hvp=(\hvp^0_t, \hvp^1_t)_{0 \leq t \leq T}$ and a shadow price $\widehat{\mathcal{S}}=(\widehat{S}^p, \widehat{S})$ taking values in the bid-ask spread $[(1-\lambda)S, S]$ in the above ``sandwiched'' sense satisfying the following properties:
a competing strategy $\varphi=(\varphi^0_t, \varphi^1_t)_{0 \leq t \leq T}$ which  is allowed to trade {\it without transaction costs} at prices defined by $\widehat{\mathcal{S}}$, while remaining solvent with respect to prices defined by $S$ under transaction costs $\lambda$, cannot do better than $\hvp$ with respect to expected utility. 

Summing up our four main contributions are:

\bi
\item[{\bf1)}] We show that the solution $\widehat{Y}=(\widehat{Y}^0, \widehat{Y}^1)$ to the dual problem is attained as an optional strong supermartingale deflator.
\item[{\bf2)}] We explain how to extend the candidate shadow price $\widehat{S}:=\frac{\widehat{Y}^1}{\widehat{Y}^0}$ to a sandwiched shadow price $\widehat{\mathcal{S}}=(\widehat{S}^p, \widehat{S})$ that allows to obtain the optional strategy $\hvp=(\hvp^0, \hvp^1)$ under transaction costs for $S$ by frictionless trading for $\widehat{\mathcal{S}}.$
\item[{\bf3)}] We clarify in which sense the primal optimiser $\hvp=(\hvp^0, \hvp^1)$ for $S$ under transaction costs is also optimal for $\widehat{\mathcal{S}}$ without transaction costs.
\item[{\bf4)}] We provide examples that illustrate that a shadow price has to be of this generalised form and a detailed analysis that exemplifies how and why these new phenomena arise.
\ei

The remainder of the article is organised as follows. We introduce our setting and formulate the problem in Section \ref{sec2}. This leads to our main results that are stated and explained in Section \ref{sec3}. For better readability, the proofs are deferred to Appendix \ref{App:B}. Section \ref{sec:ex} contains the two examples that illustrate that a shadow price has to be of our generalised form. A more detailed analysis of the examples is given in Appendix \ref{App:A}.

\section{Formulation of the problem}\label{sec2}
We consider a financial market consisting of one riskless asset and one risky asset. The riskless asset has constant price $1$. Trading in the risky asset incurs proportional transaction costs of size $\lambda \in (0,1).$ This means that one has to pay a higher ask price $S_t$ when buying risky shares but only receives a lower bid price $(1-\lambda)S_t$ when selling them. The price of the risky asset is given by a strictly positive c\`adl\`ag adapted stochastic process $S=(S_t)_{0 \leq t \leq T}$ on some underlying filtered probability space $\big(\Om, \mathcal{F}, (\mathcal{F}_t)_{0 \leq t \leq T}, P\big)$ satisfying the usual assumptions of right continuity and completeness. As usual equalities and inequalities between random variables hold up to $P$-nullsets and between stochastic processes up to $P$-evanescent sets. 

\emph{Trading strategies} are modelled by $\R^2$-valued, predictable processes $\vp=(\vp^0_t,\vp^1_t)_{0\leq t\leq T}$ of finite variation, where $\vp^0_{t}$ and $\vp^1_{t}$ describe the holdings in the riskless and the risky asset, respectively, after rebalancing the portfolio at time $t$. For any process $\psi=(\psi_t)_{0\leq t\leq T}$ of finite variation we denote by $\psi=\psi_0+\psi^{\uparrow}-\psi^{\downarrow}$ its Jordan-Hahn decomposition into two non-decreasing processes $\psi^{\uparrow}$ and $\psi^{\downarrow}$ both null at zero. The total variation $\Var_t(\psi)$ of $\psi$ on $[0,t]$ is then given by $\Var_t(\psi)=\psi^{\uparrow}_t+\psi^{\downarrow}_t$. Note that, any process $\psi$ of finite variation is in particular l\`adl\`ag (with right and left limits). For any l\`adl\`ag process $X=(X_t)_{0\leq t\leq T}$ we denote by $X^c$ its continuous part given by
$$X^c_t:=X_t-\sum_{s<t} \Delta_+ X_s -  \sum_{s\leq t} \Delta X_s,$$
where $\Delta_+ X_t:=X_{t+}-X_t$ are its right and $\Delta X_t:=X_t-X_{t-}$ its left jumps. As explained in Section 7 of \cite{CS13} in more detail, we can define for a finite variation process $\psi=(\psi_t)_{0\leq t\leq T}$ and a l\`adl\`ag process $X=(X_t)_{0\leq t\leq T}$ the integrals
\begin{align}\label{def:SI:1}
\int^t_0 X_u(\om) d\psi_u(\om):=& \int^t_0 X_u(\om) d\psi^c_u(\om) + \sum_{0 <u \leq t} X_{u-}(\om) \Delta\psi_u(\om) + \sum_{0 \leq u < t} X_u(\om) \Delta_+\psi_u(\om)
\end{align}
and
\begin{align}
\label{def:SI:2}
\psi \sint X_t:= \int^t_0 \psi_u (\om) dX_u (\om):={}& \int^t_0 \psi^c_u (\om) dX_u (\om)+ \sum_{0 <u \leq t} \Delta\psi_u(\om) \big(X_t(\om) - X_{u-}(\om)\big)\notag\\
& + \sum_{0 \leq u < t} \Delta_+\psi_u(\om) \big(X_t(\om) - X_{u}(\om)\big)
\end{align}
pathwise by using Riemann-Stieltjes integrals such that the integration by parts formula
\be
\psi_t(\om)X_t(\om)  =\psi_0(\om) X_0 (\om)+ \int^t_0 \psi_u(\om) dX_u(\om) + \int^t_0 X_u (\om)d\psi_u(\om)\label{SI:IP}
\ee
holds true. Note that, if $X=(X_t)_{0\leq t\leq T}$ is a semimartingale and $\psi=(\psi_t)_{0\leq t\leq T}$ is in addition predictable, the pathwise integral \eqref{def:SI:2} coincides with the classical stochastic integral. 

A strategy $\vp=(\vp^0_t,\vp^1_t)_{0\leq t\leq T}$ is called \emph{self-financing under transaction costs $\lambda$}, if
\begin{align}\label{sfc}
\int^t_s d\varphi^0_u \leq - \int^t_s S_u d\varphi^{1,\uparrow}_u + \int^t_s(1-\lambda)S_u d\varphi^{1,\downarrow}_u
\end{align}
for all $0 \leq s < t \leq T$, where the integrals are defined via \eqref{def:SI:1}. The self-financing condition \eqref{sfc} then states that purchases and sales of the risky asset are accounted for in the riskless position:
\begin{align}
d\varphi^{0,c}_t&\leq-S_td\varphi^{1,\uparrow,c}_t+(1-\lambda)S_td\varphi^{1,\downarrow,c}_t, &0 \leq t \leq T,\label{eq:sf2.1}\\
\Delta\varphi^0_t&\leq-S_{t-}\Delta\varphi^{1,\uparrow}_t +(1-\lambda)S_{t-}\Delta\varphi^{1,\downarrow}_t, &0 \leq t \leq T,\label{eq:sf2.2}\\
\Delta_+\varphi^0_t&\leq-S_t\Delta_+\varphi^{1,\uparrow}_t +(1-\lambda)S_t\Delta_+\varphi^{1,\downarrow}_t, & 0 \leq t \leq T.\label{eq:sf2.3}
\end{align}
A self-financing strategy $\vp$ is \emph{admissible} under transaction costs $\lambda$, if its \emph{liquidation value} $V^{liq}(\vp)$ verifies
\begin{align}\label{eq:lv}
V_t^{liq}(\vp):={}&\vp^0_t+(\vp^1_t)^+(1-\lambda) S_t-(\vp^1_t)^-S_t\geq 0 
\end{align}
for all $t\in[0,T]$.

For $x>0$, we denote by $\cA(x)$ the set of all self-financing, admissible trading strategies under transaction costs $\lambda$ starting with initial endowment $(\vp^0_0,\vp^1_0)=(x,0)$.

Applying integration by parts to \eqref{eq:lv} yields that, for $\varphi \in \cA(x)$, the liquidation value $V^{liq}_t(\vp)$ is given by the initial value of the position $\varphi^0_0 = x$ plus the gains from trading $\int^t_0 \varphi^1_s dS_s$ minus the transaction costs for rebalancing the portfolio $\lambda \int^t_0 S_s d \varphi_s^{1,\downarrow}$ minus the costs $\lambda S_t (\varphi^1_t)^+$ for liquidating the position at time $t$, i.e.~
\begin{equation}\label{eq:gtc}
V^{liq}_t(\vp) = \varphi^0_0 + \int^t_0 \varphi^1_s dS_s- \lambda \int^t_0 S_s d\varphi_s^{1,\downarrow} - \lambda S_t (\varphi^1_t)^+. 
\end{equation}

We consider an investor whose preferences are modelled by a standard utility function\footnote{This means a strictly concave, increasing and continuously differentiable function satisfying the Inada conditions $U'(0)= \lim_{x \searrow 0} U'(x)= \infty$ and $U'(\infty)= \lim_{x \nearrow \infty} U'(x)=0.$} $U:(0, \infty) \to \mathbb{R}$ that tries to maximise expected utility of terminal wealth. Her basic problem is to find the optimal trading strategy $\hvp = (\hvp^0, \hvp^1)$ to
\begin{equation}\label{PP1}
E[U(V^{liq}_T(\vp))] \to \max!, \quad \varphi \in \cA(x). 
\end{equation}
Alternatively, \eqref{PP1} can be formulated as the problem for random variables to find the optimal payoffs $\widehat{g}$ to
\begin{equation}\label{PP2}
E[U(g)] \to \max!, \quad g \in \cC(x), 
\end{equation}
where 
$$\cC(x) =\{V^{liq}_T(\vp)~|~\varphi \in \cA(x)\} \subseteq L^0_+(P)$$
denotes the set of all attainable payoffs under transaction costs.

As explained in Remark 4.2 in \cite{CS06}, we can always assume without loss of generality that the price cannot jump at the terminal time $T$, while the investor can still liquidate her position in the risky asset. This implies that we can assume without loss of generality that $\vp^1_T=0$ and therefore have
$$\cC(x)=\{\vp^0_T~|~\varphi=(\vp^0,\vp^1) \in \cA(x)\}\subseteq L^0_+(P).$$

Following the seminal paper \cite{CK96} by Cvitanic and Karatzas we investigate \eqref{PP1} by duality. For this, we consider the notion of a $\lambda$-consistent price system. A \emph{$\lambda$-consistent price system} is a pair of processes $Z=(Z^0_t, Z^1_t)_{0 \leq t \leq T}$ consisting of the density process $Z^0=(Z^0_t)_{0 \leq t \leq T}$ of an equivalent local martingale measure $Q\sim P$ for a price process $\widetilde{S}=(\widetilde{S}_t)_{0 \leq t \leq T}$ evolving in the bid-ask spread $[(1-\lambda)S,S]$ and the product $Z^1=Z^0\widetilde{S}.$ Requiring that $\widetilde{S}$ is a local martingale under $Q$ is tantamount to the product $Z^1=Z^0\widetilde{S}$ being a local martingale. We say that $S$ satisfies the condition $(CPS^\lambda)$, if it admits a $\lambda$-consistent price system, and denote the set of all $\lambda$-consistent price systems by $\mathcal{Z}.$ As has been initiated by Jouini and Kallal \cite{JK95}, these processes play a similar role under transaction costs as equivalent local martingale measures in the frictionless theory. Similarly as in the frictionless case (see \cite{KK07} and \cite{K10}) it is sufficient for the existence of an optimal strategy for \eqref{PP1} under transaction costs to assume the existence of $\lambda'$-consistent price systems locally; see \cite{BY13}. We therefore say that $S$ admits locally a $\lambda$-consistent price system or shorter satisfies the condition $(CPS^\lambda)$ locally, if there exists a strictly positive stochastic process $Z=(Z^0, Z^1)$ and a localising sequence $(\tau_n)^\infty_{n=1}$ of stopping times such that $Z^{\tau_n}$ is a $\lambda$-consistent price system for the stopped process $S^{\tau_n}$ for each $n \in \mathbb{N}$. We denote the set of all such process $Z$ by $\cZ_{loc}$.

To motivate the dual problem, let $Z=(Z^0,Z^1)$ be any $\lambda$-consistent price system or, more generally, any process in $\cZ_{loc}$. Then trading for the price $\widetilde{S}=\frac{Z^1}{Z^0}$ without transaction costs allows to buy and sell at possibly more favourable prices than applying the price $S$ under transaction costs. Therefore any attainable payoff in the market with transaction costs can be dominated by trading at the price $\widetilde{S}$ without transaction costs and hence
\be
u(x):=\sup_{\varphi \in \mathcal{A}(x)} E[U(V^{liq}_T(\vp))]\leq \sup_{\vp \in \mathcal{A}(x;\widetilde{S})}E[U(x+\vp^1\sint \tS_T)]=:u(x;\tS).\label{sec1:eq1}
\ee

Here $\mathcal{A}(x;\widetilde{S})$ denotes the set of all self-financing and admissible trading strategies $\vp=(\vp^0_t,\vp^1_t)_{0\leq t\leq T}$ for the price process $\widetilde{S}=(\tS_t)_{0\leq t\leq T}$ without transaction costs ($\lambda=0$) in the classical sense, i.e.~that $\vp^1=(\vp^1_t)_{0\leq t\leq T}$ is an $\tS$-integrable predictable process such that $X+\varphi^1\sint \widetilde{S}_t \geq 0$ for all $t \in [0,T]$ and $\vp^0=(\vp^0_t)_{0\leq t\leq T}$ is defined via $\vp^0_t=x+\int_0^t\vp^1_ud\tS_u-\vp^1_t\tS_t$, for $t\in[0,T]$. Note that $\cA(x)\subseteq\cA(x;\tS)$.

As usual we denote by
\be
V(y):= \sup_{x > 0} \{U(x) - xy\}, \quad y > 0,\label{eq:defV}
\ee
the Legendre transform of $-U(-x).$ 

By definition of $\cZ_{loc}$ we have that $Z^0\widetilde{S}=Z^1$ is a local martingale. Therefore $Z^0$ is an \emph{equivalent local martingale deflator} for the price process $\widetilde{S}=(\widetilde{S}_t)_{0 \leq t \leq T}$ in the language of Kardaras \cite{K10} and 
\begin{equation*}
Z^0 \vp^0 + Z^1 \vp^1 = Z^0(\vp^0 + \vp^1 \widetilde{S})=Z^0 (x + \vp^1 \sint \widetilde{S})
\end{equation*}
is a non-negative local martingale and hence a supermartingale for all $\vp \in \cA (x;\widetilde{S}).$

Combining the supermartingale property with the Fenchel inequality we obtain
\begin{align*}
u(x;\widetilde{S}) &= \sup_{\varphi \in \cA (x;\widetilde{S})} E[U(x+\varphi^1 \sint \widetilde{S}_T)] \leq E[V(yZ^0_T) + yZ^0_T (x + \varphi^1 \sint \widetilde{S}_T)]\leq E[V(y Z^0_T)] + xy.
\end{align*}
As $u(x) \leq u(x;\widetilde{S})$ by \eqref{sec1:eq1}, the above inequality implies that
$$u(x) \leq E[V(y Z^0_T)]$$
for all $Z=(Z^0, Z^1) \in \cZ_{loc}$ and $y >0$ and therefore motivates to consider
\begin{align}
E[V(y Z^0_T)] \to \min!, \quad Z=(Z^0, Z^1) \in \cZ_{loc},\label{DP1}
\end{align}
as dual problem. Again problem \eqref{DP1} can be alternatively formulated as a problem over a set of random variables
\begin{equation}\label{DP1.2}
E[V(h)] \to \min!, \quad h \in D(y),
\end{equation}
where 
\begin{equation}\label{DP1.3}
D(y)=\{yZ^0_T~|~Z =(Z^0, Z^1)  \in \cZ_{loc}\}= yD(1)
\end{equation}
for $y>0$ and $D(1)=D$.

If the solution $\widehat{Z}=(\widehat{Z}^0, \widehat{Z}^1) \in \cZ_{loc}$ to problem \eqref{DP1} exists, the ratio
$$\widehat{S}_t := \frac{\widehat{Z}^1_t} {\widehat{Z}^0_t}, \quad t \in [0,T],$$
is a \emph{shadow price} in the sense of the subsequent definition (compare \cite{KMK10, KMK11}). This result seems to be folklore going back to the works of Cvintanic and Karatzas~\cite{CK96} and Loewenstein\cite{L00}, but we did not find a reference. We state and prove it in Proposition \ref{lem:martingale} below.
\begin{defi}\label{def:shadow}
A semimartingale $\tS=(\tS_t)_{0\leq t\leq T}$ is called a \emph{shadow price}, if
\bi
\item[\bf{1)}] $\tS=(\tS_t)_{0\leq t\leq T}$ takes values in the bid-ask spread $[(1-\lambda)S,S]$.
\item[\bf{2)}] The solution $\tvp=(\tvp^0,\tvp^1)$ to the corresponding frictionless utility maximisation problem
\be
E[U(x+\vp^1\sint \tS_{T})]\to\max!,\qquad (\varphi^0,\varphi^1)\in\cA(x; \tS),\label{p1}
\ee
exists and coincides with the solution $\hvp=(\hvp^0,\hvp^1)$ to \eqref{PP1} under transaction costs.
\ei
\end{defi}
Note that a shadow price $\tS=(\tS_t)_{0\leq t\leq T}$ depends on the process $S$, the investor's utility function, and on her initial endowment.

The intuition behind the concept of a shadow price is the following. If a shadow price $\widetilde{S}$ exists, then an optimal strategy $\widetilde{\varphi}=(\widetilde{\varphi}^0, \widetilde{\varphi}^1)$ for the frictionless utility maximisation problem \eqref{p1} can also be realised in the market with transaction costs in the sense spelled out in \eqref{2.10} below. As the expected utility for $\widetilde{S}$ without transaction costs is by \eqref{sec1:eq1} a priori higher than that of any other strategy under transaction costs, it is -- a fortiori -- also an optimal strategy under transaction costs. In this sense the price process $\widetilde{S}$  is a least favourable frictionless market evolving in the bid-ask spread. 
The existence of a shadow price $\tS$ implies in particular that the optimal strategy $\hvp=(\hvp^0,\hvp^1)$ under transaction costs only trades, if $\widetilde{S}$ is at the bid or ask price, i.e. 
$$\{d\hvp^1>0\}\subseteq\{\tS=S\}  \quad \mbox{and} \quad \{d\hvp^1<0\}\subseteq\{\tS=(1-\lambda)S\}$$
in the sense that
\begin{align}
\{d\hvp^{1,c}>0\}&\subseteq \{\tS=S\}, & \{d\hvp^{1,c}<0\}&\subseteq \{\tS=(1-\lambda)S\},\notag \\
\{\Delta \hvp^{1}>0\}&\subseteq \{\tS_-=S_-\}, &  \{\Delta \hvp^{1}<0\}&\subseteq \{\tS_-=(1-\lambda)S_-\}, \notag \\
\{\Delta_+ \hvp^{1}>0\}&\subseteq \{\tS=S\}, &  \{\Delta_+ \hvp^{1}<0\}&\subseteq \{\tS=(1-\lambda)S\}.\label{2.10}
\end{align}

As the counter-examples in \cite{BCKMK11} and \cite{CMKS14} illustrate and we shall show in Section \ref{sec:ex} below, shadow prices fail to exit in general, at least in the rather narrow sense of Def \ref{def:shadow}. The reason for this is that, similarly to the frictionless case \cite{KS99}, the solution $\widehat{h}$ to \eqref{DP1.2} is in general only attained as a $P$-a.s.~limit
\begin{align}\label{p7}
\widehat{h} = y\lim_{n \to \infty} Z^{0,n}_T
\end{align}
of a minimising sequence $(Z^n)_{n=1}^\infty$ of local consistent price systems $Z^n=(Z^{0,n}, Z^{1,n})$.

To ensure the existence of an optimiser, one has therefore to work with relaxed versions of the dual problems \eqref{DP1} and \eqref{DP2}. For the dual problem \eqref{DP2} on the level of random variables it is clear that one has to consider
\begin{equation}\label{DP2}
E[V(h)] \to \min!, \quad h \in\overline{\sol\big(D(y)\big)}, 
\end{equation}
where
$$\overline{\sol\big(D(y)\big)} = \{yh \in L^0_+ (P)~| ~\exists Z^n =(Z^{0,n}, Z^{1,n}) \in \cZ_{loc} \ \mbox{such that}\ h \leq \lim_{n \to \infty} Z^{0,n}_T\}$$
is the closed, convex, solid hull of $D(y)$ defined in \eqref{DP1.3} for $y>0.$ 

As sets $\cC(x)$ and $\overline{\sol\big(D(y)\big)}$ are polar to each other in $L^0_+(P)$ (see Lemma \ref{lpolar}), the abstract versions (Theorems 3.1 and 3.2) of the main results of \cite{KS99} carry over verbatim to the present setting under transaction costs. This has already been observed in \cite{CW01, CMKS14, BY13} and gives \emph{static} duality results in the sense that they provide duality relations between the solutions to the problems \eqref{PP2} and \eqref{DP2} which are problems for random variables rather than stochastic processes. See also \cite{DPT01, CO11} for static results for more general multivariate utility functions. 
However, in the context of dynamic trading,  this is not yet completely satisfactory. Here one would not only like to know the optimal terminal positions but also how to dynamically trade to actually attain those. We therefore ask, if we can extend these static results to dynamic ones in the same spirit as Theorems 2.1 and 2.2 of \cite{KS99}. In particular, we address the following questions:
\bi
\item[\bf{1)}] Is there a ``reasonable'' stochastic process $\hY=(\hY_t^0,\hY_t^1)_{0\le t\le T}$ such that $\hY^0_T=\hh,$ where $\hh$ is a dual optimiser as in \eqref{p7}?
\item[\bf{2)}] Do we have $\{d\hvp^1>0\}\subseteq\{\hS=S\}$ and $\{d\hvp^1<0\}\subseteq\{\hS=(1-\lambda)S\}$ as in \eqref{2.10} for $\hS=\frac{\hY^1}{\hY^0}$?\label{questions}
\item[\bf{3)}] In which sense is $\hvp=(\hvp^0, \hvp^1)$ optimal for $\widehat{S}$?
\ei
\section{Main results}\label{sec3}
In this section, we consider the three questions above that lead to our main results. For better readability the proofs are deferred to Appendix \ref{App:B}.

Let us begin with the first question. Similarly as in the frictionless duality \cite{KS99}, we consider \emph{supermartingale deflators} as dual variables. These are non-negative (not necessarily c\`adl\`ag) supermartingales $Y=(Y^0, Y^1) \geq 0$ such that $\widetilde{S}:= \frac{Y^1}{Y^0}$ is valued in the bid-ask spread $[(1-\lambda)S,S]$ and that turn all trading strategies $\varphi=(\varphi^0, \varphi^1) \in \cA(1)$ into supermartingales, i.e.
\begin{equation}\label{df}
Y^0\varphi^0 + Y^1 \varphi^1 = Y^0 (\varphi^0 + \varphi^1 \widetilde{S})
\end{equation}
is a supermartingale for all $\varphi \in \cA (1).$
Recall that in the frictionless case \cite{KS99}, the solution to the dual problem for an arbitrary semimartingale price process $\widetilde{S}=(\widetilde{S}_t)_{0 \leq t \leq T}$ is attained in the set of (one-dimensional) \emph{c\`adl\`ag} supermartingale deflators
\begin{multline*}
Y(y;\widetilde{S})=\{Y=(Y_t)_{0 \leq t \leq T}\geq 0~|~Y_0=y \quad \mbox{and} \quad Y(\varphi^0 + \varphi^1\widetilde{S})=Y(1 + \varphi^1\sint\widetilde{S})\\
\text{is a c\`adl\`ag supermartingale for all $\varphi \in \cA(1;\widetilde{S})$}\}.
\end{multline*}
The reason for this is that by the frictionless self-financing condition the value $\varphi^0 + \varphi^1\widetilde{S}$ of the position is equal to the gains from trading given by $x + \varphi^1 \sint \widetilde{S}.$ As the stochastic integral $x + \varphi^1 \sint \widetilde{S}$ is right-continuous, the optimal supermartingale deflator to the dual problem can be obtained as the c\`adl\`ag Fatou limit of a minimising sequence of equivalent local martingale or supermartingale deflators; see Lemma 4.2 and Proposition 3.1 in \cite{KS99}. This means as the c\`adl\`ag modification of the $P$-a.s.~pointwise limits along the rationals that are obtained by combining Koml\'os' lemma with a diagonalisation procedure.
 
We show in \cite{CS14} that the dual optimiser is attained as Fatou limit under transaction costs as well, if the price process $S$ is \emph{continuous}. As the price process does not jump, it doesn't matter, if one is trading immediately before, or just at a given time and one can model trading strategies by \emph{c\`adl\`ag} adapted finite variation processes. By \eqref{df} the right-continuity of $(\varphi^0, \varphi^1)$ then allows to pass the supermartingale property onto to the Fatou limit as in the frictionless case.

For c\`adl\`ag price processes $S=(S_t)_{0 \leq t \leq T}$ under transactions costs $\lambda$, however, one has to use predictable finite variation strategies $\varphi=(\vp^0_t,\vp^1_t)_{0 \leq t \leq T}$ that can have left and right jumps to model trading strategies as motivated in the introduction. This is unavoidable in order to obtain that the set $\cC(x)$ of attainable payoffs under transaction costs is closed in $L^0_+(P)$ (see Theorem 3.5 in \cite{CS06} or Theorem 3.4 in \cite{S14}). As we have to optimise simultaneously over $Y^0$ and $Y^1$ to obtain the optimal supermartingale deflator, we need a different limit than the Fatou limit in \eqref{df} to remain in the class of supermartingale deflators. This limit also needs to ensure the convergence of a minimising sequence $Z^n=(Z^{0,n}_t, Z^{1,n}_t)_{0 \leq t \leq T}$ of consistent price systems at the jumps of the trading strategies. It turns out that the convergence in probability at all finite stopping times is the right topology to work with (compare \cite{CS13}). The limit of the non-negative local martingales $Z^n=(Z^{0,n}_t, Z^{1,n}_t)_{0 \leq t \leq T}$ for this convergence is then an \emph{optional strong supermartingale}.
\begin{defi}
A real-valued stochastic process $X=(X_t)_{0\leq t\leq T}$ is called an \emph{optional strong supermartingale}, if
\bi
\item[\bf{1)}] $X$ is optional.
\item[\bf{2)}] $X_\tau$ is integrable for every $[0,T]$-valued stopping time $\tau$.
\item[\bf{3)}] For all stopping times $\sigma$ and $\tau$ with $0\leq\sigma\leq\tau\leq T$ we have
$$X_\sigma\geq E[X_\tau|\cF_\sigma].$$
\ei
\end{defi}
These processes have been introduced by Mertens \cite{M72} as a generalisation of the notion of a c\`adl\`ag supermartingale. Like the Doob-Meyer decomposition in the c\`adl\`ag case every optional strong supermartingale admits a unique decomposition
\begin{equation}\label{MD}
X=M-A
\end{equation}
called the \emph{Mertens decomposition} into a c\`adl\`ag local martingale $M=(M_t)_{0\leq t\leq T}$ and a non-decreasing and  hence l\`adl\`ag (but in general neither c\`adl\`ag nor c\`agl\`ad) predictable process $A=(A_t)_{0\leq t\leq T}$. The existence of the decomposition \eqref{MD} implies in particular that every optional strong supermartingale is l\`adl\`ag.

As dual variables we then consider the set of \emph{optional strong supermartingale deflators}
\begin{multline}
\cB(y)=\big\{(Y^0,Y^1) \geq 0\ \big|\ Y^0_0=y,\, \tS=\tfrac{Y^1}{Y^0} \in[(1-\lambda)S,S]\text{ and }Y^0(\vp^0+  \vp^1\tS)=Y^0\vp^0+Y^1\vp^1  \\
\text{is a non-negative optional strong supermartingale for all $(\vp^0,\vp^1)\in\cA(1)$}\big\}\label{def:B}
\end{multline}
and, accordingly, 
$$\cD(y)=\{Y^0_T~|~\text{$(Y^0,Y^1)\in\cB(y)$}\}  \quad \mbox{for $y>0$}.$$
We will show in Lemma \ref{lpolar} below that we have $\cD(y)=\overline{\sol\big(D(y)\big)}$ with this definition.

Using a version of Koml\'os' lemma (see Theorem 2.7 in \cite{CS13}) pertaining to optional strong supermartingales then allows us to establish our first main result. It is in the well-known spirit of the duality theory of portfolio optimisation as initiated by \cite{Pliska,KLSX91,HP91, KS99}.
\bt\label{mainthm}
Suppose that the adapted c\`adl\`ag process $S$ admits locally a $\lambda'$-consistent price system for all $\lambda'\in(0,\lambda)$, the asymptotic elasticity of $U$ is strictly less than one, i.e., $AE(U):=\limsup\limits_{x\to\infty}\frac{xU'(x)}{U(x)}<1$, and the maximal expected utility is finite, $u(x):=\sup_{g\in\cC(x)}E[U(g)]<\infty$, for some $x\in(0,\infty)$. Then:
\bi
\item[{\bf 1)}] The primal value function $u$ and the dual value function
$$v(y):=\inf_{h\in\cD(y)}E[V(h)]$$
are conjugate, i.e.,
\begin{eqnarray*}u(x)=\inf_{y>0}\{v(y)+xy\},\qquad v(y)=\sup_{x>0}\{u(x)-xy\},
\end{eqnarray*}
and continuously differentiable on $(0,\infty)$. The functions $u$ and $-v$ are strictly concave, strictly increasing, and satisfy the Inada conditions
$$\text{$\Lim_{x\to0}u'(x)=\infty,\qquad\Lim_{y\to\infty}v'(y)=0,\qquad\Lim_{x\to\infty}u'(x)=0,\qquad\Lim_{y\to0}v'(y)=-\infty$}.$$
\item[{\bf 2)}] For all $x,y>0$, the solutions $\widehat g (x)\in\cC(x)$ and $\widehat h(y)\in\cD(y)$ to the primal problem
\begin{equation*}
\textstyle
E\left[U(g)\right]\to\max!, \qquad{g\in\cC(x)},
\end{equation*}
and the dual problem
\begin{equation}
\textstyle
E\left[V(h)\right]\to\min!\label{DP3}, \qquad{h\in\cD(y)},
\end{equation}
exist, are unique, and there are
$\big(\hvp^0(x),\hvp^1(x)\big)\in\cA(x)$ and $\big(\widehat{Y}^0(y),\widehat{Y}^1(y)\big)\in\cB(y)$ such that
\be
\text{$V_{T}^{liq}\big(\hvp(x)\big)=\widehat g(x)\qquad$ and $\qquad\widehat{Y}^0_T(y)=\widehat h(y)$.}\label{martcond}
\ee
\item[{\bf 3)}] For all $x>0$, let $\widehat y (x)=u'(x)>0$ which is the unique solution to
$$
v(y)+xy\to\min!,\qquad y>0.
$$
Then, $\widehat g (x)$ and $\widehat h \big(\widehat y(x)\big)$ are given by $(U')^{-1}\big(\widehat h \big(\widehat y(x)\big)\big)$ and $U'\big(\widehat g (x)\big)$, respectively, and we have that $E\big[\widehat g(x)\widehat h\big(\widehat y (x)\big)\big]=x\widehat y(x)$. In particular, the process
$$\widehat{Y}^0\big(\widehat y(x)\big)\hvp^0(x)+ \widehat{Y}^1\big(\widehat y(x)\big)\hvp^1(x)=\Big(\widehat{Y}^0_t\big(\widehat y(x)\big)\hvp^0_t(x)+ \widehat{Y}^1_t\big(\widehat y(x)\big)\hvp_t^1(x)\Big)_{0\leq t\leq T}$$
is a c\`adl\`ag martingale for all $\big(\hvp^0(x),\hvp^1(x)\big)\in\cA(x)$ and $\big(\widehat{Y}^0\big(\widehat y(x)\big),\widehat{Y}^1\big(\widehat y(x)\big)\big)\in\cB\big(\widehat y (x)\big)$ satisfying \eqref{martcond} with $y=\widehat y (x)$.
\item[{\bf 4)}] Finally, we have
\be
v(y)=\Inf_{(Z^0,Z^1)\in\mathcal{Z}_{loc}}E[V(yZ^0_T)].\label{mt1:eq4}
\ee
\ei
\et
Before we continue, let us briefly comment -- for the specialists -- on the assumption that $S$ admits locally a $\lambda'$-consistent price system for \emph{all} $\lambda'\in(0,\lambda)$. We have to make this assumption, since we chose that $V^{liq}(\vp)\geq 0$ as admissibility condition; compare \cite{S13} and \cite{S14}. Without this assumption Bayraktar and Yu show that a primal optimiser still exists, if $S$ admits locally a $\lambda'$-consistent price system for \emph{some} $\lambda'\in(0,\lambda)$; see \cite[Theorem 5.1]{BY13}. However, then a modification of the example in \cite[Lemma 2.1]{S13} shows that the dual optimiser is only a supermartingale deflator in this case that can no longer be approximated by local consistent price systems. To resolve this issue, one can alternatively use (a local version of) the admissibility condition of Campi and Schachermayer \cite[Definition 2.7]{CS06} and say that a self-financing trading strategy $\vp=(\vp^0,\vp^1)$ is admissible, if $Z^0\vp^0+Z^1\vp^1$ is a non-negative supermartingale for all $Z=(Z^0,Z^1)\in\cZ_{loc}$. Then one could also replace the ``\emph{all}'' by a ``\emph{some}'' in the assumption.

In order to obtain a crisp theorem instead of getting lost in the details of the technicalities we therefore have chosen to use the (stronger) hypothesis pertaining to {\it all} $\lambda' \in (0, \lambda).$

Let us now turn to the second question raised at the end of the last section. Defining $\widehat{S}:= \frac{\widehat{Y}^1}{\widehat{Y}^0}$ the above theorem provides a price process evolving in the bid-ask spread and so the natural question is in which sense this can be interpreted as a shadow price. For example, we show in \cite{CS14} that for continuous processes $S=(S_t)_{0 \leq t \leq T}$ satisfying the condition $(NUPBR)$ of ``no unbounded profit with bounded risk'' the definition $\widehat{S}= \frac{\widehat{Y}^1}{\widehat{Y}^0}$ does yield a shadow price in the sense of Definition \ref{def:shadow}. However, in general, the counter-examples in \cite{BCKMK11, CMKS14, CS14} illustrate that the frictionless optimal strategy for $\widehat{S}$ to \eqref{p1} might do strictly better (with respect to expected utility of terminal wealth) than the optimal strategy under transaction costs and both strategies are different. While we show in Theorem 2.6 in \cite{CS14} that the dual optimiser is always a c\`adl\`ag supermartingale, if the underlying price process $S$ is continuous, we shall see in Example \ref{sec:ex1} below that it may indeed happen that the dual optimiser $\widehat{Y}=(\widehat{Y}^0, \widehat{Y}^1)$ as well as its ratio $\widehat{S}$ do \emph{not} have c\`adl\`ag trajectories and therefore fail to be semimartingales. Though we are not in the standard setting of stochastic integration we can still define the stochastic integral $\widehat{\varphi}^1 \sint \widehat{S}$ of a predictable finite variation process $\hvp^1=(\hvp^1_t)_{0 \leq t \leq T}$ with respect to the l\`adl\`ag process $\widehat{S}=(\widehat{S}_t)_{0 \leq t \leq T}$ by integration by parts; see \eqref{def:SI:1} and \eqref{def:SI:2}. This yields
\begin{align}\label{D4}
(\vp^1 \sint \widehat{S})_t &= \int^t_0 \vp_u^{1,c} d\widehat{S}_u+ \sum_{0 < u \leq t} \Delta \vp^1_u \big(\widehat{S}_t - \widehat{S}_{u-} \big)+ \sum_{0 \leq u < t} \Delta_+ \vp^1_u  \big( \widehat{S}_t - \widehat{S}_u\big).
\end{align}
The integral \eqref{D4} can still be interpreted as the gains from trading of the self-financing trading strategy $\hvp^1=(\hvp^1_t)_{0 \leq t \leq T}$ without transaction costs for the price process $\widehat{S}=(\widehat{S}_t)_{0 \leq t \leq T}.$ We may ask, whether $\widehat{S}$ is the frictionless price process for which the optimal trading strategy $\hvp=(\hvp^0, \hvp^1)$ under transaction costs trades in the sense of \eqref{2.10}.

It turns out that the left jumps $\Delta \hvp^1_u$ of the optimiser $\hvp^1$ need special care. The crux here is that, as shown in \eqref{D4}, the trades $\Delta\hvp^1_u$ are not carried out at the price $\widehat{S}_u$ but rather at its left limit $\widehat{S}_{u-}.$ As motivated in the introduction we need to consider a pair of processes $Y^p=(Y^{0, p}_t, Y^{1, p}_t)_{0 \leq t \leq T}$ and $Y=(Y^0_t, Y^1_t)_{0 \leq t \leq T}$ that correspond to the limit of the left limits $Z^n_-=(Z^{0,n}_-, Z^{1,n}_-)$ and the limit of the approximating consistent price systems $Z^n=(Z^{0,n}, Z^{1,n})$ themselves retrospectively. As we shall see in Example \ref{sec:ex2} below, the process $Y^p$ and $Y_-$ do not need to coincide so that we have that ``limit of left limits $\ne$ left limit of limits''.

Like the left limits $Z^n_-=(Z^{0,n}_-, Z^{1,n}_-)$ their limit $Y^p=(Y^{0,p}, Y^{1,p})$ is a predictable strong supermartingale.
\begin{defi}\label{def:pred}
A real-valued stochastic process $X=(X_t)_{0 \leq t \leq T}$ is called a \emph{predictable strong supermartingale}, if
\bi
\item[\bf{1)}] $X$ is predictable.
\item[\bf{2)}] $X_{\tau}$ is integrable for every $[0,T]$-valued \emph{predictable} stopping time $\tau$.
\item[\bf{3)}] For all \emph{predictable} stopping times $\sigma$ and $\tau$ with $0 \leq \sigma \leq \tau \leq T$ we have
$$X_\sigma \geq E[X_\tau | \cF_{\sigma-}].$$
\ei
\end{defi}
These processes have been introduced by Chung and Glover \cite{CG79} and we refer also to Appendix I of \cite{DM82} for more information on this concept.

We combine the two classical notions of predictable and optional strong supermartingales in the following concept.
\begin{defi}
A {\it sandwiched strong supermartingale} is a pair $\mathcal{X}=(X^p, X)$ such that $X^p$ (resp. $X$) is a predictable (resp. optional) strong supermartingale and such that 
\begin{align}\label{D5}
X_{\tau-} \geq X^p_\tau \geq E[X_\tau | \cF_{\tau-}],
\end{align}
for all predictable stopping times $\tau.$
\end{defi}
For example, starting from an optional strong supermartingale $X=(X_t)_{0 \leq t \leq T}$ we may define the process
\be
X^p_t:= X_{t-},\quad t\in[0,T],\label{D5.1}
\ee
to obtain a ``sandwiched'' strong supermartingale $\mathcal{X}=(X^p, X)$. If $X$ happens to be a local martingale, this choice is unique as we have equalities in \eqref{D5}. But in general there may be strict inequalities. This is best illustrated in the (trivial) deterministic case: if $X_t=f_t$ for a non-increasing function $f$, we may choose $X^p_t=f^p_t$, where $f^p_t$ is any function sandwiched between  $f_{t-}$ and $f_t.$

For a sandwiched strong supermartingale $\mathcal{X}=(X^p, X)$ and a predictable process $\psi$ of finite variation we may define a stochastic integral in {\it ``a sandwiched sense''} by
\begin{align}\label{D7}
(\psi \sint \mathcal{X}) &= \int^t_0 \psi^c_udX_u+ \sum_{0 \leq u <t} \Delta \psi_u(X_t - X^p_u)+ \sum_{0 < u \leq t} \Delta_+ \psi_u (X_t - X_u).
\end{align}
We note that \eqref{D7} differs from \eqref{D4} and \eqref{def:SI:2} only by replacing $X_-$ by $X^p$ and the two formulas are therefore consistent, as we can extend every optional strong supermartingale $X=(X_t)_{0 \leq t \leq T}$ to a sandwiched strong supermartingale $\mathcal{X}=(X^p, X)$ by \eqref{D5.1}.
Hence in the case of a local martingale both integrals \eqref{D4} and \eqref{D7} are equal to the usual stochastic integral.

In the context of Theorem \ref{mainthm} above we call $\cY=(Y^p, Y)=\big( (Y^{0,p}, Y^{1,p}), (Y^0, Y^1)\big)$ a \emph{sandwiched strong supermartingale deflator} (see \eqref{def:B}), if $Y=(Y^0, Y^1) \in \cB(y)$ and $(Y^{0,p}, Y^0)$ and $(Y^{1,p}, Y^1)$ are sandwiched strong supermartingales and the process $\widetilde{S}^p$ lies in the bid-ask spread, i.e.
$$\widetilde{S}^p_t:=\frac{Y^{1,p}_t}{Y^{0,p}_t} \in [(1-\lambda)S_{t-},S_{t-}], \quad t\in[0,T].$$

The definitions above allow us to obtain the following extension of Theorem \ref{mainthm}, which is our second main result. Roughly speaking, it states that the hypotheses of Theorem \ref{mainthm} suffice to yield a shadow price if one is willing to interpret this notion in a more general ``sandwiched sense'' rather than in the strict sense of Definition \ref{def:shadow}.
\bt\label{mt2}
Under the assumptions of Theorem \ref{mainthm}, let $(Z^n)_{n=1}^\infty$ be a minimising sequence of local $\lambda$-consistent price systems $Z^n=(Z^{0,n}_t, Z^{1,n}_t)_{0 \leq t \leq T}$ for the dual problem \eqref{mt1:eq4}, i.e.
$$
E\big[V\big(\widehat{y}(x)Z^{0,n}_T\big)\big]\searrow v\big(\widehat{y}(x)\big),\quad\text{as $n\to\infty$}.
$$
Then there exist convex combinations $\tZ^n\in\conv(Z^n, Z^{n+1},\ldots)$ and a sandwiched strong supermartingale deflator $\widehat{\cY}=(\widehat{Y}^p, \widehat{Y})$ such that
\begin{align}
\widehat{y}(x)(\tZ^{0,n}_{\tau-}, \tZ^{1,n}_{\tau-})&\stackrel{P}{\longrightarrow} (\widehat{Y}^{0,p}_\tau, \widehat{Y}^{1,p}_\tau), \label{12.1}\\
\widehat{y}(x)(\tZ^{0,n}_\tau, \tZ^{1,n}_\tau)&\stackrel{P}{\longrightarrow} (\widehat{Y}^{0}_\tau, \widehat{Y}^{1}_\tau), \label{12.2}
\end{align}
as $n\to\infty$, for all $[0,T]$-valued stopping times $\tau$ and we have, for any primal optimiser $\hvp=(\hvp^0,\hvp^1)$, that
\begin{equation}\label{mt2:eq1}
\widehat{Y}^0 \hvp^0(x) + \widehat{Y}^1 \hvp^1(x) = \widehat{Y}^0 \big(x+\hvp^1(x) \sint \widehat{\mathcal{S}}\big),
\end{equation}
where 
$$\widehat{\mathcal{S}}=(\widehat{S}^p, \widehat{S})= \left( \frac{\widehat{Y}^{1,p}}{\widehat{Y}^{0,p}}, \frac{\widehat{Y}^1}{\widehat{Y}^0} \right)$$
and 
\begin{equation}\label{p13}
x+\hvp^1(x) \sint \widehat{\mathcal{S}}_t:=x+ \int^t_0 \hvp^{1,c}_u(x) d\hS_u+ \sum_{0 \leq u <t} \Delta\hvp^1_u(x)(\hS_t - \hS^p_u)+ \sum_{0 < u \leq t} \Delta_+ \hvp^1_u(x) (\hS_t - \hS_u).
\end{equation}
This implies (after choosing a suitable version of $\hvp^1 (x)$) that
\begin{align}
\{d\hvp^{1,c}(x)>0\}&\subseteq \{\hS=S\},\notag&\{d\hvp^{1,c}(x)<0\}&\subseteq \{\hS=(1-\lambda)S\},\notag\\
\{\Delta \hvp^{1}(x)>0\}&\subseteq \{\hS^p=S_-\},\notag&\{\Delta \hvp^{1}(x)<0\}&\subseteq \{\hS^p=(1-\lambda)S_-\},\notag\\
\{\Delta_+ \hvp^{1}(x)>0\}&\subseteq \{\hS=S\},&\{\Delta_+ \hvp^{1}(x)<0\}&\subseteq \{\hS=(1-\lambda)S\}.\label{mt2:eq2}
\end{align}
\et
\vskip10pt
For \emph{any} sandwiched supermartingale deflator $\mathcal{Y}=(Y^p, Y)$, with the associated price process $\widetilde{\mathcal{S}}=(\widetilde{S}^p, \widetilde{S}) =(\frac{Y^{1,p}}{Y^{0,p}}, \frac{Y^1}{Y^0})$, and \emph{any} trading strategy $\varphi\in\cA(x)$ we have for the liquidation value $V^{liq}(\vp)$ defined in \eqref{eq:lv} that
\begin{equation}\label{D9}
V^{liq}_t(\vp)\leq x + \int^t_0 \varphi^{1,c}_u d \widetilde{S}_u + \sum_{0 \leq u < t} \Delta\varphi^1_u(\widetilde{S}_t-\widetilde{S}^p_u) + \sum_{0 <u \leq t} \Delta_+\varphi^1_u (\widetilde{S}_t - \widetilde{S}_u)=: x + \varphi^1 \sint \widetilde{\mathcal{S}}_t.
\end{equation}
Indeed, the usual argument applies that a self-financing trading for any price process $\widetilde{\mathcal{S}}=(\widetilde{S}^p, \widetilde{S})$ taking values in the bid-ask spread and without transaction costs is at least as favourable than trading for $S$ with transaction costs. The relations \eqref{mt2:eq1} and \eqref{mt2:eq2} illustrate that the optimal strategy $\hvp=(\hvp^0, \hvp^1)$ only trades when $\widehat{\mathcal{S}}=(\widehat{S}^p, \widehat{S})$ assumes the least favourable position in the bid-ask spread.

Let us now come to the third question posed at the end of section 2. We shall state in Theorem \ref{mt3} that the sandwiched strong supermartingale deflator $\widehat{\mathcal{S}}=(\widehat{S}^p, \widehat{S})$ may be viewed as a frictionless shadow price if one is ready to have a more liberal concept than Def. \ref{def:shadow} above.

Recall once more that the basic message of the concept of a shadow price $\widehat{\mathcal{S}}$ is that a strategy $\varphi$ which is trading in this process without transaction costs cannot do better (w.r.~to expected utility) than the above optimiser $\hvp$ by trading on $S$ under transaction costs $\lambda$. For this strategy $\hvp$ we have established in \eqref{p13} that trading at prices $\widehat{\mathcal{S}}$ without transaction costs or trading in $S$ under transaction costs $\lambda$ amounts to the same thing. These two facts can be interpreted as the statement that $\widehat{\mathcal{S}}$ serves as shadow price.

Let us be more precise which class of processes $\varphi^1=(\varphi^1_t)_{0 \leq t \leq T}$ we allow to compete against $\hvp^1=(\hvp^1_t)_{0 \leq t \leq T}$ in \eqref{p13}. First of all, we require that $\varphi^1$ is predictable and of finite variation so that the stochastic integral \eqref{p13} is well-defined. Secondly, we allow $\varphi^1$ to trade without transaction costs in the process $\widehat{\mathcal{S}}$ which is precisely reflected by \eqref{p13}. 
More formally, we may associate to the process $\varphi^1$ of holdings in stock the process $\varphi^0$ of holdings in bond by equating $\varphi^0_t + \varphi^1_t \widehat{S}_t$ to the right hand side of \eqref{D9}, i.e.
\begin{align}\label{D1}
\varphi^0_t&:=x+\vp^1\sint \widehat{\mathcal{S}}_t-\varphi^1_t \widehat{S}_t, \quad 0 \leq t \leq T.
\end{align}
One may check that $\varphi^0$ is a predictable finite variation process and also satisfies $\varphi^0_{t-} = x+\vp^1\sint \widehat{\mathcal{S}}_{t-} - \varphi^1_{t-} \widehat{S}^p_{t-}.$ The process $\varphi=(\varphi^0_t, \varphi^1_t)_{0 \leq t \leq T}$ then models the holdings in bond and stock induced by the process $\varphi^1$ considered as trading strategy without transaction costs on  $\widehat{\mathcal{S}}$.

We now come to the third requirement on $\varphi$ , namely the delicate point of {\it admissibility}. The admissibility condition which naturally corresponds to the notion of frictionless trading is $\varphi^0_t + \varphi^1_t \widehat{S}_t \geq 0$, for all $0 \leq t \leq T$. This notion was used in Definition \ref{def:SI:1}. However, it is too wide in order to allow for a meaningful theorem in the present general context, even if we restrict to continuous processes $\widehat{S}$. This is shown by a counterexample in \cite{CS14} (compare also \cite{BCKMK11} and \cite{CMKS14} for examples in discrete time).
Instead, we have to be more modest and define the admissibility in terms of the original process $S$ under transaction costs $\lambda$. We therefore impose the requirement that the liquidation value $V^{liq}_t(\vp)$ as defined in \eqref{eq:lv} has to remain non-negative, i.e.
\begin{align}\label{2.32}
V_t^{liq}(\vp):={}&\vp^0_t+(\vp^1_t)^+(1-\lambda) S_t-(\vp^1_t)^-S_t\geq 0.
\end{align}
Summing up in economic terms: we compare the process $\hvp$ in Theorem \ref{mt2} with all competitors $\varphi$  which are self-financing w.r.~to $\widehat{\mathcal{S}}$ (without transaction costs) and such that their liquidation value $V^{liq}_t(\vp)$ under transaction costs $\lambda$ remains non-negative \eqref{2.32}.

\bt\label{mt3}
Under the assumptions of Theorem \ref{mt2} let $\varphi = (\varphi^0_t, \varphi^1_t)_{0 \leq t \leq T}$ be a predictable process of finite variation which is self-financing for $\widehat{\mathcal{S}}$ without transaction costs, i.e.~satisfies \eqref{D1} and is admissible in the sense of \eqref{2.32}.
Then the process 
\begin{align}\label{P1}
\widehat{Y}^0_t \varphi^0_t + \widehat{Y}^1_t \varphi^1_t=\widehat{Y}^0_t\big(x+\vp^1\sint\widehat{\cS}_t\big), \quad 0 \leq t \leq T,
\end{align}
is a non-negative supermartingale and
\begin{align}
E\big[U\big(x+\varphi^1\sint \widehat{\cS}_T\big)\big] \leq E\big[U\big(x+\hvp^1 \sint \widehat{\cS}_T \big) \big] = E\big[U\big(\hvp^0_T + \hvp^1_T \widehat{S}_T \big)\big] = E\big[U\big(V_{T}^{liq}(\hvp)\big)\big].\label{P1.2}
\end{align}
\et

We finish this section by formulating some positive results in the context of Theorem \ref{mainthm}. As in \cite{CMKS14}, we have \emph{under the assumptions of Theorem \ref{mainthm}}, the following two results clarifying the connection between dual minimisers and shadow prices in the sense of Def. \ref{def:shadow}. The first result is motivated by the work of Cvitanic and Karatzas \cite{CK96} shows that the following ``folklore'' is also true in the present framework of general  c\`adl\`ag processes $S$: if there is no ``loss of mass'' in the dual problem under transaction costs, then its minimiser corresponds to a shadow price in the usual sense.

\begin{prop}\label{lem:martingale}
If there is a minimiser $(\widehat{Y}^0,\widehat{Y}^1)\in\cB\big(\widehat{y}(x)\big)$ of the dual problem \eqref{DP3} which is a local martingale, then $\hS:=\widehat{Y}^1/\widehat{Y}^0$ is a shadow price in the sense of Def. \ref{def:shadow}.
\end{prop}

Conversely, the following result shows that \emph{if} a shadow price exists as above and satisfies $(NUPBR)$, it is necessarily derived from a dual minimiser. Note that by Proposition 4.19 in \cite{KK07} the existence of an optimal strategy to the frictionless utility maximisation problem \eqref{p1} for $\hS$ essentially implies that $\hS$ satisfies $(NUPBR)$.

\begin{prop}\label{lem:connection}
If a shadow price $\hS$ in the sense of Def. \ref{def:shadow} exists and satisfies $(NUPBR)$, it is given by $\hS=\widehat{Y}^1/\widehat{Y}^0$ for a minimiser $(\widehat{Y}^0,\widehat{Y}^1)\in\cB\big(\widehat{y}(x)\big)$ of the dual problem \eqref{DP3}.
\end{prop}

Similarly as in the frictionless case the duality relations above simplify for logarithmic utility.

\begin{prop}\label{prop:log}
For $U(x)=\log(x)$, we have under the assumptions of Theorem \ref{mainthm} that the solutions $\hvp=(\hvp^0_t,\hvp^1_t)_{0\leq t\leq T}$ to the primal problem
\begin{equation*}
\textstyle
E\big[\log\big(V_T^{liq}(\vp)\big)\big]\to\max!, \qquad \vp\in\cA(x),
\end{equation*}
and $\widehat{Y} = (\widehat{Y}^0_t, \widehat{Y}^1_t)_{0\leq t\leq T}$ to the dual problem
$$E[-\log(Y^0_T)-1] \to \min!, \quad Y=(Y^0, Y^1)\in \cB\big(\hy (x) \big),$$
for $\hy (x) = u'(x)=\frac{1}{x}$ exist and satisfy
\begin{align*}
\big(\hY^0,\hY^1\big)=\left(\frac{1}{\hvp^0_t + \hvp^1_t\hS_t},\frac{\widehat{S}_t}{\hvp^0_t + \hvp^1_t\widehat{S}_t}\right)_{0\leq t\leq T}
\end{align*}
where $\hS=\Big(\frac{\widehat{Y}^1_t}{\hY^0_t}\Big)_{0\leq t\leq T}$ can be characterised by \eqref{mt2:eq2}.
\end{prop}
\bp
Since $V_{T}^{liq}(\hvp)=\hvp^0_T+\hvp^1_T\hS_T$ and $U'(x)=\frac{1}{x}$, we have that $\hY^0_T=\frac{1}{\hvp^0_T + \hvp^1_T\hS_T}$ and 
$$\widehat{Y}^0_T\hvp^0_T+ \widehat{Y}^1_T\hvp^1_T=\widehat{Y}^0_T(\hvp^0_T+ \hvp^1_T\widehat{S}^1_T)=1$$
by part 3) of Theorem \ref{mainthm}. Therefore the martingale $\widehat{Y}^0\hvp^0+ \widehat{Y}^1\hvp^1=(\widehat{Y}^0_t\hvp^0_t+ \widehat{Y}^1_t\hvp^1_t)_{0\leq t\leq T}$ is constant and equal to $1$, which implies that $\big(\hY^0,\hY^1\big)=\left(\frac{1}{\hvp^0_t + \hvp^1_t\hS_t},\frac{\widehat{S}_t}{\hvp^0_t + \hvp^1_t\widehat{S}_t}\right)_{0\leq t\leq T}$.
\ep
\section{Examples}\label{sec:ex}
\subsection{Truly l\`adl\`ag primal and dual optimisers}\label{sec:ex1}
We give an example of a price process $S=(S_t) _{0 \le t \le 1}$ in continuous time  such that for the problem of maximising expected logarithmic utility $U(x)=\log(x)$ the following holds for a fixed and sufficiently small $\lambda\in(0,1)$.
\bi
\item[{\bf 1)}] $S$ satisfies $(NFLVR)$ and therefore also $(CPS^{\lambda'})$ for all levels $\lambda'\in(0,1)$ of transaction costs.
\item[\textbf{2)}] The optimal trading strategy $\hvp=(\hvp^0,\hvp^1)\in\cA(1)$ under transaction costs exists and is truly l\`adl\`ag. This means that it is neither c\`adl\`ag nor    c\`agl\`ad. 
\item[\textbf{3)}]The candidate shadow price $\hS:=\frac{\hY^1}{\hY^0}$ given by the ratio of both components of the dual optimiser $\hY=(\hY^0,\hY^1)$ is truly l\`adl\`ag.
\ei
In particular, 3) implies that $\hS$ cannot be a semimartingale and therefore
\bi
\item[\textbf{4)}]No shadow price exists (in the sense of Def.\ref{def:shadow}).
 \ei
 
Note, however, that a shadow price in the more general ``sandwiched sense'' exists as made more explicit in Theorem \ref{mt3}.

For the construction of the example, let $\xi$ and $\eta$ be two random variables such that
\begin{align*}
P[\xi=3] &=\textstyle 1 - P[\xi=\frac{1}{2}] = \frac{5}{6} = p,\\
\\
P[\eta=2] &= (1-\ve),\\
\textstyle P[\eta=\frac{1}{n}] &= \ve 2^{-n}, \quad n\geq 1,
\end{align*}
where $\ve\in (0,\frac{1}{3})$. Let $\tau$ be an exponentially distributed random variable normalised by $E[\tau]=1$. We assume that $\xi,\, \eta$ and $\tau$ are independent of each other.
The ask price of the risky asset is given by
\be\label{ex1:pp}
S_t:=(1+(\xi-1)\mathbbm{1} _{[\frac{1}{2},1]}(t))\big(1+a_t(\eta-1)\mathbbm{1}_{[(\tau+\frac{1}{2})\wedge 1,1]}(t)\big)\quad\text{for  $t\in[0,1]$},
\ee
where $a_t=\frac{1}{3}-\frac{1}{3}(t-\frac{1}{2})$ is a linearly decreasing function and $\sigma=(\tau+\frac{1}{2})\wedge 1$. As filtration $\cF=(\F _t) _{0\le t \le 1}$ we take the one generated by $S=(S_t)_{0\le t \le 1}$ made right continuous and complete.

In prose the behaviour of the ask price $S$ is described as follows. The process starts at $1$ and remains constant until it jumps by $\Delta S_{\frac{1}{2}} = (\xi-1)$ at time $\frac{1}{2}$. After time $\frac{1}{2}$ the process jumps again by $\Delta S_{\sigma}=(1+(\xi-1)\mathbbm{1} _{\llbracket\frac{1}{2},1\rrbracket})\big(1+a_\sigma(\eta-1)\big)$ at the stopping time $\sigma$. 

Let us motivate intuitively why $S$ enjoys the above properties 1) - 4). We first concentrate on $t \in [\frac{1}{2},1]$ where the definition of $\eta$ plays a crucial role. There is an overwhelming probability for $\eta$ to assume the value $2$ which causes a positive jump of $S$ at time $\sigma$. Hence the log utility maximiser wants to hold many of these promising stocks when $\sigma$ happens. What prevails her from buying too many stocks is the (small but) strictly positive probability that $\eta$ takes values less than $1$, which results in a negative jump of $S$ at time $\sigma$. Similarly as in (\cite{KS99}, Example 5.1') the definition of $\eta$ is done in a way that at time $\sigma$ the ``worst case"', i.e. $\{\eta=0\}$, does not happen with positive probability, while the ``approximately worst cases'' $\{\eta=\frac{1}{n}\}$ happen with strictly positive probability. The explicit calculations in Appendix \ref{A:Ex1} below show that, similarly as in (\cite{KS99}, Ex. 5.1'), the optimal strategy for the log utility maximiser consists in holding precisely as many stocks such that, if $S$ happens to jump at time $t$ and $\eta$ {\it would} assume the value $\eta=0$ (which $\eta$ does {\it not} with positive probability) the resulting liquidation value $V^{liq}_t(\vp)$ {\it would be} precisely $0$ (compare Appendix \ref{A:Ex1} below) which would result in $U(0)=-\infty$. Spelling out the corresponding equation (see Proposition \ref{A:Ex1:prop1}) results in
$$\hvp^1_t=\frac{\hvp^0_{t-}+\hvp^1_{t-}S_{t-}}{S_{t-}}\frac{1}{\lambda+(1-\lambda) a_t}, \qquad (t,\om) \in\textstyle \rrbracket\frac{1}{2}, \sigma  \rrbracket$$
which the log utility maximiser will follow for $t \in (\frac{1}{2}, \sigma]$. As $(a_t)_{\frac{1}{2} \leq t \leq 1}$ was chosen to be strictly decreasing we obtain
$$d\hvp^1_t > 0, \qquad t \in \textstyle( \frac{1}{2}, \sigma].$$
Speaking economically, the log utility maximiser increases her holdings in stock during the entire time interval $( \frac{1}{2}, \sigma]$. Hence a candidate $\widehat{S}=(\widehat{S}_t)_{0 \leq t \leq 1}$ for a shadow price process has to equal the ask price $S_t$ for $t \in (\frac{1}{2}, \sigma).$

Let us also discuss the optimal strategy $\hvp_t$ for $0 \leq t \leq \frac{1}{2}.$ The random variable $\xi$ is designed in such a way that the resulting jump $\Delta S_{\frac{1}{2}}$ of $S$ at time $t=\frac{1}{2}$ has sufficiently positive expectation so that the log utility maximiser wants to be long in stock at time $t= \frac{1}{2}$, i.e.~$\hvp^1_{\frac{1}{2}} > 0$ (compare Proposition \ref{A:Ex1:prop1}). As the initial endowment $\hvp_{0}=(1,0)$ has no holdings in stock, the log utility maximiser will purchase the stock at some time during $[0, \frac{1}{2}).$ It does not matter when, as $S$ is constant during that time interval. As a consequence, a candidate $\widehat{S}$ for a shadow price process must equal the ask price $S$ during the entire time interval $[0, \frac{1}{2})$, i.e.~$S_t = \widehat{S}_t,$ for $t \in [0, \frac{1}{2}).$

Finally let us have a look what happens to the log utility maximiser at time $t=\frac{1}{2}.$ If $\Delta S_{\frac{1}{2}} < 0$ (which happens with positive probability as $P[\xi= \frac{1}{2}] = \frac{1}{6} >0$) she immediately has to reduce her holdings in stock, i.e.~at time $t=\frac{1}{2}.$ Otherwise there is the danger that the totally inaccessible stopping time $\sigma$ will happen arbitrarily shortly after $t=\frac{1}{2}.$ If, in addition, $\eta$ assumes the value $\frac{1}{n}$, for large enough $n$, this would result in a negative liquidation value $V_{T}^{liq}(\hvp)$ with positive probability which is forbidden. Hence, conditionally on the set $\{ \xi= \frac{1}{2}\}$, each candidate $\widehat{S}$ for a shadow price must equal the bid price $(1-\lambda)S$ at time $t=\frac{1}{2}$,~i.e. 
$$\widehat{S}_{\frac{1}{2}}= (1-\lambda)S_{\frac{1}{2}} \quad \mbox{on}\quad  \{\Delta S_{\frac{1}{2}}<0\}.$$
Summing up: On $\{\Delta S_{\frac{1}{2}}<0\}=\{\xi=\frac{1}{2}\}$ a shadow price process $\widehat{S}=(\widehat{S}_t)_{0 \leq t \leq 1}$ necessarily satisfies with positive probability
$$
\widehat{S}_t:=\begin{cases}
S_t  &: 0 \leq t < \frac{1}{2},\\
(1-\lambda) S_t &: t= \frac{1}{2}, \\
S_t &: \frac{1}{2} < t < \sigma,\\
(1-\lambda)S_t &: \sigma \leq t\leq 1.
\end{cases}
$$

In other words, the process $\widehat{S}$ has to be {\it truly  l\`adl\`ag} at $t=\frac{1}{2}.$ In particular, $\widehat{S}$ cannot be a semimartingale and therefore there cannot be a shadow price process in the sense of Definition \ref{def:shadow}. We have thus shown the validity of assertions 1)--4) above.

Let us still have a look at the dual optimiser which can be explicitly calculated (see Proposition \ref{A:Ex1:prop1})
\begin{align*}
\hY=\left(\widehat{Y}^0,\widehat{Y}^1\right)=\left(\frac{1}{\hvp^0 + \hvp^1 \widehat{S}},\frac{\widehat{S}}{\hvp^0 + \hvp^1 \widehat{S}}\right).
\end{align*}
This process is a genuine optional strong supermartingale which displays right jumps
\begin{align*}
\Delta_+ \hY^0_{\frac{1}{2}} &=\hY^{0}_{\frac{1}{2}}\frac{-\lambda}{\lambda+(1-\lambda)a_{\frac{1}{2}}}\label{ex1:do:a}\\
\Delta_+ \hY^1_{\frac{1}{2}}&=\hY^{1}_{\frac{1}{2}}\left(1-\frac{\lambda}{\lambda+(1-\lambda)a_{\frac{1}{2}}}\right).
\end{align*}

The property of having right jumps is in stark contrast to being a (local) martingale which is always c\`adl\`ag.

However, according to Theorem \ref{mt2} we know that there exists an approximating sequence $(Z^n)_{n=1}^\infty$ of $\lambda$-consistent price systems $Z^n=(Z^{0,n}_t, Z^{1,n}_t)_{0 \leq t \leq 1}$ for the dual minimiser $\hY=(\hY^0_t,\hY^1_t)_{0 \leq t \leq 1}$ such that
$$(Z^{0,n}_\tau, Z^{1,n}_\tau)\xrightarrow{P}(\hY^0_\tau,\hY^1_\tau),\quad \text{as $n\to\infty$},$$
for all $[0,1]$-valued stopping times $\tau$. This illustrates nicely how a sequence of c\`adl\`ag processes produces a right jump in the limit and we give such an approximating sequence $(Z^n)_{n=1}^\infty$ of $\lambda$-consistent price systems $Z^n=(Z^{0,n}_t, Z^{1,n}_t)_{0 \leq t \leq 1}$ in Proposition \ref{A:Ex1:propZn}.

The reader who wants to verify the above characteristics may consult the explicit calculations in Appendix \ref{A:Ex1} below.

\subsection{Left limit of limits $\ne$ limit of left limits}\label{sec:ex2}
While the previous example showed the necessity of going beyond the framework of c\`adl\`ag processes we now show that there is indeed no way to avoid the appearance of ``sandwiched processes'' for the dual optimiser in Theorem \ref{mt2}.

For the problem of maximizing logarithmic utility under transaction costs $\lambda \in(0,1)$ with initial endowment $(\varphi^0_0, \varphi^1_0)=(1,0)$, we give an example of a semimartingale price process $S=(S_t)_{0 \leq t \leq 1}$ such that:
\bi
\item[{\bf 1)}] $S$ satisfies $(NFLVR)$ and therefore also $(CPS^{\lambda'})$ for all levels $\lambda'\in(0,1)$ of transaction costs.
\item[{\bf 2)}] The primal and dual optimisers $\hvp=(\hvp^0, \hvp^1)$ and $\hY=(\hY^0, \hY^1)$ exist.
\item[{\bf 3)}] The predictable supermartingale $\hY^p=(\hY^p_t)_{0\leq t\leq T}$ in Theorem \ref{mt2} does not coincide with the left limit $\hY_-=(\hY_{t-})_{0\leq t\leq T}$ of the optional strong supermartingale.
\ei
More precisely, the more detailed properties are:
\bi
\item[{\bf 4)}] There exists a predictable stopping time $\vr > 0$ such that, on $\{\vr<\infty\}$, the optimal strategy buys stocks immediately before time $\vr$, i.e. $\Delta\hvp_\vr=\hvp_\vr-\hvp_{\vr-} > 0$, but $\hS_{\vr-}:= \frac{\hY^1_{\vr-}}{\hY^0_{\vr-}}=(1-\lambda)S_{\vr-} \not=S_{\vr-}$.
\item[{\bf 5)}] There is a minimising sequence $Z^n=(Z^{0,n}, Z^{1,n})$ of consistent price systems for the dual problem \eqref{mt1:eq4} such that
$$(Z^{0,n}_\tau, Z^{1,n}_\tau) \stackrel{P}{\longrightarrow}  (\hY^0_\tau, \hY^1_\tau)$$
for all finite stopping times $\tau$ and 
$$\widetilde{S}^n_{\vr-}:= \frac{Z^{1,n}_{\vr-}}{Z^{0,n}_{\vr-}} \stackrel{P}{\longrightarrow} S_{\vr-}\ne(1-\lambda)S_{\vr-}=\hS_{\vr-}= \frac{\hY^1_{\vr-}}{\hY^0_{\vr-}}\quad\text{ on $\{\vr<\infty\}$}.$$
\ei

To construct the example, we set $t_j:=\frac{1}{2}-\frac{1}{2+j}$ for $j\in \mathbb{N}$ and $t_\infty=\frac{1}{2}$ and consider a stopping time $\sigma$ valued in $\{\frac{1}{2} - \frac{1}{2+j}~|~j \in \mathbb{N}\} \cup \{\frac{1}{2}\}$ such that $P(\sigma=t_j)=\frac{1}{2} \cdot \frac{1}{2^j}$ and $P(\sigma=t_\infty=\frac{1}{2})=\frac{1}{2}$. Let $\eta$ be a random variable independent of $\sigma$ such that 
\begin{align*}
P(\eta=2)&=(1-\varepsilon),\\
P(\eta=\textstyle\frac{1}{n})&=\varepsilon 2^{-n}, \qquad n\in\N,
\end{align*}
where $\varepsilon \in(0,\frac{1}{3})$. Let $(a_j)^\infty_{j=1}$ be a strictly increasing sequence of real numbers such that $a_j > \frac{1}{2}$ and $\lim_{j\to\infty} a_j=\frac{2}{3}.$ We then define the ask price $S=(S_t)_{0 \leq t \leq 1}$ to be a process such that $S_0=1$ and 
\be
\Delta S_\sigma=S_\sigma - S_{\sigma-} =\begin{cases}a_j(\eta-1) &: \sigma = t_j,\\
\frac{1}{2}(\eta-1) &:\sigma = \frac{1}{2},\end{cases}\label{Ex2:S}
\ee
and that is constant anywhere else.

As the jumps $\Delta S_{t_j} = a_j(\eta-1) \mathbbm{1}_{\{\sigma=t_j\}}$ and $\Delta S_{\frac{1}{2}}=\frac{1}{2} (\eta-1) \mathbbm{1}_{\{\sigma=\frac{1}{2}\}}$ are very favourable for the logarithmic investor, she wants to hold as many stocks as possible, provided the admissibility constraint $V_{T}^{liq}(\hvp) \geq 0$ is not violated. Similarly as in the preceding example this amounts to buying before time $t_1$ the maximal amount $\hvp^1_{t_1}$ of stocks such that in the hypothetic event $\{\eta=0\}$ the liquidation value would equal precisely zero which results in 
$$\hvp^1_{t_1}=\frac{1}{\lambda+(1-\lambda)a_1}.$$
At time $t_1$ we have to possibilities: either $\sigma=t_1$ in which case the investor may liquidate her position and go home, as the stock will remain constant after time $t_1.$ Or $\sigma > t_1$ so that there is still the possibility of jumps at time $t_2, t_3, \dots, t_\infty.$ At some point during the internal $[t_1, t_2)$ the utility maximiser will adjust the portfolio so that the liquidity constraint $V_{t_2}(\hvp)\geq 0$ is not violated. Again this results in holding the maximal amount $\hvp^1_{t_2}$ of stocks at time $t_2$ so that, in the hypothetical event $\{\eta=0\}$ we find for the liquidation value $V_{t_2} (\hvp)=0.$  A straightforward computation (see Proposition \ref{A:Ex2:prop} below) yields
$$\hvp^1_{t_2} = (1-\lambda \hvp^1_{t_1}) \frac{1}{(1-\lambda)a_2}.$$
The decisive point is the following: as $a_2 > a_1$, we obtain $\hvp^1_{t_2} < \hvp^1_{t_1}$; in other words, the investor has to {\it sell} stock between $t_1$ and $t_2$. Of course she can only do this at the bid price $(1-\lambda)S$. Continuing in an obvious way, the investor keeps selling stock in each interval $[t_j, t_{j+1})$ if she was not stopped before, i.e. in the event $\{\sigma >t_j\}.$ Therefore a shadow price must satisfy $\hS_{t_j}=(1-\lambda)S_{t_j}$ for all $j\geq 2$ and hence $\hS_{\frac{1}{2}-}=\lim_{j\to\infty}(1-\lambda)S_{t_j}=(1-\lambda)S_{\frac{1}{2}-}$ on the event $\{\sigma \geq \frac{1}{2}\}$. At time $t=\frac{1}{2}$ the situation changes again. As $\lim_{j\to\infty} a_j = \frac{2}{3}$ is higher than $\frac{1}{2}$, the agent {\it buys} stock immediately before $t=\frac{1}{2}$ (but after all the $t_j$'s), i.e.~at time $t=\frac{1}{2}-.$ Of course, for this purchase the ask price $S_{\frac{1}{2}-}$ applies. But this is in flagrant contradiction to the above requirement  that $\hS_{t_j}=(1-\lambda)S_{t_j}$ for all $j\geq 2$ on $\{\sigma=\frac{1}{2}\}$. The way out of this dilemma is precisely the notion of a ``sandwiched supermartingale deflator'' as isolated in Theorem \ref{mt2}.

Let us understand this phenomenon in some detail. We approximate the process $S$ by a sequence $(S^n)^\infty_{n=1}$ of simpler processes, all defined on the same filtered probability space $\big(\Omega, \mathcal{F}, (\mathcal{F}_t)_{0 \leq t \leq 1}, P\big)$ generated by $S$.
Let
\begin{align*}
\eta^n(\omega)= 
\begin{cases}
\eta(\omega) &: \eta(\omega) \geq \frac{1}{\eta},\\
\frac{1}{n} &: \eta(\omega) <  \frac{1}{\eta}
\end{cases}
\end{align*}
and
\begin{align*}
\sigma_n(\omega)= 
\begin{cases}
\sigma(\omega) &: \sigma(\omega) \leq t_n, \\
 \frac{1}{2} &: \mbox{else}.
\end{cases}
\end{align*}
Similarly as above we define
\begin{align}
\Delta S^n_{\sigma_n} = S^n_{\sigma_n} - S^n_{\sigma_n-} = \begin{cases}a_j(\eta^n-1)&: \sigma_n \leq t_n,\\
 \frac{1}{2} (\eta^n-1)&: \sigma_n=\frac{1}{2}.
 \end{cases}
 \label{Ex2:Sn}
\end{align}
The $\sigma$-algebra generated by process $S^n$ is finite and therefore the duality theory of portfolio optimisation is straightforward (compare \cite{KMK11} and \cite{S14}). The primal and dual optimiser for the log utility maximisation problem for $S^n$ can be easily computed; see Lemma \ref{A:Ex2:lSn} below.  The dual optimiser $\widehat{Z}^n=(\widehat{Z}_t^{0,n}, \widehat{Z}_t^{1,n})_{0 \leq t \leq 1}$ now is a true martingale (taking only finitely many values). One may explicitly show that the quotient $\widehat{S}^n=\frac{\widehat{Z}^{1,n}}{\widehat{Z}^{0,n}}$ is a shadow prices in the sense of Definition \ref{def:shadow} for which we obtain
\be\label{C2}
\widehat{S}^n_t= 
\begin{cases}
S^n_t&: 0 \leq t < t_1,\\
(1-\lambda) S^n_t &: t_1 \leq t < t_n,\\
S^n_t &: t_n \leq t < \frac{1}{2},\\
(1-\lambda) S^n_t &: \frac{1}{2} \leq t \leq 1 
\end{cases}
\ee
on $\{\sigma=\frac{1}{2}\}$ for sufficiently large $n$. (More precisely, that it can be extended to a shadow price.) What is the limit of the processes $(\widehat{S}^n_t)_{0 \leq t \leq 1}$? Obviously the process $\widehat{S}=(\widehat{S}_t)_{0 \leq t \leq T}$ defined as
\begin{align}
\widehat{S}_t= 
\begin{cases}
S_t&: 0 \leq t < t_1,\\
(1-\lambda)S_t&: t_1 \leq t \leq 1
\end{cases}\label{C3}
\end{align}
satisfies $\widehat{S}^n_\tau \to (1-\lambda) S_\tau$ $P$-a.s.~for all $[0, 1]$-valued stopping times $\tau$. However, 
\begin{align}\label{C6}
\widehat{S}^n_{\frac{1}{2}-} \xrightarrow{\text{$P$-a.s.}}S_{\frac{1}{2}-},\quad\text{as $\to\infty$},
\end{align}
an information which is not encoded in the process $\widehat{S}$, but only in the approximating sequence $\hS^n$. The remedy is to pass to the ``sandwiched supermartingales'' $\big((\widehat{Y}_t^{0,p})_{0 \leq t \leq 1}, (\widehat{Y}_t^0)_{0 \leq t \leq 1}\big)$ and $\big((\widehat{Y}_t^{1,p})_{0 \leq t \leq 1}, (\widehat{Y}_t^1)_{0 \leq t \leq 1}\big)$ and to accompany the process $\widehat{S} = \frac{\widehat{Y}^1}{\widehat{Y}^0}$ with the predictable process $\widehat{S}^p=\frac{\widehat{Y}^{1,p}}{\widehat{Y}^{0,p}}$ for which we find 
$$\widehat{S}^p_{\frac{1}{2}} = \lim_{n \to \infty} \frac{\widehat{Y}^{1,n}_{\frac{1}{2}-}}{\widehat{Y}^{0,n}_{\frac{1}{2}-}} = S_{\frac{1}{2}-}$$
as in \eqref{C6} above.

Again the reader who wants to verify the above characteristics may consult the explicit calculations in Appendix \ref{App:Ex2} below.
\appendix
\section{Proofs for Section \ref{sec3} }\label{App:B}

The proof of  parts 1)--3) of Theorem \ref{mainthm} follow from the abstract versions of the main results in \cite[Theorems 3.1 and 3.2]{KS99} once we have shown in the lemma below that the relations in \cite[Proposition 3.1]{KS99} hold true. We call a set $\mathcal{G}\subseteq L^0_+(P)$ \emph{solid}, if $0\leq f\leq g$ and $g\in\mathcal{G}$ imply that $f\in\mathcal{G}$, and use that $\cC(x)=x\cC(1)=:x\cC$ and \mbox{$\cD(y)=y\cD(1)=:y\cD$.} 

\bl\label{lpolar}
Suppose that $S$ satisfies $(CPS^{\lambda'})$ locally for all $\lambda'\in(0,\lambda)$. Then:
\bi
\item[{\bf 1)}] $\cC$ and $\cD$ are convex, solid and closed in the topology of convergence in measure.
\item[{\bf 2)}] $g\in\cC$ iff $E[gh]\leq 1$, for all $h\in\cD$, and $h\in\cD$ iff $E[gh]\leq 1$, for all $g\in\cC$.
\item[{\bf 3)}] The closed, convex, solid hull of $D$ in $L^0_+(P)$ is given by $\cD$, i.e.~$\overline{\sol(D)}=\cD$.
\item[{\bf 4)}] $\cC$ is a bounded subset of $L^0_+(P)$ and contains the constant function $1$.
\item[{\bf 5)}] $D:=\{Z^0_T~|~(Z^0,Z^1)\in\mathcal{Z}_{loc}\}$ is closed under countable convex combinations.
\ei
\el
\bp
1) The sets $\cC$ and $\cD$ are convex and solid by definition.

To prove the closedness of $\cC$, let $\varphi^n=(\varphi^{0,n}, \varphi^{1,n}) \in \cA (1)$ be such that $g^n:=V_T(\varphi^n)$ converge to some $g \in L^0_+(P)$ in probability. By the proof of Theorem 3.5 in \cite{CS06} (or Theorem 3.4 in \cite{S14}) it is then sufficient to show that $\big(\Var_T(\varphi^{1,n})\big)^\infty_{n=1}$ and hence also $\big(\Var_T(\varphi^{0,n})\big)^\infty_{n=1}$ are bounded in probability to deduce that $g=V_T(\varphi) \in \cC$ for some $\varphi=(\varphi^0, \varphi^1) \in \cA(1).$ Indeed, by Proposition 3.4 in \cite{CS06} (or the proof of Theorem 3.4 in \cite{S14}) there then exists a sequence of convex combinations $(\widetilde{\varphi}^{0,n}, \widetilde{\varphi}^{1,n}) \in \conv \big((\varphi^{0,n}, \varphi^{1,n}), (\varphi^{0,n+1}, \varphi^{1,n+1}),\dots\big)$ and a predictable finite variation process $\varphi=(\varphi^0, \varphi^1)$ such that
$$P\Big[(\widetilde{\varphi}^{0,n}_t, \widetilde{\varphi}^{1,n}_t) \to (\varphi^0_t, \varphi^1_t),\ \forall t \in [0,T]\Big]=1,$$
which already implies that $\varphi=(\varphi^0, \varphi^1) \in \cA(1).$ To see the boundedness of $\big(\Var_T(\varphi^{1,n})\big)^\infty_{n=1}$ in probability, we observe that it is sufficient to establish that $\big(\Var_{\tau_m}(\varphi^{1,n})\big)^\infty_{n=1}$ is bounded in probability for each $m \in \mathbb{N}$ for a localising sequence $(\tau_m)^\infty_{m=1}$ of stopping times. But this follows from the assumption that $S$ satisfies $( CPS^{\lambda'})$ locally for \emph{some} $\lambda' \in (0, \lambda)$ by Lemma 3.2 in \cite{CS06} (or Lemma 3.1 in \cite{S14}). Note that our notion of admissibility in \eqref{eq:lv} implies condition (iii) of Definition 2.7 in \cite{CS06} for any $a > 0$ locally.

The closedness of $\cD$ follows by combining similar arguments as in Lemma 4.1 in \cite{KS99} with a new version of Koml\'os lemma for non-negative optional strong supermartingales in \cite{CS13}. To that end, let $(h^n)$ be a sequence in $\cD$ converging to some $h$ in measure. Then there exists a sequence $\big((Y^{0,n},Y^{1,n})\big)^\infty_{n=1}$ in $\cB(1)$ such that $Y^{0,n}_T=h^n$ for each $n\in\N$. Since $Y^{0,n}$ and $Y^{1,n}$ are non-negative optional strong supermartingales, there exist by Theorem 2.7 in \cite{CS13} a sequence $(\widetilde{Y}^{n,0},\widetilde{Y}^{n,1})\in\mathrm{conv}\big((Y^{0,n},Y^{1,n}),(Y^{0,n+1},Y^{1,n+1}),\ldots\big)$ for $n\geq 1$ and optional strong supermartingales $\widetilde{Y}^0$ and $\widetilde{Y}^1$ such that
\be
(\widetilde{Y}^{n,0}_\tau,\widetilde{Y}^{n,1}_\tau)\overset{P}{\longrightarrow}(\widetilde{Y}^0_\tau,\widetilde{Y}^1_\tau),\quad\text{as $n\to\infty$,}\label{eq:conv}
\ee
for all $[0,T]$-valued stopping times $\tau$. This convergence in probability is then sufficient to deduce that $\widetilde{Y}^0_0=1$, $\widetilde{Y}^0_T= h$, and that $\widetilde{Y}^0\vp^0+\widetilde{Y}^1\vp^1$ is a non-negative optional strong supermartingale for all $(\vp^0,\vp^1)\in\cA(1)$. To see the latter, observe that, for all stopping times $\sigma$ and $\tau$ such that $0\leq \sigma\leq \tau\leq T$, we have that
\begin{align*}
\widetilde{Y}^0_\sigma\vp^0_\sigma+\widetilde{Y}^1_\sigma\vp^1_\sigma&=\liminf_{n\to\infty}\big(\widetilde{Y}^{0,n}_\sigma\vp^0_\sigma+\widetilde{Y}^{1,n}_\sigma\vp^1_\sigma\big)\\
&\geq\liminf_{n\to\infty}E\big[\widetilde{Y}^{0,n}_\tau\vp^0_\tau+\widetilde{Y}^{1,n}_\tau\vp^1_\tau\big|\cF_\sigma\big]\\
&\geq E\Big[\liminf_{n\to\infty}\big(\widetilde{Y}^{0,n}_\tau\vp^0_\tau+\widetilde{Y}^{1,n}_\tau\vp^1_\tau\big)\Big|\cF_\sigma\Big]\\
&=E\big[\widetilde{Y}^0_\tau\vp^0_\tau+\widetilde{Y}^1_\tau\vp^1_\tau\big|\cF_\sigma\big]
\end{align*}
by Fatou's lemma for conditional expectations.

To conclude that $(\widetilde{Y}^0,\widetilde{Y}^1)\in\cB(1)$ and hence that $h\in\cD$, it remains to show that $(\widetilde{Y}^0,\widetilde{Y}^1)$ is $\R_+^2$-valued and $\tS:=\frac{\widetilde{Y}^1}{\widetilde{Y}^0}$ is valued in $[(1-\lambda)S,S]$. We begin with the latter assertion. For this, we assume by way of contradiction that the set $F:=\big\{\tS\notin[(1-\lambda)S,S]\big\}$ is not $P$-evanescent in the sense that $P\big(\pi(F)\big)>0$, where $\pi$ denotes the projection from $\OmT$ onto $\Om$ given by $\pi\big(\omt\big)=\om$. Since $F=\big\{\tS\notin[(1-\lambda)S,S]\big\}$ is optional, there exists by the optional cross-section theorem (see Theorem IV.84 in \cite{DM78}) a $[0,T]\cup\{\infty\}$-valued stopping time $\sigma$ such that $\llbracket \sigma_{\{\sigma<\infty\}}\rrbracket\subseteq F$, which means that $\tS_\sigma\notin[(1-\lambda)S_\sigma,S_\sigma]$ on $\{\sigma<\infty\}$, and $P(\sigma<\infty)>0$. By \eqref{eq:conv} we obtain that $\tS^n_\tau:=\frac{\widetilde{Y}^{1,n}_\tau}{\widetilde{Y}^{0,n}_\tau}\overset{P}{\longrightarrow}\tS_\tau$ for the $[0,T]$-valued stopping time $\tau:=\sigma\wedge T$. As $\tS^n_\tau\in[(1-\lambda)S_\tau,S_\tau]$, this implies that also $\tS_\tau$ is valued in $[(1-\lambda)S_\tau,S_\tau]$ and therefore yields a contradiction to the assumption that $P\big(\pi(F)\big)>0$. The assertion that $(\widetilde{Y}^0,\widetilde{Y}^1)$ is $\R_+^2$-valued follows by the same arguments and its proof is therefore omitted.

2) The first assertion follows by the local version of the superreplication theorem under transaction costs (Lemma \ref{A2}) below. We then obtain the second assertion  by the same arguments as the proof of Proposition 3.1 in \cite{KS99} which also imply 3).

4) The fact that $\cC$ contains the constant function $1$ follows by definition; the $L^0_+(P)$-boundedness is implied by the existence of a strictly positive element in $\cD$.

5) Given $(Z^{0,n},Z^{1,n})_{n=1}^\infty$ in $\cZ_{loc}$ and $(\mu_n)_{n=1}^\infty$ positive numbers such that $\sum_{n=1}^\infty\mu_n=1$, we have that $\sum_{n=1}^\infty\mu_n Z^{0,n}$ is a non-negative local martingale starting at $1$, $\sum_{n=1}^\infty\mu_n Z^{1,n}$ is a local martingale and $\frac{\sum_{n=1}^\infty\mu_n Z^{1,n}}{\sum_{n=1}^\infty\mu_n Z^{0,n}}$ takes values in $[(1-\lambda)S,S]$ which already gives 5).
\ep

\bl\label{A2}
Suppose that $S$ satisfies $(CPS^{\lambda'})$ locally for all $\lambda'\in(0,\lambda)$. Then we have that $g\in L^0_+(P)$ is in $\cC$ if and only if $E[gZ^0_T]\leq 1$ for all $Z=(Z^0,Z^1)\in\cZ_{loc}$.
\el

\bp The ``only if'' part follows from the fact that $Z^0\vp^0+Z^1\vp^1$ is a non-negative local supermartingale by Proposition 1.6 in \cite{S13} and hence a true supermartingale for all $\vp=(\vp^0,\vp^1)\in\cA(1)$ and $Z=(Z^0,Z^1)\in\cZ_{loc}$.

For the ``if'' part, let $(\tau_m)^\infty_{m=1}$ be a localising sequence of stopping times for some $Z \in \cZ_{loc}$ such that $Z^{\tau_m}=(Z^0, Z^1)^{\tau_m}$ is a $\lambda'$-consistent price system for $S^{\tau_m}$ for some $\lambda' \in (0,\lambda).$ Then 
$$g_m:=g \mathbbm{1}_{\{\tau_m = T\}} \in \cC_m:=\{V_{\tau_m}(\varphi)| \varphi \in \cA (1)\}$$
and $g \in \cC$ if and only if $g_m \in \cC_m$ for each $m \in \mathbb{N}$, as $\cC_m \subseteq \cC$, $g_m \stackrel{P}{\longrightarrow} g$ and $\cC$ is closed.

Assume now for a proof by contradiction that there exists some $m' \in \mathbb{N}$ such that $g_{m'}\notin \cC_{m'}$. As $S^{\tau_{m'}}$ satisfies the assumptions of the superreplication theorem under transaction costs in the version of Theorem 1.4 in \cite{S14}, there exists a $\lambda'$-consistent price system $\overline{Z}=(\overline{Z}^0, \overline{Z}^1)$ for $S^{\tau_{m'}}$ such that
$$E[g_{m'} \overline{Z}^0_{\tau_{m'}}] > 1.$$ 
By the assumption that $S$ admits a local $\mu$-consistent price system for any $\mu \in (0, \lambda)$ we can extend $\overline{Z}$ to a local $\lambda$-consistent price system $\widetilde{Z}=(\widetilde{Z}^0, \widetilde{Z}^1)$ by setting 
\begin{equation*}
\widetilde{Z}^0_t= 
\begin{cases}
\overline{Z}^0_t&: 0 \leq t < \tau_{m'},\\
\check{Z}^0_t \frac{\overline{Z}^0_{\tau_{m'}}}{\check{Z}^0_{\tau_{m'}} }&: \tau_{m'} \leq t \leq T,  \\
\end{cases}
\end{equation*}
\begin{equation*}
\widetilde{Z}^1_t= 
\begin{cases}
\overline{Z}^10_t&: 0 \leq t < \tau_{m'},\\
\check{Z}^1_t \frac{\overline{Z}^1_{\tau_{m'}}}{\check{Z}^1_{\tau_{m'}} }&: \tau_{m'} \leq t \leq T  \\
\end{cases}
\end{equation*}
for some local $\mu'$-consistent price system $\check{Z}=(\check{Z}^0, \check{Z}^1)$ with $0 < \mu' < \frac{\lambda - \lambda'}{2}.$
Since
$$E[g \widetilde{Z}^0_T] \geq E[g_m \overline{Z}^0_{\tau_m}] > 1,$$
this yields the contradiction to the assumption that $E[g Z^0_T] \leq 1$ for all $Z \in \cZ_{loc}$.
\ep

\bp[Proof of Theorem \ref{mainthm}]
The proof follows immediately from the abstract versions of the main results (Theorems 3.1 and 3.2) and Proposition 3.2 in \cite{KS99} by Lemma B.1. The process
$$\hY^0(\widehat{y}(x)) \hvp^0(x) + \hY^1 (\widehat{y}(x)) \hvp^1(x)= \big(\hY^0_t(\widehat{y}(x)) \hvp^1_t(x) + \hY^1_t(\widehat{y}(x)) \hvp^1_t(x)\big)_{0\leq t \leq T}$$
is a martingale, as it is an optional strong supermartingale that has constant expectation.
\ep

\bp[Proof of Theorem \ref{mt2}]
By the self-financing condition and integration by parts we can write
\begin{align*}
\hY^0_t(\widehat{y}(x)) \hvp^0_t(x) + \hY^1_t(\widehat{y}(x)) \hvp^1_t(x) = \hY^0_t(\widehat{y}(x))(\hvp^0_t(x) + \hvp^1_t(x) \hS_t) = \hY^0_t(\widehat{y}(x)) (x+\hvp^1 \sint \widehat{\cS}_t + A_t),
\end{align*}
where $A=(A_t)_{0 \leq t \leq T}$ is a non-increasing predictable process given by
\begin{align*}
A_t:={}&\int_0^t\big(\hS_u-S_u\big)d\hvp^{1,\uparrow,c}_u(x)+\int_0^t\big((1-\lambda)S_u-\hS_u\big)d\hvp^{1,\downarrow,c}_u(x)\\
&+\sum_{0<u\leq t}\big(\hS^p_{u}-S_{u-}\big)\Delta\hvp^{1,\uparrow}_u(x)+\sum_{0<u\leq t}\big((1-\lambda)S_u - \hS^p_{u}\big)\Delta\hvp^{1,\downarrow}_u(x)\\
&+\sum_{0\leq u< t}\big(\hS_u-S_u\big)\Delta_+\hvp^{1,\uparrow}_u(x)+\sum_{0\leq u< t} \big(\hS_u-(1-\lambda)S_u\big)\Delta_+\hvp^{1,\downarrow}_u(x)
\end{align*}
for $t \in [0,T]$. Since $A\equiv 0$ if and only if \eqref{mt2:eq2} holds $P\otimes \Var (\hvp^1)$-a.e., we immediately obtain the equivalence of \eqref{mt2:eq1} and  \eqref{mt2:eq2} after choosing a suitable version of $\hvp^1$ and therefore that it is sufficient to prove \eqref{mt2:eq2}.

To that end, we observe that by the proof of part 1) of Lemma \ref{lpolar} above and part 4) of Theorem \ref{mainthm} there exists a sequence $\big((Z^{0,n},Z^{1,n})\big)^{\infty}_{n=1}$ in $\mathcal{Z}_{loc}$ such that
\be
\Big(\widehat{y}(x)Z^{0,n}_\tau,\widehat{y}(x)Z^{1,n}_\tau\Big)\overset{P}\longrightarrow\Big(\hY^0_\tau\big(\widehat{y}(x)\big),\hY^1_\tau\big(\widehat{y}(x)\Big)\label{eq:convCPS1}
\ee
and
\be
\Big(\widehat{y}(x)Z^{0,n}_{\tau-},\widehat{y}(x)Z^{1,n}_{\tau-}\Big)\overset{P}\longrightarrow\Big(\hY^{0,p}_{\tau}\big(\widehat{y}(x)\big),\hY^{1,p}_{\tau}\big(\widehat{y}(x)\Big)\label{eq:convCPS2}
\ee
for all $[0,T]$-valued stopping times $\tau$. As $\tS^n:=\frac{Z^{1,n}}{Z^{0,n}}$ is valued in the bid-ask-spread $[(1-\lambda)S,S]$, any $(\vp^0,\vp^1)\in\cA(x)$ is also self-financing for $\tS^n$ without frictions ($\lambda=0$) and $Z^{0,n}\big(x+\vp^1\sint\tS^n\big)$ is a non-negative local martingale and hence a supermartingale. By integration by parts (see \eqref{SI:IP}) and the self-financing condition \eqref{sfc} we can write
\begin{align*}
\hvp^0_t(x)+\hvp^1_t(x)\tS^n_t={}&\hvp^0_t(x)+\hvp^1(x)\sint\tS^n_t\\
&+\int_0^t\tS^n_ud\hvp^{1,c}_u(x)+\sum_{0<u\leq t}\tS^n_{u-}\Delta\hvp^1_u(x)+\sum_{0\leq u< t}\tS^n_{u}\Delta_+\hvp^1_u(x)\\
={}&x+\hvp^1(x)\sint \tS^n_t+A^n_t,
\end{align*}
where
\begin{align*}
A^n_t:={}&\int_0^t\big(\tS^n_u-S_u\big)d\hvp^{1,\uparrow,c}_u(x)+\int_0^t\big((1-\lambda)S_u-\tS^n_u\big)d\hvp^{1,\downarrow,c}_u(x)\\
&+\sum_{0<u\leq t}\big(\tS^n_{u-}-S_{u-}\big)\Delta\hvp^{1,\uparrow}_u(x)+\sum_{0<u\leq t}\big((1-\lambda)S_{u-}-\tS^n_{u-}\big)\Delta\hvp^{1,\downarrow}_u(x)\\
&+\sum_{0\leq u< t}\big(\tS^n_u-S_u\big)\Delta_+\hvp^{1,\uparrow}_u(x)+\sum_{0\leq u< t}\big((1-\lambda)S_u-\tS^n_u\big)\Delta_+\hvp^{1,\downarrow}_u(x)
\end{align*}
is a non-increasing predictable process. Combining this with the supermartingale property of $Z^{0,n}\big(x+\vp^1\sint\tS^n\big)$ we obtain
\begin{align*}
E[Z^{0,n}_T\hvp^0_T(x)]&=\textstyle E\left[Z^{0,n}_T\left(A^n_T+x+\hvp^1(x)\sint \tS^n_{T}\right)\right]\leq\textstyle E\left[Z^{0,n}_TA^n_T\right]+x.
\end{align*}
By Fatou's Lemma the latter implies that
\begin{align*}
x\widehat{y}(x)&=E\big[\widehat{Y}^0_{T}\big(\widehat{y}(x)\big)\hvp_{T}^0(x)\big]\leq \liminf_{n\to\infty}E[\widehat y(x)Z^{0,n}_T\hvp^0_{T}(x)]=\liminf_{n\to\infty}\textstyle E\left[\widehat y(x)Z^{0,n}_TA^n_T\right]+x\widehat y(x)
\end{align*}
and therefore that
\be
Z^{0,n}_TA^n_T\overset{L^1(P)}{\longrightarrow} 0,\label{eq:1}
\ee
as $Z^{0,n}_TA^n_T\leq 0$. Defining $B:=\Var(\hvp)$ and $P_B:=P\times B$ on $\Big(\Om\times[0,T],\F\otimes \cB\big([0,T]\big)\Big)$, there exists by \eqref{eq:1} a subsequence $\big((Z^{0,n},Z^{1,n})\big)^\infty_{n=1}$, again indexed by $n$, and an optional process $\tS^{1,\infty}$ and a predictable process $\tS^{0,\infty}$ such that $\tS^n \longrightarrow\tS^{1,\infty}$ $P_B$-a.e.~on $F_1:=\{d\hvp^{1,c}(x)\ne0\}\cup\{\Delta_+ \hvp^{1}(x)\ne0\}$, $\tS^n_- \longrightarrow\tS^{0,\infty}$ $P_B$-a.e.~on $F_0:=\{\Delta \hvp^{1}(x)\ne0\}$ and
\begin{align*}
\{d\hvp^{1,c}(x)>0\}&\subseteq \{\tS^{1,\infty}=S\}, & \{\Delta \hvp^{1}(x)>0\}&\subseteq \{\tS^{0,\infty}=S_-\},\\
\{\Delta_+ \hvp^{1}(x)>0\}&\subseteq \{\tS^{1,\infty}=S\}, & \{d\hvp^{1,c}(x)<0\}&\subseteq \{\tS^{0,\infty}=(1-\lambda)S\},\\
\{\Delta \hvp^{1}(x)<0\}&\subseteq \{\tS^{0,\infty}=(1-\lambda)S_-\}, & \{\Delta_+ \hvp^{1}(x)<0\}&\subseteq \{\tS^{1,\infty}=(1-\lambda)S\}.
\end{align*}
To complete the proof, we only need to show that $\tS^{1,\infty}$ and $\tS^{0,\infty}$ are indistinguishable from $\hS:=\frac{\hY^1}{\hY^0}$ and $\hS^p:=\frac{\hY^{1,p}}{\hY^{0,p}}$ on $F_0$ and $F_1$, respectively, which means that $P\big(\pi(G_1)\big)=0$ and $P\big(\pi(G_0)\big)=0$, where  $G_1:=\{\tS^{1,\infty}\ne\hS\}\cap F_1$, $G_0:=\{\tS^{0,\infty}\ne\hS^p\}\cap F_0$ and $\pi:\Om\times[0,T]\to\Om$ is given by $\pi\big((\om,t)\big)=\om$. For this, we argue by contradiction and suppose that $P\big(\pi(G_i)\big)>0$ for $i=0,1$. As $G_0$ and $G_1$ are optional, there exist $[0,T]\cup\{\infty\}$-valued stopping times $\sigma_0$ and $\sigma_1$ such that $\llbracket (\sigma_i)_{\{\sigma_i<T\}}\rrbracket\subseteq G_i$ and $P(\sigma_i<\infty)>0$ for $i=0,1$ by the optional cross-section theorem (see Theorems IV.84 in \cite{DM78}).
By the definition of the stopping times $\sigma_0$ and $\sigma_1$ we then have that $\tS^{1,\infty}_{\sigma_1}\ne\hS_{\sigma_1}$ and $\tS^{0,\infty}_{\sigma_0}\ne\hS^p_{\sigma_0}$ on $\{\sigma_1<\infty\}$ and $\{\sigma_0<\infty\}$, respectively. But this contradicts the convergence \eqref{eq:convCPS1} and \eqref{eq:convCPS2}, since
\begin{align*}
\tS^n_{\tau_1}&=\frac{Z^{1,n}_{\tau_1}}{Z^{0,n}_{\tau_1}}\overset{P}{\longrightarrow}\frac{\hY^1_{\tau_1}}{\hY^0_{\tau_1}}=\hS_{\tau_1},\\
\tS^n_{\tau_0-}&=\frac{Z^{1,n}_{\tau_0-}}{Z^{0,n}_{\tau_0-}}\overset{P}{\longrightarrow}\frac{\hY^{1,p}_{\tau_0}}{\hY^{0,p}_{\tau_0}}=\hS^p_{\tau_0}
\end{align*}
for the $[0,T]$-valued stopping times $\tau_1:=\sigma_1\wedge T$ and $\tau_0:=\sigma_0\wedge T$.
\ep
\begin{proof}[Proof of Theorem \ref{mt3}] 
Fix $\varphi$ with the properties as in Theorem \ref{mt3} and let $\widetilde{Z}^n$ be a sequence of local $\lambda$-consistent price systems satisfying \eqref{12.1} and \eqref{12.2} in Theorem \ref{mt2}. To alleviate notation we write $Z^n=\widetilde{Z}^n.$ 

Define the process $\tV^n$ by 
\begin{align}\label{R1}
\tV^n_t = Z^{0,n}_t \varphi^0_t + Z^{1,n}_t \varphi^1_t.
\end{align}

This process is non-negative by \eqref{2.32} and the fact that $\widetilde{S}^n=\frac{Z^{1,n}}{Z^{0,n}}$ takes its values in the bid-ask spread of $S$.

As $\varphi$ is of finite variation, we obtain by integration by parts that
$$d\tV^n_t=(\varphi^0_t dZ^{0,n}_t + \varphi^1_t dZ^{1,n}_t) + (Z^{0,n}_t d\varphi^0_t + Z^{1,n}_t d\varphi^1_t).$$

We decompose $\tV^n$ into $\tV^n=VG^n + VT^n,$ where
\begin{align}
VG^n_t& =x+  \int^t_0(\varphi^0_u dZ^{0,n}_u + \varphi^1_u dZ^{1,n}_u)\label{R2}\\
VT^n_t &= \int^t_0(Z^{0,n}_u d\varphi^0_u + Z^{1,n}_u d\varphi^1_u). \label{R2a}
\end{align}
The names indicate that $VG^n$ correspond to a value originating from the {\it gains} due to the movements of the local consistent price system $Z^n$, while $VT^n$ corresponds to the value originating from {\it trading}, i.e.~from switching between the positions $\varphi^0$ and $\varphi^1$ at price $\widetilde{S}^n = \frac{Z^{1,n}}{Z^{0,n}}$. If $\varphi$ were self-financing for $S$ under transaction costs $\lambda$ (see \eqref{eq:sf2.1}---\eqref{eq:sf2.3}), we could conclude that the process $VT^n$ is non-increasing. However, we can only use the weaker hypothesis that $\varphi$ is self-financing for $\widehat{\cS}$ without transaction costs \eqref{D1} which does not allow for this conclusion.

But here is a substitute.

\begin{claim}
For $\ve > 0$, there is a $[0,T]\cup \{\infty\}$-valued stopping time $\sigma$ with $P[\sigma < \infty] < \ve$ and a subsequence $(n_k)^\infty_{k=1}$ such that the stopped processes $VT^{n_k, \sigma}$ satisfy the uniform estimate $|VT^{n_k, \sigma}|\leq k^{-1}.$
\end{claim}

Indeed, by Lemma 7.1 of \cite{CS13} we know that 
\begin{align*}
\lim_{n \to \infty} VT^n_t ={}& \lim_{n \to \infty} \int^t_0 (Z^{0,n}_u d\varphi^0_u + Z^{1,n}_u d\varphi^1_u )\\
={}&\int^t_0 \hY^0_ud\vp^{0,c}_u+ \sum_{0 \leq u <t} \hY^{0,p}_u\Delta \vp^0_u+ \sum_{0 < u \leq t} \hY^0_u \Delta_+ \vp^0_u\\
&+\int^t_0 \hY^1_ud\vp^{1,c}_u+ \sum_{0 \leq u <t} \hY^{1,p}_u\Delta \vp^1_u+ \sum_{0 < u \leq t} \hY^1_u \Delta_+ \vp^1_u\\
=&:\int^t_0 (\widehat{\cY}^0_u d\varphi^0_u + \widehat{\cY}^1_u d\varphi^1_u)
\end{align*}
the limit holding true in probability, uniformly in $t \in [0,T]$ (u.c.p.~topology).

The last process is identically equal to zero as $\varphi$ is self-financing for $\widehat{\mathcal{S}}$ (see \eqref{D1}). For $k \in \mathbb{N},$ let 
$$\sigma_{k,n}= \inf\{t:|VT^n_t| \geq k^{-1}\}.$$
Choose $n_k$ large enough so that we have 
$$P[\sigma_{k,n} < \infty] < \ve 2^{-k}.$$
We still have to control the possible final jump of $VT^n$ at $\{\sigma_{k,n} < \infty\}.$ To do so, note that $VT^n$ is a predictable process so that $\sigma_{k,n}$ is a predictable stopping time. We therefore may find an announcing sequence $(\sigma_{k, n_k, j})$ of stopping times, i.e. $\sigma_{k,n_k,j} < \sigma_{k, n_k}$ on $\{ \sigma_{k, n_k} < \infty\}$ and $(\sigma_{k, n_k, j})^\infty_{j=1}$ increases a.s.~to $\sigma_{k, n_k}$. Letting $\sigma_k = \sigma_{k, n_k,j}$ for large enough $j$, we have
$$P[\sigma_k < \infty] < \ve 2^{-k} \quad \mbox{and} \quad |VT^{n_k, \sigma_k}| \leq k^{-1}.$$

Defining $\sigma$ as the infimum of $(\sigma_k)^\infty_{k=1},$ we have proved the claim.
\vskip10pt
Now we turn to the processes $(VG^n)^\infty_{n=1}$ in \eqref{R2} which are local martingales. By the above claim and Proposition 2.12 of \cite{CS13} we conclude that 
\begin{align}
P-\lim_{n \to \infty}VG^n_\tau &= x+ P-\lim_{n \to \infty}\int^\tau_0 (\varphi^0_u dZ^{0,n}_u + \varphi^1_u dZ^{1,n}_u)\\
&=x+\varphi^0\sint \widehat{\cY}^0_\tau + \varphi^1\sint\widehat{\cY}^1_\tau ,
\end{align}
for all $[0,\sigma \wedge T]$-valued stopping times $\tau.$

In particular, the stopped process $(x+\varphi^0\sint \widehat{\cY}^0_t + \varphi^1\sint\widehat{\cY}^1_t)^\sigma_{0 \leq t \leq T}$ equals the stopped process $(\widehat{Y}^0_t\varphi^0_t  + \widehat{Y}^1_t\varphi^1_t)^\sigma_{0 \leq t \leq T}$ and is a non-negative supermartingale. As $\ve > 0$ in the above claim was arbitrary we may conclude that $(\varphi^0_t \widehat{Y}^0_t + \varphi^1_t \widehat{Y}^1_t)_{0 \leq t \leq T}$ is a non-negative supermartingale.

Fo the proof of \eqref{P1.2} we observe that
$$E\big[U\big(x+\varphi^1\sint \widehat{\cS}_T\big)\big] \leq E\big[V(\hY^0_T)+\hY^0_T(x+\varphi^1\sint \widehat{\cS}_T)\big]\leq E\big[V(\hY^0_T)+\hY^0_T(x+\hvp^1\sint \widehat{\cS}_T)\big]$$
by Fenchel's inequality, the supermartingale and the martingale property of $\hY^0(x+\varphi^1\sint \widehat{\cS}_T)$ and $\hY^0(x+\hvp^1\sint \widehat{\cS}_T)$, respectively. Combining this with $\hY^0_T=U'(x+\hvp^1\sint \widehat{\cS})$ and the fact that $V(y)=U\big((U')^{-1}(y)\big)-(U')^{-1}(y)y$ for $y>0$ we obtain
$$E\big[U\big(x+\varphi^1\sint \widehat{\cS}_T\big)\big]\leq E\big[U\big(x+\hvp^1\sint \widehat{\cS}_T\big)\big]= E\big[U\big(\hvp^0_T + \hvp^1_T \widehat{S}_T \big)\big] = E\big[U\big(V_{T}^{liq}(\hvp)\big)\big],$$
which completes the proof.
\end{proof}
\bp[Proof of Proposition \ref{lem:martingale}]
Suppose that $(\hY^0,\hY^1)\in\cB\big(\widehat y(x)\big)$ is a local martingale and hence c\`adl\`ag. Then the process $(\hY^{0,p},\hY^{1,p})$ coincides with $(\hY^0,\hY^1)$ as explained below \eqref{D5.1} and the integral $x+\hvp^1\sint\widehat{\mathcal{S}}$ reduces to the usual stochastic integral $x+\hvp^1\sint \hS$ with $\hS:=\frac{\hY^1}{\hY^0}$. Moreover, the process $\hY^0\vp^0+ \hY^1\vp^1= \hY^0\big(x+\vp^1\sint\hS\big)$ is a non-negative local martingale and hence a supermartingale for all $(\vp^0,\vp^1)\in\cA\big(x;\hS\big)$, which implies that $\hY^0$ is an equivalent local martingale deflator for $\hS$ without transaction costs starting at $\hY^0_0=\widehat y(x)$ and hence $\widehat{Y}^0 \in \mathcal{Y} (\hat{y}(x); \widehat{S})$. As $\hY^0_{T}=U'\big(V_{T}^{liq}(\hvp)\big)$ and $\hY^0\hvp^0+ \hY^1\hvp^1= \hY^0\big(x+\hvp^1\sint\hS\big)$ is a martingale by Theorem \ref{mainthm}, we obtain by the duality for the frictionless utility maximisation problem, i.e., Theorem 2.2 in \cite{KS99}, that $(\hvp^0,\hvp^1)\in\cA\big(x;\hS\big)$ and $\hY^0\in\cY\big(\widehat y(x);\hS\big)$ are the solutions to the frictionless primal and dual problem for $\hS$, if $\widehat y(x;\hS)=\widehat y(x)$. To see the latter, we observe that $u(x)=v\big(\widehat y(x)\big)+x\widehat y(x)$ by Theorem \ref{mainthm} and therefore
$$v\big(\widehat y(x)\big)+x\widehat y(x)=u(x)\leq u\big(x;\hS\big)\leq u\big(\widehat y(x);\hS\big)+x\widehat y(x)$$
by \eqref{sec1:eq1}. Since $v\big(\widehat y(x)\big)=E\big[V(\hY^0_T)\big]$, $E[V_{T}^{liq}(\hvp)\hY^0_T]=x\widehat y(x)$ and $\hY^0\in\cY\big(\widehat y(x);\hS\big)$, we obtain that $\widehat y(x;\hS)=\widehat y(x)$, which completes the proof.
\ep 

\bp[Proof of Proposition \ref{lem:connection}]
By Theorem 2.1 in \cite{K10} the assumption that the shadow price $\hS=(\hS)_{0\leq t\leq T}$ satisfies $(NUPBR)$ implies that $\hS$ admits an equivalent local martingale deflator. As $\hS$ is valued in the bid-ask spread $[(1-\lambda)S,S]$, any $(\vp^0,\vp^1)\in\cA(x)$ is also self-financing and admissible for $\hS$ without frictions ($\lambda=0$) and hence $\cA(x)\subseteq\cA(x;\hS)$. Since $\cA(x)\subseteq\cA(x;\hS)$, we obtain that $(Y^0,Y^1):=(Y,Y\hS)\in\cB(y)$ for all $Y\in\cY(y;\hS)$ and therefore similarly to \eqref{sec1:eq1}:
$$
v(y)=\inf_{(Y^0,Y^1)\in\cB(y)}E[V(Y^0_T)]\leq \inf_{Y\in\cY(y;\hS)}E[V(Y_T)]=: v(y;\hS).
$$
Moreover, as
$$u(x) = v\big(\widehat y(x)\big)+x\widehat y(x)\leq v\big(\widehat y(x;\hS)\big)+x\widehat y(x;\hS) \leq v\big(\widehat y(x;\hS);\hS\big)+x\widehat y(x;\hS) =u(x;\hS)=u(x),$$
it follows that $\widehat y(x)=\widehat y(x;\hS)$ and therefore that $(\widehat{Y}^0,\widehat{Y}^1):=(\hY,\hY\hS)\in\cB\big(\widehat y(x)\big)$ is the solution to the frictional dual problem \eqref{D1}, where $\hY\in\cY\big(\widehat y(x;\hS);\hS\big)$ is the solution to its frictionless counterpart
$$E[V(Y_T)]\to\min!,\quad Y\in\cY(\widehat y(x);\hS),$$
for $\widehat{S}$.
\ep 

\section{A more detailed analysis of the examples}\label{App:A}
After the formal discussion of the examples in Section \ref{sec:ex} let us now give in the next two subsections a more detailed analysis.
\subsection{Truly l\`adl\`ag primal and dual optimisers}\label{A:Ex1}
We begin with the first example discussed in Section \ref{sec:ex1}.
\begin{prop}\label{A:Ex1:prop1}
Let $\varepsilon \in (0, \frac{1}{3})$ and fix $\lambda\in(0,1)$ sufficiently small. Then:

{\bf{1)}} The solution $\hvp=(\hvp^0, \hvp^1) \in \cA(1)$ to the problem
\be
E\big[\log\big(V_1^{liq}(\vp)\big)\big] \to \max!, \quad \vp\in \cA(1),\label{A:Ex1:PP}
\ee
for the ask price $S=(S_t)_{0 \leq t \leq 1}$ defined in \eqref{ex1:pp} exists and is given by
$$\hvp^1_t=\hvp^1_{\frac{1}{2}}\mathbbm{1}_{\rrbracket 0, \frac{1}{2} \rrbracket} + \hvp^1_{\frac{1}{2}+} \mathbbm{1}_{\rrbracket \frac{1}{2},1 \rrbracket} + \int^t_{\frac{1}{2}\wedge t} d\hvp^{1,c}_s,$$
where 
\begin{align*}
\hvp^1_{\frac{1}{2}}={}& \frac{4-\lambda}{1+\lambda},\\
\hvp^1_{\frac{1}{2}+}={}&\frac{1+\hvp^1_{\frac{1}{2}} \Delta S_{\frac{1}{2}}}{S_{\frac{1}{2}} } \frac{1}{\lambda +(1-\lambda)a_{\frac{1}{2}}}\mathbbm{1}_{\{ \Delta S_{\frac{1}{2}}>0\}}+\frac{1+\hvp^1_{\frac{1}{2}} \big((1-\lambda)S_{\frac{1}{2}}-S_{\frac{1}{2}-}\big)}{(1-\lambda)S_{\frac{1}{2}} }\frac{1}{a_{\frac{1}{2}}}\mathbbm{1}_{\{ \Delta S_{\frac{1}{2}}<0\}},\\
d\hvp^{1,c}_t={}&\mathbbm{1}_{\rrbracket \frac{1}{2},\sigma \rrbracket} \frac{\hvp^0_{\frac{1}{2}+}+\hvp^1_{\frac{1}{2}+} S_{\frac{1}{2}}}{S_{\frac{1}{2}} }\frac{1-\lambda}{\big(\lambda+(1-\lambda) a_t\big)^2} \frac{1}{3}dt
\end{align*}
and $\hvp^0_0=1$ and $d\hvp^0$ is determined by the self-financing condition \eqref{eq:sf2.1} - \eqref{eq:sf2.3} with equality. 

{\bf{2)}} The solution $\widehat{Y} = (\widehat{Y}^0, \widehat{Y}^1)$ to the dual problem
\be
E[-\log(Y^0_T)-1] \to \min!, \quad Y=(Y^0, Y^1) \in \cB\big(\hy (x) \big),\label{A:Ex1:DP}
\ee
for $\hy (x) = u'(x)=\frac{1}{x}=1$ exists and is given by
\be
(\hY^0,\hY^1)=\left(\frac{1}{\hvp^0 + \hvp^1\hS},\frac{\widehat{S}}{\hvp^0 + \hvp^1\widehat{S}}\right),\label{A:Ex1:DO}
\ee
where
\be
\hS_t=\begin{cases}S_0&: 0\leq t<\frac{1}{2},\\
 S_{\frac{1}{2}}\mathbbm{1}_{\{ \Delta S_{\frac{1}{2}}>0\}}+(1-\lambda)S_{\frac{1}{2}} \mathbbm{1}_{\{\Delta S_{\frac{1}{2}}<0\}}&: t=\frac{1}{2},\\
 S_{\frac{1}{2}}&:\frac{1}{2}<t<\sigma,\\
 (1-\lambda)S_{\sigma}&:\sigma\leq t\leq 1
 \end{cases}
\label{A:Ex1:hS}
\ee
and
\begin{align}
\hvp^0_t + \hvp^1_t\hS_t=1+{}&\hvp^1\sint \hS_t\nonumber\\
=1+{}&\hvp^1_{\frac{1}{2}}\Delta S_{\frac{1}{2}}\mathbbm{1}_{\{\Delta S_\frac{1}{2}>0\}}\mathbbm{1}_{\{\frac{1}{2}\leq t\}}+\hvp^1_{\frac{1}{2}}\big( (1-\lambda ) S_{\frac{1}{2}} - S_{\frac{1}{2}-}\big)\mathbbm{1}_{\{\Delta S_\frac{1}{2}<0\}}\mathbbm{1}_{\{\frac{1}{2}\leq t\}}\nonumber\\
+{}&\hvp^1_{\frac{1}{2}}\big( S_{\frac{1}{2}} - (1-\lambda ) S_{\frac{1}{2}}\big)\mathbbm{1}_{\{\Delta S_\frac{1}{2}<0\}}\mathbbm{1}_{\{\frac{1}{2}< t\}}+\hvp^1_{\sigma}\big( (1-\lambda )S_{\sigma} -  S_{\sigma-}\big)\mathbbm{1}_{\{\sigma\leq t\}} \label{A:Ex1:dw}
\end{align}
for $t\in[0,1]$.
\end{prop}
\begin{proof}
1) Since trading for any price within the bid-ask spread is always more favourable than trading under transaction costs (see \eqref{D9}), we have that
\begin{align}
V_1^{liq}(\vp)&\leq 1 + \varphi^1_{\frac{1}{2}}\Big(\big( (1-\lambda ) S_{\frac{1}{2}}-S_{\frac{1}{2}-}\big) \mathbbm{1}_{\{\Delta S_{\frac{1}{2}}<0\}} +  \Delta S_{\frac{1}{2}} \mathbbm{1}_{\{\Delta S_{\frac{1}{2}}> 0\}}\Big)\nonumber\\
&\qquad+ \varphi^1_{\frac{1}{2}}\big(S_{\frac{1}{2}}-(1-\lambda) S_{\frac{1}{2}} \big)\mathbbm{1}_{\{ \Delta S_{\frac{1}{2}} < 0\}}+ \varphi^1_\sigma \big( (1-\lambda ) S_\sigma - S_{\frac{1}{2}}\big)\nonumber\\
&= \Big(1+\ell_{\frac{1}{2}}\Big( \big((1-\lambda ) S_{\frac{1}{2}} - S_{\frac{1}{2}}\big) \mathbbm{1}_{\{\Delta S_{\frac{1}{2}} < 0\}} + \Delta S_{\frac{1}{2}} \mathbbm{1}_{\{\Delta S_{\frac{1}{2}} < 0\}}\Big)\nonumber \\
&\qquad\times\big(1+ \ell_{\frac{1}{2}+}\lambda  \mathbbm{1}_{\{\Delta S_{\frac{1}{2}}< 0\} } \big)\Big(1+ \ell_\sigma\big( (1-\lambda ) (1+ a_\sigma(\eta -1)) - 1\big)\Big),\label{A:Ex1:eq1}
\end{align}
where
\begin{align}
\ell_{\frac{1}{2}} &= \varphi^1_{\frac{1}{2}},\nonumber\\
\ell_{\frac{1}{2}+}&= \frac{\varphi^1_{\frac{1}{2}_+} S_{\frac{1}{2}}}{1+ \varphi^1_{\frac{1}{2}}\Big(\big((1-\lambda)S_{\frac{1}{2}}-S_{\frac{1}{2}-}\big) \mathbbm{1}_{\{\Delta S_{\frac{1}{2}} < 0\}} + \Delta S_{\frac{1}{2}} \mathbbm{1}_{\{\Delta S_{\frac{1}{2}} > 0\}}\Big)},\nonumber\\
\ell_t&= \frac{\varphi^1_t S_{\frac{1}{2}}}{1+ \varphi^1_{\frac{1}{2}}\Big(\big((1-\lambda)S_{\frac{1}{2}}-S_{\frac{1}{2}-}\big) \mathbbm{1}_{\{\Delta S_{\frac{1}{2}} < 0\}} +\Delta S_{\frac{1}{2}} \mathbbm{1}_{\{\Delta S_{\frac{1}{2}} > 0\} }\Big)+ \varphi^1_{\frac{1}{2}+} \lambda S_{\frac{1}{2}}\mathbbm{1}_{\{\Delta S_{\frac{1}{2}} < 0\} }},\label{A:Ex1:eq1a}
\end{align}
for $t \in ({\frac{1}{2}},1].$ By the scaling of the logarithm this implies that
\begin{align}\label{A:Ex1:eq2}
E\big[\log\big(V_1^{liq}(\vp)\big)\big] \leq{}& E\Big[\log \Big(1+\ell_{\frac{1}{2}}\Big(\big( (1-\lambda ) S_{\frac{1}{2}} - S_{\frac{1}{2}-}\big) \mathbbm{1}_{\{\Delta S_{\frac{1}{2}}<0 \}}+ \Delta S_{\frac{1}{2}} \mathbbm{1}_{\{\Delta S_{\frac{1}{2}} > 0 \}}\Big)\Big)\Big]\nonumber\\
&+ E [\log (1+\ell_{\frac{1}{2}+} \lambda ) ]+ E\Big[\log\Big(1+ \ell_\sigma\Big((1-\lambda)\big(1+a_\sigma(\eta-1)\big)-1\Big)\Big) \Big].
\end{align}
The basic idea to derive the optimal trading strategy $\hvp=(\hvp^0_t, \hvp^1_t)_{0 \leq t \leq 1}$ for \eqref{A:Ex1:PP} is now to maximise the right hand side of \eqref{A:Ex1:eq2} over all predictable processes $\ell = (\ell_t)_{0 \leq t \leq 1}$ and to show that solving \eqref{A:Ex1:eq1a} allows us to define a self-financing trading strategy under transaction costs such that we have equality in \eqref{A:Ex1:eq1}. For this, we observe that we can maximise the terms on the right hand side of \eqref{A:Ex1:eq2} independently of each other and only need to solve
\begin{align}
&E\Big[\log \Big(1+\ell_{\frac{1}{2}}\Big(\big( (1-\lambda ) S_{\frac{1}{2}} - S_{\frac{1}{2}-}\big) \mathbbm{1}_{\{\Delta S_{\frac{1}{2}}<0 \}}+ \Delta S_{\frac{1}{2}} \mathbbm{1}_{\{\Delta S_{\frac{1}{2}} > 0 \}}\Big)\Big)\Big]\to\max!,\quad \ell_{\frac{1}{2}}\in\cF_{\frac{1}{2}-},\label{A:Ex1:P1}\\
&E\Big[\log\Big(1+ \ell_\sigma\Big((1-\lambda)\big(1+a_\sigma(\eta-1)\big)-1\Big)\Big) \Big] \to \max!,\quad \text{  predictable $ (\ell_t)_{\frac{1}{2} <t \leq 1}$.}\label{A:Ex1:P2}
\end{align}
We show below that the solution to \eqref{A:Ex1:P1} and \eqref{A:Ex1:P2} are given by
\begin{align}
\widehat{\ell}_{\frac{1}{2}} &= \frac{4-\lambda}{1+\lambda},\\
\widehat{\ell}_t&= \frac{1}{\lambda + (1-\lambda ) a_t} \quad \mbox{for $t \in (\frac{1}{2}, 1].$}
\end{align}
This will then also imply the optimal value for 
\begin{align}
E[\log(1 + \ell_{\frac{1}{2}+} \lambda) ] \to \mbox{max!}
\end{align}
when we maximise over all possible limits $\ell_{\frac{1}{2}+}= \lim_{t \searrow \frac{1}{2}} \ell_t$ of predictable processes $(\ell_t)_{\frac{1}{2} < t \leq 1}$ for which the problem \eqref{A:Ex1:P2} is well-defined $>-\infty$, i.e. $\ell_t \leq \frac{1}{\lambda + (1-\lambda)a_t}$ for $t \in (\frac{1}{2},1]$ and therefore $\widehat{\ell}_{\frac{1}{2}+} = \frac{1}{\lambda+(1-\lambda) a_{\frac{1}{2}}}.$

We first illustrate how to obtain the solution to \eqref{A:Ex1:P2}. To that end, we observe that
$$h(\ell, t):= E\Big[\log \Big(1+ \ell \Big((1-\lambda)\big(1+a_t (\eta-1)\big) -1\Big)\Big)\Big]$$
is given by
$$h(\ell, t)= \log \Big(1+ \ell \big((1-\lambda)(1+a_t) -1\big)\Big)(1-\ve) + \sum^\infty_{n=1} \textstyle\log \Big(1+ \ell\Big((1-\lambda)\big(1+a_t (\frac{1}{n}-1)\big)-1\Big)\Big) \ve 2^{-n}$$
and its derivative $\frac{\partial h}{\partial \ell}(\ell, t)$ by 
\begin{align}
\frac{\partial h}{\partial \ell} (\ell, t)&= \frac{(1-\lambda) a_t-\lambda}{1+\ell\big((1-\lambda) a_t-\lambda\big)} (1-\ve) + \sum^\infty_{n=1} \frac{(1-\lambda ) a_t(\frac{1}{n}-1) - \lambda}{1 + \ell\big((1-\lambda) a_t(\frac{1}{n}-1)-\lambda\big)} \ve 2^{-n}.
\end{align}
As 
\begin{align}
\frac{\partial h}{\partial \ell} (\widehat{\ell}_t, t)={}&\frac{\big((1-\lambda)a_t\big)^2-\lambda^2}{2(1-\lambda )a_t} (1-\varepsilon)\nonumber\\
& - \sum^\infty_{n=1} \frac{\big((1-\lambda )a_t + \lambda \big)^2}{(1-\lambda )a_t} \varepsilon n2^{-n} + \sum^\infty_{n=1}\big((1-\lambda )a_t + \lambda \big) \varepsilon 2^{-n}\nonumber\\
={} & \frac{\big((1-\lambda)a_t\big)^2 - \lambda^2}{2(1-\lambda)a_t} (1-\varepsilon)\nonumber\\
& - \frac{\big((1-\lambda)a_t + \lambda\big)^2}{(1-\lambda)a_t} \frac{1}{2} \frac{1}{(1-\frac{1}{2})^2} \varepsilon+ \big((1-\lambda)a_t + \lambda\big) \varepsilon > 0  \label{B.1}
\end{align}
for $\widehat{\ell}_t=\frac{1}{\lambda + (1-\lambda)a_t}$ and $\varepsilon \in (0,\frac{1}{3})$ and $\lambda \in (0,1)$ sufficiently small, the concave function $\ell \mapsto h(\ell, t)$ is maximised over its domain $(- \frac{1}{(1+\lambda)a_t}, \widehat{\ell}_t]$ by $\widehat{\ell}_t.$

The solution to \eqref{A:Ex1:P1} is simply obtained by solving the first order condition $f'(\widehat{\ell}_{\frac{1}{2}})=0$ for 
\begin{align*}
f(\ell)&=E \left[\log \Big(1+ \ell\big((1-\lambda)S_{\frac{1}{2}}-S_{\frac{1}{2}-}\big)\Big) \mathbbm{1}_{\{\Delta S_{\frac{1}{2}} < 0\}} + \log (1+\ell \Delta S_{\frac{1}{2}}) \mathbbm{1}_{\{\Delta S_{\frac{1}{2}} > 0\}}\right]\\
&=\textstyle \log \Big(1+ \ell \big((1-\lambda ) \frac{1}{2} -1\big)\Big) \frac{1}{6} + \log (1+\ell 2) \frac{5}{6}
\end{align*}
and
\begin{align}\label{A.8}
f'(\ell) = \frac{(1-\lambda) \frac{1}{2}-1}{1+\ell((1-\lambda)\frac{1}{2}-1)} \frac{1}{6} + \frac{2}{1+ \ell 2} \frac{5}{6},
\end{align}
which gives $\widehat{\ell}_{\frac{1}{2}}= \frac{4-\lambda}{1+\lambda}.$

To obtain the optimal strategy $\hvp=(\hvp^0_t, \hvp^1_t)_{0 \leq t \leq T}$ to \eqref{A:Ex1:PP}, we only need to observe that solving \eqref{A:Ex1:eq1a} gives a self-financing and admissible trading strategy under transaction costs, as
\begin{align*}
\Delta_+\hvp^1_0 &= \frac{4 - \lambda}{1+\lambda} >0,\\
\Delta_+\hvp^1_{\frac{1}{2}} &= \frac{1 + \hvp^1_{\frac{1}{2}} \Delta S_{\frac{1}{2}}}{S_{\frac{1}{2}}} \frac{1} {\lambda + (1-\lambda)a_{\frac{1}{2}}} -\frac{4 - \lambda}{1+\lambda}> 0 \quad \mbox{on $\{\Delta S_{\frac{1}{2}}>0\},$}\\
\Delta_+\hvp^1_{\frac{1}{2}} &= \frac{1 + \hvp^1_{\frac{1}{2}} ((1-\lambda ) S_{\frac{1}{2}} - S_{\frac{1}{2}-})} {(1-\lambda)S_{\frac{1}{2}}} \frac{1}{a_{\frac{1}{2}}} -\frac{4 - \lambda}{1+\lambda}< 0 \quad \mbox{on $\{ \Delta S_{\frac{1}{2}} <0\},$}\\
d \hvp^{1,2}_t &= \mathbbm{1}_{\rrbracket \frac{1}{2}, \sigma \rrbracket} \frac{\hvp^0_{\frac{1}{2}+} + \hvp^1_{\frac{1}{2}} S_{\frac{1}{2}}} {S_{\frac{1}{2}}} \frac{1-\lambda}{(\lambda +(1-\lambda )a_t )^2} \frac{1}{3} dt
\end{align*}
and $\hvp^0$ can be defined by the self-financing conditions \eqref{eq:sf2.1} -- \eqref{eq:sf2.3} with equality.

2) That the solution to the dual problem \eqref{A:Ex1:DP} is given by \eqref{A:Ex1:DO} follows immediately from Proposition \ref{prop:log} and the formulas \eqref{A:Ex1:hS} and \eqref{A:Ex1:dw} from \eqref{mt2:eq2}. 
\end{proof}

Let us now come to the approximation of the dual optimiser by consistent price systems. For this, it is more convenient to think of the consistent prices systems $Z^n=(Z^{0,n}, Z^{1,n})$ as pairs $(Q^n, \widetilde{S}^n)$ of processes $\widetilde{S}^n=(\widetilde{S}^n_t)_{0 \leq t \leq 1}$ evolving in the bid-ask spread and equivalent martingale measures $Q^n$ for those.

Since $\widehat{Y}=(\widehat{Y}^0, \widehat{Y}^1)$ is a martingale on $[0, \frac{1}{2}]$ by the first order condition \eqref{A.8}, we can simply set
\begin{align*}
Z_t^{0,n} &= \widehat{Y}^0_t, \qquad t \in[0, \textstyle\frac{1}{2}],\\
Z_t^{1,n} &= \widehat{Y}^1_t, \qquad t \in[0, \textstyle\frac{1}{2}],
\end{align*}
for all $n$ or, equivalently,
\begin{align*}
\frac{dQ^n}{dP}\Big|_{\mathcal{F}_{\frac{1}{2}}}&=\widehat{Y}^0_{\frac{1}{2}},\\
\widetilde{S}^n_t&=S_0,\quad t \in[0, \textstyle\frac{1}{2}),\\
\widetilde{S}^n_{\frac{1}{2}}&=\begin{cases} S_{\frac{1}{2}}&: \Delta S_{\frac{1}{2}}>0,\\
(1-\lambda) S_{\frac{1}{2}}&: \Delta S_{\frac{1}{2}}<0.
\end{cases}
\end{align*}

On $\{\Delta S_{\frac{1}{2}}<0\}$ we can extend the probability measures $Q^n$ to measure $\widetilde{Q}^n\sim P$ such that $\tau\sim \exp(n)$ is exponentially distributed, the expectations $E_{\tQ^n} [\eta - 1] =:b_n$ decrease to $-1$, i.e.~$E_{\tQ^n} [\eta - 1] =b_n\searrow -1$, and $\tau$ and $\eta$ remain independent under $\tQ^n$. Indeed, let $(\widetilde{\eta}^n)^\infty_{n=1}$ be a sequence of strictly positive $\sigma(\eta)$-measurable random variables such that $E[\widetilde{\eta}^n]=1$ and $E[\widetilde{\eta}^n (\eta-1)]=b_n \searrow -1.$ Then
$$\frac{d\widetilde{Q}^n}{dP}=\widehat{Y}^0_{\frac{1}{2}} \exp\big(-(n-1)\tau\big) \widetilde{\eta}^n$$
is the Radon-Nikodym derivative of a probability measure $\widetilde{Q}^n\sim P$ such that $E_{\widetilde{Q}^n} [\eta-1]=b_n$ and $\tau \sim \exp(n)$ under $\widetilde{Q}^n.$ The density process $\widetilde{Z}^n$ of $\widetilde{Q}^n$ is given by
\be
\widetilde{Z}^n_t=\widehat{Y}^0_{\frac{1}{2}\wedge t}\textstyle\exp\big(-(n-1)(\sigma \wedge t-\frac{1}{2})^+\big)\big(1+(n \widetilde{\eta}^n -1) \mathbbm{1}_{\llbracket \sigma, 1 \rrbracket}(t)\big),\quad 0\leq t\leq1.\label{A:Ex1:app1}
\ee
Therefore the expectation of the jump $E_{\tQ^n}[(1-\lambda) \Delta S_\sigma]$ of the bid price $(1-\lambda)S$ at time $\sigma$ under $\tQ^n$ is strictly negative.

Since the stopping time $\sigma$ remains totally inaccessible on $(\frac{1}{2},1)$ under $\tQ^n$, the compensator $A^n_t$ of the bid price $(1-\lambda)S$ under $\tQ^n$ is a continuously decreasing process
$$A^n_t:=\int^{t \wedge \sigma}_{\frac{1}{2}} (1-\lambda)S_{\frac{1}{2}} a_sE_{\tQ^n} [\eta-1]n\, ds=(1-\lambda)S_{\frac{1}{2}}\int^{t \wedge \sigma}_{\frac{1}{2}} a_s b_n n \, ds,\quad t \in [\textstyle\frac{1}{2},1).$$
Therefore the $\widetilde{Q}^n$-martingale
$$M^n_t:=(1-\lambda)S_t - A^n_t, \quad t \in \textstyle[\frac{1}{2},1),$$
is continuously increasing, if there is no jump, and we need to stop it at
$$\sigma_n:= \inf \textstyle\{t > \frac{1}{2}~|~M^n_t >S_\frac{1}{2}\}$$
to keep it in the bid-ask spread $[(1-\lambda)S,S].$

As $n$ increases, the martingales $M^n$ increase steeper and steeper, if there is no jump, so that the stopping times $\sigma_n$ converge $P$-a.s.~to $\frac{1}{2}$ and the right jump
$$\Delta_+\widehat{S}_{\frac{1}{2}}=\lim_{n\to \infty}(M^n_{\sigma_n}-M^n_{\frac{1}{2}})=\lambda S_{\frac{1}{2}}$$
arises as the limit of the continuous compensators $A^n$. 

As the probability that $\sigma$ is very close to $\frac{1}{2}$ under $\tQ^n$ increases with $n$, we obtain that 
$$\lim_{n \to \infty} \tQ^n \left[\textstyle\frac{1}{2} \leq \sigma \leq \sigma_n~\Big|~\cF_{\frac{1}{2}}\right]=: c > 0 \quad \mbox{on}\quad \{\Delta S_{\frac{1}{2}} < 0\}.$$
But, since $\lim_{n \to \infty} P[\frac{1}{2} \leq \sigma \leq \sigma_n] =0$, this implies that the measures $\tQ^n$ loose the mass $c$ on the sets $\{\frac{1}{2} \leq \sigma \leq \sigma_n\}$, which causes a right jump of the limit of the density processes $\tZ^{0,n}$ of the $\tQ^n$, i.e.
$$\lim_{n \to \infty} (\tZ^{0,n}_{\sigma_n} - \tZ^{0,n}_{\frac{1}{2}})=-c < 0 \quad \mbox{on}\quad \{\Delta S_{\frac{1}{2}} < 0\}.$$
However, comparing $c$ with $\Delta_+ \widehat{Y}^0_{\frac{1}{2}}$ we obtain that 
$$\Delta_+ \widehat{Y}^0_{\frac{1}{2}} > -c \quad \mbox{on} \quad \{\Delta S_{\frac{1}{2}} < 0\}.$$
The reason for this is that the martingale $M^n$ does not jump to the bid price at time $\sigma$, but we rather have
$$M^n_\sigma =(1-\lambda) S_\sigma- A^n_\sigma > (1 -\lambda) S_\sigma.$$
In order to adjust for this, we need to modify $M^n$ to obtain a $\tQ^n$-martingale $\widetilde{M}^n$ such that 
$$\widetilde{M}^n_\sigma=(1-\lambda)S_\sigma.$$
This results in choosing the jump of $\widetilde{M}^n$ such that 
$$\Delta \widetilde{M}^n_\sigma =(1-\lambda) \Delta S_\sigma + \widetilde{A}^n_\sigma < (1-\lambda) \Delta S_\sigma$$
and therefore gives a steeper (than $A^n$) decreasing compensator $\widetilde{A}^n$, where $\widetilde{M}^n$ and $\widetilde{A}^n$ are both implicitly related by
\be
\widetilde{A}^n_t=\int^{t \wedge \sigma}_{\frac{1}{2}} E_{\tQ^n} [\Delta \widetilde{M}^n_\sigma | \mathcal{F}_{\sigma-}] n\, ds = \int^{t \wedge \sigma}_{\frac{1}{2}}\big((1-\lambda)S_{\frac{1}{2}} a_s b_n + \widetilde{A}^n_s\big)n\, ds.\label{Ex1:int}
\ee
As $\widetilde{A}^n$ is decreasing steeper than $A^n$ on $\{\sigma_n < \sigma\}$, we obtain that the stopping times
\begin{align}\label{A:Ex1:st}
\tilde{\sigma}_n:= \textstyle\inf \{ t > \frac{1}{2}~|~\widetilde{M}^n_t > S_{\frac{1}{2}} \}
\end{align}
decrease faster to $\frac{1}{2}$ than $\sigma_n$ and therefore the measures $\tQ^n$ loose less mass on the sets $\{ \frac{1}{2} \leq \sigma \leq \tilde{\sigma}_n \}$ so that
$$\Delta_+ \widehat{Y}^0_{\frac{1}{2}} = -\lim_{n \to \infty} \tQ^n \textstyle\left[\frac{1}{2} \leq \sigma \leq \tilde{\sigma}_n~\Big|~\cF_{\frac{1}{2}}\right].$$
To show the existence of the $\widetilde{Q}^n$-martingales $(\widetilde{M}^n_t)_{t \in [\frac{1}{2},1]}$, we only need to observe that 
$$\widetilde{A}^n_t=\textstyle(1-\lambda) S_{\frac{1}{2}}\left( \frac{n-1}{n} a_{\frac{1}{2}}b_n + \left(\int^t_{\frac{1}{2}} na_s b_n \exp\big(-n(s-\frac{1}{2})\big) ds - \frac{n-1}{n} a_{\frac{1}{2}}b_n\right) \exp\big(n(t-\frac{1}{2})\big)\right),$$
for $t \in [\frac{1}{2},1]$, is a solution to the integral equation \eqref{Ex1:int} satisfying the boundary condition $\widetilde{A}^n_{\frac{1}{2}}=0$ and therefore setting
\be\label{A:Ex1:app2}
\widetilde{M}_t^n= 
\begin{cases}
\widehat{S}_t&: 0 \leq t \leq \frac{1}{2},\\
(1-\lambda) S_{\frac{1}{2}} - \widetilde{A}^n_t &: \frac{1}{2} < t < \sigma,  \\
(1-\lambda) S_\sigma&: \sigma\leq t\leq 1
\end{cases}
\ee
gives the desired $\widetilde{Q}^n$-martingale.

Moreover, since $\widetilde{M}^n_{\tilde{\sigma}_n}= S_{\frac{1}{2}}$ on $\{\tilde{\sigma}_n < \sigma\}$ and $\tilde{\sigma}_n \searrow \frac{1}{2},$ we have that
\begin{align*}
&\lim_{n\to \infty}(-\widetilde{A}^n_{\sigma_n})\\
&= \lim_{n\to\infty}\textstyle(1-\lambda)S_{\frac{1}{2}}\left[\frac{n-1}{n} a_{\frac{1}{2}} b_n + \left(\int^{\tilde{\sigma}_n}_{\frac{1}{2}} na_s b_n \exp \big(-n(s-\frac{1}{2})\big)ds - \frac{n-1}{n} a_{\frac{1}{2}}b_n\right) \exp \big(n(\tilde{\sigma}_n-\frac{1}{2})\big)\right]\\
&=\textstyle (1-\lambda) S_{\frac{1}{2}} a_{\frac{1}{2}} \Big(1-\lim_{n\to \infty} \exp \big(n(\tilde{\sigma}_n - \frac{1}{2})\big)\Big)= S_{\frac{1}{2}}
\end{align*}
and hence
$$\lim_{n \to \infty} \textstyle\widetilde{Q}^n\big[\frac{1}{2} \leq \sigma \leq \tilde{\sigma}_n\big|\cF_{\frac{1}{2}}\big] = \widehat{Y}^0_{\frac{1}{2}} \Big(1-\lim_{n \to \infty} \exp \big(-n(\tilde{\sigma}_n - \frac{1}{2})\big)\Big) = \widehat{Y}^0_{\frac{1}{2}} \frac{\lambda}{\lambda + (1-\lambda) a_{\frac{1}{2}}} = -\Delta_+ \widehat{Y}^0_{\frac{1}{2}}$$
and $\lim_{n \to \infty} (\widetilde{Z}^n_{\tilde{\sigma}_n} - \widetilde{Z}^n_{\frac{1}{2}}) = \Delta_+ \widehat{Y}^0_{\frac{1}{2}}.$
To simplify notation, we set $\tilde{\sigma}_n=\frac{1}{2}$ on $\{\Delta S_{\frac{1}{2}} > 0\}$ and then use the above to define $Q^n=\tQ^n$ on $\F_{\tilde{\sigma}_n}$ and $\tS^n_t=\tM^n_t$ for $t\in(\frac{1}{2},\tilde{\sigma}_n]$.

To obtain an approximation sequence of consistent price systems on $\rrbracket \tilde\sigma_n,1\rrbracket$, we recall that from Proposition \ref{A:Ex1:prop1} the primal and dual optimisers on $\rrbracket \tilde\sigma_n,1\rrbracket$ are determined by the solution $\widehat{\ell}_t$ to the problem
\be\label{A1:Pinfty}
h(\ell,t)=E\Big[\log \Big(1+\ell \big((1-\lambda)(1+a_t(\eta-1)-1\big)\Big)\Big] \to \max!,\quad\ell \in \mathbb{R}.
\ee
In order to approximate the dual optimiser we therefore consider, for $n \in \mathbb{N}$, the auxiliary problems
\begin{align}\label{A1:Pn}
h^n(\ell,t):=E\Big[\log \Big(1+\ell \Big((1-\lambda)\big(1+a_t(\eta^n-1)\big)-1\Big)\Big)\Big] \to \max!,\quad\ell \in \mathbb{R},
\end{align} 
where
\begin{align*}
\eta^n(\omega):= 
\begin{cases}
\eta(\omega) &: \eta(\omega) \geq \frac{1}{n},\\
\frac{1}{n} &: \eta(\omega) < \frac{1}{n}
\end{cases}
\end{align*}
for $n \in \mathbb{N}.$ These problems can be interpreted as logarithmic utility maximisation problems (without transaction costs) for a one-period price process $\widetilde{R}^{n,t}$ given by $\widetilde{R}^{n,t}_0=1$ and $\widetilde{R}^{n,t}_1:=(1+\lambda)\big(1+a_t(\eta^n-1)\big)$. Therefore we obtain the following lemma from the theory of one-period frictionless utility maximisation.

\begin{lemma}\label{A:Ex1:lA3}
{\bf{1)}} The solution $\widehat{\ell}^n_t$ to problem \eqref{A1:Pn} exists and satisfies, for all $t \in [\frac{1}{2},1]$, that
\begin{align}
&E[\widehat{\eta}^n_t]=1,\label{A:Ex1:l1}\\
&E \Big[ \widehat{\eta}^n_t \Big((1-\lambda)\big(1+a_t(\eta^n-1)\big) - 1\Big) \Big] =0, \label{A:Ex1:l2}
\end{align}
where 
\begin{align}\label{A:Ex1:app4}
\widehat{\eta}^n_t:=\frac{1}{1+ \widehat{\ell}^n_t \Big((1-\lambda)\big(1+a_t(\eta^n-1)\big)-1\Big)}
\end{align}
for all $n \in \mathbb{N}$ and $t \mapsto \widehat{\ell}^n_t$ is continuous.

{\bf{2)}} Moreover, we have that
\begin{align}
\widehat{\ell}^n_t & \longrightarrow\widehat{\ell}_t,\label{Ex1:Pronv1}\\ 
\widehat{\eta}^n_t  &\longrightarrow \frac{1}{1+\widehat{\ell}_t ((1-\lambda)(2+a_t(\eta-1))-1)} =:\widehat{\eta}^\infty_t, \label{Ex1:Pronv2}
\end{align} 
as $n \to \infty,$ for all $t \in [\frac{1}{2},1].$
\end{lemma}
\begin{proof}
1) This part follows essentially from the frictionless duality theory for utility maximisation; see for example \cite{S04b}. Since $\eta^n$ takes for $n \in \mathbb{N}$ only finitely many different values, the solution $\widehat{\ell}^n_t$ to problem \eqref{A1:Pn} is determined by $(h^n)'(\widehat{\ell}^n_t,t)=0$, where 
\begin{align*}
h^n(\ell, t)={}&E\Big[\log \Big(1 + \ell \big((1-\lambda)(1+a_t(\eta^n-1)\big)-1\Big)\Big]\\
={}&(1-\varepsilon) \log \Big(1+\ell\big((1-\lambda)a_t-\lambda\big)\Big)\\
&+\sum^{n-1}_{m=1} \varepsilon^{-m} \log \textstyle\Big(1-\ell \big((1-\lambda)a_t(1-\frac{1}{m})+\lambda)\big)\Big)\\
&+\varepsilon 2^{-n+1} \log \textstyle\Big(1-\ell \big((1-\lambda)a_t(1-\frac{1}{n})+\lambda)\big)\Big) 
\end{align*} 
and
\begin{align*}
\frac{\partial h^n}{\partial \ell} (\ell, t)={}&E\left[\frac{(1-\lambda)\big(1+a_t(\eta^n-1)\big)-1}{1+\ell\big((1-\lambda)(1+a_t(\eta^n-1))-1}\right]\\
={}&(1-\varepsilon) \frac{(1-\lambda) a_t - \lambda}{1+ \ell\big((1-\lambda)a_t2-1\big)}\\
&-\sum^{n-1}_{m=1} \varepsilon 2^{-m}\frac{(1-\lambda) a_t (1-\frac{1}{m}) + \lambda}{1-\ell ((1-\lambda)a_t(1-\frac{1}{m}) + \lambda)}-\varepsilon 2^{-m+1}\frac{(1-\lambda) a_t (1-\frac{1}{m}) + \lambda}{1-\ell ((1-\lambda)a_t(1-\frac{1}{m}) + \lambda)}.
\end{align*}
Since $(\frac{\partial^2}{\partial \ell^2}h^n)(\widehat{\ell}^n_{t}, t) > 0$ by concavity, we obtain by an application of the implicit function theorem that $t \mapsto \widehat{\ell}^n_t$ is continuous. Since $\lim_{n\to \infty} (\frac{\partial}{\partial \ell}h^n)(\widehat{\ell}_t) = (\frac{\partial}{\partial \ell}h) (\widehat{\ell}_t,t)>0$ for $\widehat{\ell}_t=\frac{1}{\lambda + (1-\lambda)\alpha_t}$ and $\lim_{\ell\nearrow \ell^{n,\max}_t}(h^n)'(\ell,t):= \frac{1}{(1-\lambda)a_t(1-\frac{1}{n}) + \lambda},$ we obtain that $\widehat{\ell}^n_t \in (\widehat{\ell}^\infty_t, \ell^{n,\max}_t)$ for sufficiently large $n$ and therefore \eqref{Ex1:Pronv1} and \eqref{Ex1:Pronv2}.
\end{proof}

After these preparations, we have now everything in place to give the approximating sequence of $\lambda$-consistent price systems.

\begin{prop}\label{A:Ex1:propZn}
Let the processes $(\widetilde{Z}^n_t)_{0 \leq t \leq 1}, (\widetilde{M}^n_t)_{\frac{1}{2} \leq t \leq 1}$ and $(\widehat{\eta}^n_t)_{\frac{1}{2} \leq t \leq 1}$ be as defined in \eqref{A:Ex1:app1}, \eqref{A:Ex1:app2}, and \eqref{A:Ex1:app4} and the stopping times $\widetilde{\sigma}_n$ as in \eqref{A:Ex1:st}. Then:
\begin{itemize}
\item[{\bf1)}] The processes $Z^n=(Z^{0,n}, Z^{1,n})$ given by 
\begin{align*}
Z^{0,n}_t=
\begin{cases}
\widehat{Y}^0_t &: 0 \leq t \leq \frac{1}{2},\\
\widetilde{Z}^n_t &: \frac{1}{2} < t \leq \widetilde{\sigma}_n,\\
\widetilde{Z}^n_{\widetilde{\sigma}_n} \big(1 + (\widehat{\eta}^n_t - 1) \mathbbm{1}_{\llbracket \sigma, 1 \rrbracket}(t)\big)&: \widetilde{\sigma}_n < t \leq 1
\end{cases}
\end{align*}
and
\begin{align*}
Z^{1,n}_t=
\begin{cases}
\widehat{Y}^1_t &: 0 \leq t \leq \frac{1}{2},\\
\widetilde{Z}^n_t \widetilde{M}^n_t&: \frac{1}{2} < t \leq \widetilde{\sigma}_n,\\
\widetilde{Z}^n_{\widetilde{\sigma}_n} \big(1 + (\widehat{\eta}^n_t - 1) [(1-\lambda ) (1+a_t(\eta^n-1 ) -1)] \mathbbm{1}_{\llbracket \sigma, 1 \rrbracket}(t) \big)S_{\frac{1}{2}}&: \widetilde{\sigma}_n < t \leq 1
\end{cases}
\end{align*}
are martingales and, in particular, $\lambda$-consistent price systems (for sufficiently large $n$).
\item[{\bf2)}] We have that
$$(Z^{0,n}_\tau, Z^{1,n}_\tau) \stackrel{P}{\longrightarrow} (\widehat{Y}^0_\tau, \widehat{Y}^1_\tau), \quad \mbox{as} \quad n \to \infty, $$ 

\end{itemize}
for all finite stopping times $\tau.$
\end{prop}

\begin{proof}
1) From the first order condition $f'(\widehat{\ell}_{\frac{1}{2}})=E[\widehat{Y}^0_{\frac{1}{2}} (\widehat{S}_{\frac{1}{2}} - \widehat{S}_0)]=0$ in \eqref{A.8} we obtain that 
$$E[\widehat{Y}^0_{\frac{1}{2}}|\F_t]=1 - \hvp^1_{\frac{1}{2}} E[\widehat{Y}^0_{\frac{1}{2}}(\widehat{S}_{\frac{1}{2}}-\widehat{S}_0)|\F_t]=1= \widehat{Y}^0_t,\quad 0\textstyle \leq t \leq \frac{1}{2},$$
and therefore that $(\widehat{Y}^0_t)_{0 \leq t \leq \frac{1}{2}}$ and hence $(Z^{0,n}_t)_{0 \leq t \leq \frac{1}{2}}$ are martingales. This also implies that
$$E[\Delta \widehat{Y}^1_{\frac{1}{2}}|\F_{\frac{1}{2}-}] = E[ \widehat{Y}^0_{\frac{1}{2}} \Delta \widehat{S}_{\frac{1}{2}}] + E[\Delta \widehat{Y}^0_{\frac{1}{2}} S_0]=0$$
and therefore that $(\widehat{Y}^1_t)_{0 \leq t \leq \frac{1}{2}}$ and hence $(Z^{1,n}_t)_{0 \leq t \leq \frac{1}{2}}$ are martingales.

The martingale property of $Z^{0,n}$ on $(\frac{1}{2}, \widetilde{\sigma}_n]$ follows from the definition of $\widetilde{Z}^n$ as density process of $\widetilde{Q}^n$ and that of $Z^{1,n}$ on $(\frac{1}{2}, \widetilde{\sigma}_n]$ by Bayes formula, since $(\widetilde{M}^n_t)_{\frac{1}{2} \leq t \leq 1}$ is a $\widetilde{Q}$-martingale. That $Z^{0,n}$ and $Z^{1,n}$ are martingales on $(\widetilde{\sigma}_n, 1]$ as well then follows from the fact that $\eta^n$ and $\widehat{\eta}^n$ are $\sigma(\eta)$-measurable, $\eta$ is independent of $\sigma$ and one can therefore verify the martingale condition directly by using \eqref{A:Ex1:l1} and \eqref{A:Ex1:l2}.

To conclude that $Z^n=(Z^{0,n}, Z^{1,n})$ is a $\lambda$-consistent price system, it remains to observe that $\widetilde{S}^n:=\frac{Z^{1,n}}{Z^{0,n}}$ is valued in the bid-ask spread $[(1-\lambda)S,S]$ for sufficiently large $n \geq n(\lambda).$ To see this, we observe that we can fix $n (\lambda) \in \mathbb{N}$ such that $(1-\lambda)\big(1 + a_t(\frac{1}{n}-1)\big) < 1 - a_t$ for all $n \geq n (\lambda)$ and therefore have $(1-\lambda) S_t \leq \widetilde{S}^n_t \leq S_t$ for all $n \geq n (\lambda).$

2) For the proof of the convergence in probability at all finite stopping times, let $\tau$ be any finite stopping time and recall that 
$\widetilde{\sigma}_n \stackrel{P}{\longrightarrow} \frac{1}{2},$
$\widetilde{Z}^n_{\widetilde{\sigma}_n} \stackrel{P}{\longrightarrow} \widetilde{Y}^0_{\frac{1}{2}+},$
$\widetilde{M}_{\widetilde{\sigma}_n} \stackrel{P}{\longrightarrow} S_{\frac{1}{2}}$
and $\widehat{\eta}^n \stackrel{P}{\longrightarrow} \widehat{\eta}^\infty$, as $n \to \infty$, by the definitions and discussions above.
Therefore
\begin{align*}
Z^{0,n}_\tau ={}& \widehat{Y}^0_\tau \mathbbm{1}_{\{ \tau \leq \frac{1}{2}\}} + \widetilde{Z}^n_\tau \mathbbm{1}_{\{ \frac{1}{2} < \tau \leq \widetilde{\sigma}_n\}} + \widetilde{Z}^n_{\widetilde{\sigma}_n} (1 + (\widehat{\eta}^n - 1) \mathbbm{1}_{\llbracket \sigma, 1 \rrbracket} (\tau) \mathbbm{1}_{\{\widetilde{\sigma}_n < \tau \leq 1\} } \stackrel{P}{\longrightarrow} \widehat{Y}^0_\tau, \\
Z^{1,n}_\tau ={}& \widehat{Y}^1_\tau \mathbbm{1}_{\{\tau \leq \frac{1}{2}\}} +\widetilde{Z}^n_\tau \widetilde{M}^n_\tau \mathbbm{1}_{\{\frac{1}{2} < \tau \leq \widetilde{\sigma}_n \}} \\
&+ \widetilde{Z}^n_{\widetilde{\sigma}_n} (1+(\widehat{\eta}^n-1 ) \mathbbm{1}_{\llbracket \sigma, 1 \rrbracket} (\tau ) )S_{\frac{1}{2}} (1+a_\tau (\eta^n-1)\mathbbm{1}_{\llbracket \sigma, 1 \rrbracket} (\tau )) \mathbbm{1}_{\{ \widetilde{\sigma}_n < \tau \leq 1\}} \stackrel{P}{\longrightarrow} \widehat{Y}^1_\tau,
\end{align*}
as $ n \to \infty$, as $\mathbbm{1}_{\{ \frac{1}{2}< \tau \leq \widetilde{\sigma}_n\}} \stackrel{P}{\longrightarrow} 0$ and $\mathbbm{1}_{\{ \widetilde{\sigma}_n < \tau \leq 1\}} \stackrel{P}{\longrightarrow} \mathbbm{1}_{\{ \frac{1}{2}< \tau \leq 1\}}.$
\end{proof}

\subsection{Left limit of limits $\ne$ limit of left limits}\label{App:Ex2}
Let us now turn to the second example discussed in Section \ref{sec:ex2}.
\begin{prop}\label{A:Ex2:prop}
Let $\varepsilon \in (0, \frac{1}{3})$ and fix $\lambda \in (0,1)$ sufficiently small. Then: 

{\bf{1)}} The solution $\hvp=(\hvp^0_t, \hvp^1_t)_{0 \leq t \leq 1} \in \cA(1)$
to the problem 
\begin{equation}\label{A:Ex2:PP}
E\big[\log \big(V_1^{liq}(\vp)\big )\big] \to\max!, \quad \varphi \in \cA(1),
\end{equation}
for the ask price $S=(S_t)_{0 \leq t \leq 1}$ given by \eqref{Ex2:S} exists and is given by
\begin{align*}
\hvp^1_t = \sum^\infty_{j=1}\hvp^1_{t_j} \mathbbm{1}_{(t_{j-1}, t_j]}(t) + \hvp^1_{\frac{1}{2}} \mathbbm{1}_{[\frac{1}{2},1]}(t)
\end{align*}
for $t \in [0,1]$, where
\begin{align*}
\hvp^1_{t_1}&= \Delta_+ \hvp^1_0 = \frac{1}{\lambda +(1-\lambda )a_1}=:\widehat{\pi}_{t_1}> 0,\\
\hvp^1_{t_j}&= (1-\lambda \hvp^1_{t_1}) \frac{1}{(1-\lambda ) a_j}  \mathbbm{1}_{\{ \sigma > t_{j-1}\}}=:(1-\lambda \hvp^1_{t_1})\widehat{\pi}_{t_j},\quad j \geq 2,\\
\hvp^1_{\frac{1}{2}-}&=\lim_{j\to\infty}\hvp^1_{t_j}= (1-\lambda \hvp^1_{t_1}) \frac{1}{(1-\lambda) a_\infty} \mathbbm{1}_{\{\sigma=\frac{1}{2}\}}=(1-\lambda \hvp^1_{t_1})\widehat{\pi}_{\frac{1}{2}-},\\
\hvp^1_{\frac{1}{2}} &= \big(1-\lambda (\hvp^1_{t_1}- \hvp^1_{\frac{1}{2}-})\big) \frac{1}{\lambda + (1-\lambda ) \frac{1}{2}} \mathbbm{1}_{ \{\sigma = \frac{1}{2} \} }=:\big(1-\lambda (\hvp^1_{t_1}- \hvp^1_{\frac{1}{2}-})\big)\widehat{\pi}_{\frac{1}{2}}
\end{align*}
and $\hvp^0_0=1$ and $d\hvp^0$ is determined by the self-financing conditions \eqref{eq:sf2.1} - \eqref{eq:sf2.3} with equality.

{\bf{2)}} The solution $\widehat{Y} = (\widehat{Y}^0, \widehat{Y}^1)$ to the dual problem
\be
E[-\log(Y^0_T)-1] \to \min!, \quad Y=(Y^0, Y^1) \in \cB\big(\hy (x) \big),\label{A:Ex2:DP}
\ee
for $\hy (x) = u'(x)=\frac{1}{x}=1$ exists and is given by
\be
(\hY^0,\hY^1)=\left(\frac{1}{\hvp^0 + \hvp^1\hS},\frac{\widehat{S}}{\hvp^0 + \hvp^1\widehat{S}}\right),\label{A:Ex2:DO}
\ee
where
\be
\hS_t=1+\big((1-\lambda ) S_{t_1} - S_{t_1-}\big)\mathbbm{1}_{\{t_1\leq t\}}+ \sum^\infty_{j=2}(1-\lambda ) \Delta S_{t_j}\mathbbm{1}_{\{t_j\leq t\}}+\big((1-\lambda ) S_{\frac{1}{2}} - S_{\frac{1}{2}-}\big)\mathbbm{1}_{\{\frac{1}{2}\leq t\}}\label{A:Ex2:hS}
\ee
and
\begin{align}
\hvp^0_t + \hvp^1_t\hS_t={}&1+\hvp^1_{t_1}\big( (1-\lambda ) S_{t_1} - S_{t_1-}\big)\mathbbm{1}_{\{t_1\leq t\}}+ \sum^\infty_{j=2} \hvp^1_{t_j} (1-\lambda ) \Delta S_{t_j}\mathbbm{1}_{\{t_j\leq t\}}\nonumber\\
&\quad+ \hvp^1_{\frac{1}{2}-} \big(S_{\frac{1}{2}-} - (1-\lambda )S_{\frac{1}{2}-}\big)\mathbbm{1}_{\{\frac{1}{2}\leq t\}}+ \hvp^1_{\frac{1}{2}} \big((1-\lambda ) S_{\frac{1}{2}} - S_{\frac{1}{2}-}\big)\mathbbm{1}_{\{\frac{1}{2}\leq t\}}\nonumber\\
={}&\Big(1+\widehat{\pi}_{t_1} \big((1-\lambda ) S_{t_1} - S_{t_1-}\big)\mathbbm{1}_{\{t_1\leq t\}}\Big)\prod^\infty_{j=2}\big(1+\widehat{\pi}_{t_j}(1-\lambda ) \Delta S_{t_j}\mathbbm{1}_{\{t_j\leq t\}}\big)\nonumber\\
&\quad\times\Big(1+\widehat{\pi}_{\frac{1}{2}-} \big(S_{\frac{1}{2}-} - (1-\lambda )S_{\frac{1}{2}-}\big)\mathbbm{1}_{\{\frac{1}{2}\leq t\}}\Big)\Big(1+\widehat{\pi}_{\frac{1}{2}} \big((1-\lambda) S_{\frac{1}{2}} - S_{\frac{1}{2}-}\big)\mathbbm{1}_{\{\frac{1}{2}\leq t\}}\Big)\label{A:Ex2:dw}
\end{align}
for $t\in[0,1]$.
\end{prop}

\begin{proof}
1) We begin with the solution to the primal problem \eqref{A:Ex2:PP}. As already explained in Section \ref{sec:ex2}, since $S$ is constant on $[t_{j-1}, t_j)$ for $j \in \mathbb{N}$ and $[\frac{1}{2},1]$, it does not matter, where the positions are rebalanced during the intervals $[t_{j-1}, t_j)$ and we can therefore assume that the trades take place immediately after time $t_{j-1}$ for $j \in \mathbb{N}$ and there is no trading after time $\frac{1}{2}$. Next we recall that trading for any price within the bid-ask spread is always more favourable than trading under transaction costs. So we have
\begin{align}
V_1^{liq}(\vp) \leq 1 &+ \vp^1_{t_1}\big( (1-\lambda ) S_{t_1} - S_{t_1-}\big)+ \sum^\infty_{j=2} \vp^1_{t_j} (1-\lambda ) \Delta S_{t_j}\notag\\
&+ \vp^1_{\frac{1}{2}-} \big(S_{\frac{1}{2}-} - (1-\lambda )S_{\frac{1}{2}-}\big)+ \vp^1_{\frac{1}{2}} \big((1-\lambda ) S_{\frac{1}{2}} - S_{\frac{1}{2}-}\big)\label{A:Ex2:eq1}
\end{align}
for all $\varphi \in \cA(1)$. By the scaling of the logarithm this allows us to estimate 
\begin{align}
E\big[\log \big(V_1^{liq}(\vp)\big)\big] \leq &E\Big[\log\Big(1+\pi_{t_1} \big((1-\lambda ) S_{t_1} - S_{t_1-}\big)\Big)\Big]\notag\\
&+ \sum^\infty_{j=1} E\big[\log\big(1+\pi_{t_j}(1-\lambda ) \Delta S_{t_j}\big)\big]\notag\\
&+ E\Big[\log \Big(1+\pi_{\frac{1}{2}-} \big(S_{\frac{1}{2}-} - (1-\lambda )S_{\frac{1}{2}-}\big)\Big)\Big]\notag\\
&+ E\Big[\log \Big(1+\pi_{\frac{1}{2}} \big((1-\lambda) S_{\frac{1}{2}} - S_{\frac{1}{2}-}\big)\Big)\Big], \label{A:Ex2:eq2}
\end{align} 
where
\begin{align}
\pi_{t_1} &= \vp^1_{t_1}\notag\\
\pi_{t_j} &= \frac{\vp^1_{t_j}}{1+\vp^1_{t_1} ((1-\lambda) S_{t-1} - S_{t_1-}) + \sum^{j-1}_{k=2} \vp^1_{t_k} (1-\lambda) \Delta S_{t_k}},\notag\\
\pi_{\frac{1}{2}-}&= \frac{\vp^1_{\frac{1}{2}-}} {1+ \vp^1_{t_1}((1-\lambda ) S_{t_1} - S_{t_1-}) + \sum^\infty_{k=2} \vp^1_{t_k} (1-\lambda ) \Delta S_{t_k}},\notag\\
\pi_{\frac{1}{2}} &= \frac{\vp^1_{\frac{1}{2}}} {1 + \vp^1_{t_1} ((1-\lambda ) S_{t_1} - S_{t_1-}) + \sum^\infty_{k=2} \vp^1_{t_k} (1-\lambda ) \Delta S_{t_k} + \vp^1_{\frac{1}{2}-}(S_{\frac{1}{2}-} - (1-\lambda )S_{\frac{1}{2}-})}. \label{A:Ex2:eq3}
\end{align}
The basic idea to derive the optimal trading strategy $\hvp=(\hvp^0_t, \hvp^1_t)_{0 \leq t \leq 1}$ for \eqref{A:Ex2:PP} under transaction costs is now to maximise the right hand side of \eqref{A:Ex2:eq2} over all predictable processes $\pi=(\pi_t)_{0 \leq t \leq 1}$ and to show that this allows us to define by solving \eqref{A:Ex2:eq3} a self-financing strategy under transaction costs such that we have equality in \eqref{A:Ex2:eq1}. For this, we observe that we can maximise the terms on the right hand side of \eqref{A:Ex2:eq2} independently of each other and only need to solve
\begin{align}
&E\Big[\log \Big(1+\pi_{t_1}\big((1-\lambda ) S_{t_1} - S_{t_1-}\big)\Big)\Big] \to \max!, &\pi_{t_1} \in \F_{t_1-},\label{A:Ex2:P1}\\ 
&E\Big[\log \Big(1+\pi_{t_j}(1-\lambda ) \Delta S_{t_j} \Big)\Big] \to \max!, &\pi_{t_j} \in \F_{t_j-},\label{A:Ex2:P2}\\ 
&E\Big[\log \Big(1+\pi_{\frac{1}{2}}\big((1-\lambda ) S_{\frac{1}{2}} - S_{\frac{1}{2}-}\big)\Big)\Big] \to \max!,& \pi_{\frac{1}{2}} \in \F_{\frac{1}{2}-}, \label{A:Ex2:P3}
\end{align}
where the solutions are, as explained below, given by
\begin{align}
\widehat{\pi}_{t_1}& = \frac{1}{\lambda + (1-\lambda) a_1},\label{A:Ex2:S1}\\
\widehat{\pi}_{t_j} &= \frac{1}{(1-\lambda) a_j} \mathbbm{1}_{\{ \sigma > t_{j-1}\}}, \quad j \geq 2, \label{A:Ex2:S2}\\
\widehat{\pi}_{\frac{1}{2}} &= \frac{1}{\lambda + (1-\lambda ) \frac{1}{2}} \mathbbm{1}_{\{ \sigma = \frac{1}{2}\}}. \label{A:Ex2:S3}
\end{align}
This will also give the optimal value for 
\begin{align}\label{A:Ex2:eq4}
E\Big[\log\Big(1 +\pi_{\frac{1}{2}-} \big(S_{\frac{1}{2}-} - (1-\lambda) S_{\frac{1}{2}-}\big)\Big)\Big] \to \max!
\end{align}
when maximising over all possible limits $\pi_{\frac{1}{2}-}=\lim_{j \to \infty} \pi_{t_j}$ of processes $\pi=(\pi_t)_{0 \leq t \leq 1}$ for which the problems \eqref{A:Ex2:P2} are well-defined $>-\infty$, i.e.~$\widehat{\pi}_{t_j} \leq \frac{1}{(1-\lambda)a_j} \mathbbm{1}_{\{\sigma > t_j\}}$ for all $j\geq 2$. As $S_{\frac{1}{2}-}-(1-\lambda) S_{\frac{1}{2}-} = \lambda S_{\frac{1}{2}-} > 0$ on $\{\sigma=\frac{1}{2}\}$, this is precisely the upper boundary $\widehat{\pi}_{\frac{1}{2}-}= \lim_{j \to \infty} \widehat{\pi}_{t_j}= \frac{1}{(1-\lambda) \frac{2}{3}}$ of the domain of \eqref{A:Ex2:eq4}.

The proof that the solution to problems \eqref{A:Ex2:P1} - \eqref{A:Ex2:P3} are given by \eqref{A:Ex2:S1} - \eqref{A:Ex2:S3} follows by the same arguments as that of Proposition \ref{A:Ex1:prop1} and is therefore omitted. Note, however, that these arguments use that $\varepsilon\in(0,\frac{1}{3})$ and $\lambda$ is sufficiently small, as in \eqref{B.1}.

To conclude we only need to observe that defining $\hvp=(\hvp^0_t, \hvp^1_t)_{0 \leq t \leq 1}$ by solving \eqref{A:Ex2:eq3}, i.e.~by $(\hvp^0_0, \hvp^1_0)=(1,0)$,
\begin{align*}
\hvp^1_{t_1}&=\widehat{\pi}_{t_1},\\
\hvp^1_{t_j}&=\left(1+\hvp^1_{t_1}\big((1-\lambda) S_{t_1}-S_{t_1}\big) + \sum^{j-1}_{k=2} \hvp^1_{t_k} (1-\lambda) \Delta S_{t_k}\right) \widehat{\pi}_{t_j}\\
&=(1-\lambda \hvp^1_{t_1}) \frac{1}{(1-\lambda)a_j} \mathbbm{1}_{\{ \sigma > t_{j-1}\}},\quad j \geq2,\\
\hvp^1_{\frac{1}{2}-}&=\lim_{j\to\infty}(1-\lambda \hvp^1_{t_1}) \frac{1}{(1-\lambda)a_j}\mathbbm{1}_{\{ \sigma > t_{j-1}\}}=(1-\lambda \hvp^1_{t_1}) \frac{1}{(1-\lambda)\frac{2}{3}} \mathbbm{1}_{\{ \sigma =\frac{1}{2}\}},\\
\hvp^1_{\frac{1}{2}}&=\left(1+\hvp^1_{t_1}\big((1-\lambda ) S_{t_1} - S_{t_1-}\big) + \sum^\infty_{k=2} \hvp^1_{t_k}(1-\lambda )\Delta S_{t_k}+ \hvp^1_{\frac{1}{2}-} \big(S_{\frac{1}{2}-} - (1-\lambda ) S_{\frac{1}{2}-}\big)\right) \widehat{\pi}_{\frac{1}{2}}\\
&= \big(1-\lambda (\hvp^1_{t_1}- \hvp^1_{\frac{1}{2}-})\big) \frac{1}{\lambda+(1-\lambda)\frac{1}{2}}
\end{align*}
and $d\hvp^0$ by the self-financing condition \eqref{sfc} with equality gives a self-financing and admissible trading strategy under transaction costs, since
\begin{align}
\Delta_+\hvp^1_0 &=\hvp^1_{t_1}=\frac{1}{\lambda + (1-\lambda) a_1}>0,\nonumber\\
\Delta_+\hvp^1_{t_1}&=\hvp^1_{t_2} - \hvp^1_{t_1}=\frac{1-\lambda \hvp^1_{t_1}}{1-\lambda} \frac{1}{a_2} - \frac{1}{\lambda + (1-\lambda) a_1} <0,\nonumber\\
\Delta_+ \hvp^1_{t_j}&=\hvp^1_{t_{j+1}}-\hvp^1_{t_j} = \frac{1-\lambda\hvp^1_{t_1}}{1-\lambda} \left(\frac{1}{a_{j+1}}-\frac{1}{a_j}\right) <0, \quad j \geq 2,\nonumber\\
\Delta\hvp^1_{\frac{1}{2}}&=\hvp^1_{\frac{1}{2}}-\hvp^1_{\frac{1}{2}-}=\big(1-\lambda (\hvp^1_{t_1}-\hvp^1_{\frac{1}{2}-})\big)\frac{1}{\lambda +(1-\lambda)\frac{1}{2}}-(1-\lambda\hvp^1_{t_1}) \frac{1}{(1-\lambda)\frac{2}{3}} >0,\label{A:Ex2:pr}
\end{align}
where we use that $\lambda$ is sufficiently small in \eqref{A:Ex2:pr}.

2) Again that the solution to the dual problem \eqref{A:Ex2:DP} is given by \eqref{A:Ex2:DO} follows immediately from Proposition \ref{prop:log} and the formulas \eqref{A:Ex2:hS} and \eqref{A:Ex2:dw} from \eqref{mt2:eq2}. 
\end{proof}

Let us now explain how one can construct a sequence $Z^n=(Z^{0,n}, Z^{1,n})$ of $\lambda$-consistent price systems that is approximating the dual optimiser $\hY=(\hY^0,\hY^1)$.

\begin{lemma}\label{A:Ex2:lSn}
The solution $\hvp^n=(\hvp^{0,n}_t, \hvp^{1,n}_t)_{0 \leq t \leq 1}$ to the frictionless utility maximisation problem
\be\label{A:Ex2:flPn}
E[\log (1+\varphi^1 \sint \widehat{S}^n_1 )]\to \max!, \quad \varphi \in \cA (1; \widehat{S}^n),
\ee
for the price process $\widehat{S}^n=(\widehat{S}^n_t)_{0 \leq t \leq 1}$ defined in \eqref{C2} is given by
$$\hvp^{1,n}=\sum^n_{j=1} \hvp^{1,n}_{t_j}\mathbbm{1}_{(t_{j-1}, t_j]}(t)$$
for $t \in [0,1],$ where
\begin{align*}
\hvp^{1,n}_{t_1}&=\widehat{\pi}^n_{t_1} >0,\\
\hvp^{1,n}_{t_j} &= (1-\lambda \hvp^{1,n}_{t_1}) \widehat{\pi}^n_{t_j} \mathbbm{1}_{\{\sigma > t_{j-1}\}}, \quad 2 \leq j \leq n,\\
\hvp^{1,n}_{\frac{1}{2}} &=\big(1-\lambda (\hvp^{1,n}_{t_1} - \hvp^{1,n}_{t_n})\big) \widehat{\pi}^n_{\frac{1}{2}} \mathbbm{1}_{\{\sigma > t_n\}},
\end{align*}
$\hvp^{0,n}$ is defined by the frictionless self-financing condition with equality and $(\widehat{\pi}^n_{t_j})^n_{j=1}$ and $\widehat{\pi}^n_{\frac{1}{2}}$ are the solutions to
\begin{align*}
&E[\log (1+\pi_{t_j} \Delta \widehat{S}^n_{t_j})] \to \max!,\quad \pi_{t_j} \in \F_{t_j-},\quad 1 \leq j \leq n,\\
&E[\log (1+\pi_{\frac{1}{2}} \Delta \widehat{S}^n_{\frac{1}{2}})] \to \max!,\quad \pi_{\frac{1}{2}} \in \F_{\frac{1}{2}-}.
\end{align*}
Moreover, we have that
\begin{align*}
&E[\widehat{\eta}^n_{\frac{1}{2}}] = E[\widehat{\eta}^n_j]=1, \quad 1 \leq j \leq n,\\
&E[\widehat{\eta}^n_{\frac{1}{2}} \Delta \widehat{S}^n_{\frac{1}{2}}]=E[\widehat{\eta}^n_j \Delta\widehat{S}^n_{t_j}]=0, \quad 1 \leq j \leq n,
\end{align*}
where $\widehat{\eta}^n_j = \frac{1}{1 + \widehat{\pi}^n_{t_j} \Delta \widehat{S}^n_{t_j}}$ and $\widehat{\eta}^n_{\frac{1}{2}}=\frac{1}{1 + \widehat{\pi}^n_{\frac{1}{2}} \Delta \widehat{S}^n_{\frac{1}{2}}}.$
\end{lemma}
\bp
The proof follows by similar scaling arguments as that of Proposition \ref{A:Ex2:prop} and is therefore omitted.
\ep
Since $\eta^n \to \eta$, as $n \to \infty$, we obtain as in Lemma \ref{A:Ex1:lA3} that $\widehat{\pi}^n_{t_j} \to \widehat{\pi}_{t_j}$ for $j \in \mathbb{N}$ and $\widehat{\pi}^n_{\frac{1}{2}} \to \widehat{\pi}_{\frac{1}{2}},$ as $n \to \infty$, and therefore also $\hvp^{1,n}_{t_j} \to \hvp^1_{t_j}$ for $j \in \mathbb{N}$ and $\hvp^{1,n}_{\frac{1}{2}} \to \hvp^1_{\frac{1}{2}}$, as $n \to \infty.$ As this implies that $\Delta_+ \hvp^{1,n}_0 > 0,$ $\Delta_+ \hvp^{1,n}_{t_j} < 0$ for $1 \leq j \leq n-1$ and $\Delta_+ \hvp^{1,n}_{t_n} > 0$
for sufficiently large $n$, the optimal strategy for the frictionless utility maximisation problem \eqref{A:Ex2:flPn} coincides with the solution $\hvp^n=(\hvp^{0,n}_t, \hvp_t^{1,n})_{0 \leq t \leq 1}$ to the utility maximisation problem under transaction costs
$$E[\log (V^n_T(\varphi))] \to \max!, \quad \varphi \in \cA^n(x),$$
for the price process $S^n =(S^n_t)_{0 \leq t \leq 1}$ given by \eqref{Ex2:Sn}, where 
$$V^n_T(\varphi):=\varphi^0_T + (\varphi^1_T)^+ (1-\lambda)S^n_T - (\varphi^1_T)^-S^n_T$$
and $\cA^n(x)$ denotes the set of all self-financing and admissible trading strategies under transaction costs $\lambda$ for the price process $S^n$.

For the frictionless dual problem corresponding to \eqref{A:Ex2:flPn} we obtain that the solution $\widehat{Y}^n=(\widehat{Y}^n_t)_{0 \leq t \leq 1}$ is given by 
$$\widehat{Y}^n_t=\frac{1}{1+\hvp^{1,n} \sint \widehat{S}^n_t}=\prod_{j=1}^n\big(1+(\widehat{\eta}^n_j-1)\mathbbm{1}_{\{ t_j\leq t\}}\big)\big(1+(\widehat{\eta}^n_{\frac{1}{2}}-1)\mathbbm{1}_{\{ \frac{1}{2}\leq t\}}\big), \quad t \in [0,1],$$
where
\begin{align}
&1+\hvp^{1,n}\sint \widehat{S}^n_t = 1 + \sum^n_{j=1}\hvp^{1,n}_{t_j} \Delta \widehat{S}_{t_j} + \hvp^{1,n}_{\frac{1}{2}} \Delta \widehat{S}_{\frac{1}{2}} \nonumber\\
&\quad= \Big( 1 + \widehat{\pi}^n_{t_1} \big((1-\lambda) S^n_{t_1} - S^n_{t_1-}\big) \mathbbm{1}_{\{t_1 \leq t\}}\Big)\prod^{n-1}_{j=2} \big(1+ \widehat{\pi}^n_{t_j} (1-\lambda) \Delta S^n_{t_j} \mathbbm{1}_{\{ t_j \leq t \} }\big)\nonumber\\
&\qquad\times \Big(1+\widehat{\pi}^n_{t_n} \big(S^n_{t_n} - (1-\lambda) S^n_{t_n-} \big) \mathbbm{1}_{\{ t_n \leq t \}} \Big)\times\Big(1+ \widehat{\pi}^n_{\frac{1}{2}} \big((1-\lambda) S^n_{\frac{1}{2}} - S^n_{\frac{1}{2}-}\big) \mathbbm{1}_{\{ \frac{1}{2} \leq t \}}\Big),\label{A:Ex2:dwn}
\end{align}
and is the density of an equivalent martingale measure for $\widehat{S}^n$. Therefore
\be
\widehat{Z}^n_t= (\widehat{Z}^{0,n}_t, \widehat{Z}^{1,n}_t):=(\widehat{Y}^n_t, \widehat{Y}^n_t \widehat{S}^n_t), \quad t \in [0,1],
\ee
is a $\lambda$-consistent price system for the price process $S^n=(S^n_t)_{0 \leq t \leq 1}$. Comparing the formulas \eqref{A:Ex2:dw} with \eqref{A:Ex2:dwn} and \eqref{C3} and \eqref{C2} we immediately obtain that
\be
(\widehat{Z}^{0,n}_\tau, \widehat{Z}^{1,n}_\tau) \stackrel{P}{\longrightarrow} (\widehat{Z}^0_\tau, \widehat{Z}^1_\tau), \quad \mbox{as} \quad n \to \infty,\label{A:Ex2:hZn}
\ee
for all finite stopping times $\tau$ and \eqref{C6}, as $\widehat{\pi}^n_{t_j} \to \widehat{\pi}_{t_j}$ for all $j \in \mathbb{N}.$ To turn the $\widehat{Z}^n$'s (for sufficiently large $n$) into $\lambda$-consistent price systems for the price $S=(S_t)_{0 \leq t \leq 1}$, we need to modify the $\widehat{Z}^n$'s on $\{t_n < \sigma < \frac{1}{2}\}$ to obtain martingales $Z^n=(Z^{0,n}_t, Z^{1,n}_t)_{0 \leq t \leq 1}$ such that their ratio $\widetilde{S}^n:= \frac{Z^{1,n}}{Z^{0,n}}$ is valued in the bid-ask spread on $\{t_n < \sigma < \frac{1}{2}\}$ as well.

\begin{prop}
Let $\widehat{Z}^n=(\widehat{Z}^{0,n}, \widehat{Z}^{1,n})$ be as defined in \eqref{A:Ex2:hZn} and $\overline{\eta}$ be a strictly positive $\sigma(\eta)$-measurable random variable such that $E[\overline{\eta}]=1$ and $E[\overline{\eta}(\eta-1)]=0$. Then:
\bi
\item[{\bf1)}] The processes $Z^n=(Z^{0,n}, Z^{1,n})$ given by 
\begin{align*}
Z^{0,n}_t(\omega)&=\begin{cases} \widehat{Y}^n_{t_n\wedge t} (\omega)\big(1 + (\overline{\eta}(\omega) - 1) \mathbbm{1}_{\{\sigma\leq t\}}\big) &:\sigma(\om)\in (t_n, \frac{1}{2}),\\
\widehat{Y}^n_ t&:\text{else},
\end{cases}\\
\widetilde{S}_t(\omega)&= 
\begin{cases}
\widehat{S}^n_t\mathbbm{1}_{\{t<\sigma\}}+S_t(\omega)\mathbbm{1}_{\{\sigma\leq t\}}&: \sigma (\omega) \in (t_n, \frac{1}{2}),\\
\widehat{S}^n_t (\omega)&:\text{else},\\
\end{cases}\\
Z^{1,n}_t(\om)& = Z^{0,n}_t(\om) \widetilde{S}^n_t(\om)
\end{align*}
are (for sufficiently large $n$) $\lambda$-consistent price systems.

\item[{\bf2)}] We have that
\be
(Z^{0,n}_\tau, Z^{1,n}_\tau) \stackrel{P}{\longrightarrow} (\widehat{Y}^0_\tau, \widehat{Y}^1_\tau),\qquad\text{as $n\to \infty$,}\label{A:Ex2:convZn}
\ee
 for all finite stopping times $\tau$ and \eqref{C6}, i.e.
\be
\widetilde{S}^n_{\frac{1}{2}-}= \frac{Z^{1,n}_{\frac{1}{2}-}}{Z^{0,n}_{\frac{1}{2}-}} \stackrel{P}{\longrightarrow} S_{\frac{1}{2}-}, \quad \mbox{as} \quad n \to \infty.\label{A:Ex2:convZn2}
\ee
\ei
\end{prop}
\begin{proof}
1) To see the martingale property of $Z^{0,n}$, we simply observe that the process $M^1_t=\widehat{Y}^n_t = \widehat{Z}^{0,n}_t$ and $M^2_t= \widehat{Y}^n_{t_n\wedge t} (1 + (\overline{\eta}-1) \mathbbm{1}_{ \{\sigma \leq t \}})$ are strictly positive martingales and so their ``fork convex'' combination or ``pasting''
\begin{align*}
Z^{0,n}_t=
\begin{cases}
M^1_t &: 0 \leq t < t_n,\\
M^1_{t_n}\left(\mathbbm{1}_F \sint \frac{M^1_t}{M^1_{t_n}} + \mathbbm{1}_{F^c} \sint \frac{M^2_t}{M^2_{t_n}}\right)&: t_n \leq t \leq 1
\end{cases}
\end{align*}
is a martingale as well, where the predictable set $F$ is given by $F:= \cup^\infty_{j=n} \{\sigma > t_j\} \times (t_j, 1).$

Similarly, we obtain that
\begin{align*}
Z^{1,n}_t=
\begin{cases}
N^1_t&: 0 \leq t < t_n,\\
N^1_{t_n}\left(\mathbbm{1}_F \sint \frac{N^1_t}{N^1_{t_n}} + \mathbbm{1}_{F^c} \sint \frac{N^2_t}{N^2_{t_n}}\right)&: t_n \leq t \leq 1
\end{cases}
\end{align*}
is a martingale, since $N^1_t=\widehat{Z}^{1,n}_t$ and $N^2_t= \widehat{Z}^{1,n}_{t_n \wedge t}\big(1+(\overline{\eta}-1) S_t \mathbbm{1}_{\{ \sigma \leq t\}}\big)$ are.

That $Z^n=(Z^{0,n}, Z^{1,n})$ is for sufficiently large $n$ a $\lambda$-consistent price system then follows from the fact that $\widetilde{S}^n_t$ is valued in the bid-ask spread $[(1-\lambda) S_t, S_t]$ for $t < \sigma$ as well as $\sigma \leq t$ for sufficiently large $n$.

2) The convergences \eqref{A:Ex2:convZn} and \eqref{A:Ex2:convZn2} then simply follow from \eqref{A:Ex2:hZn} and \eqref{C6} after observing that $P\big(\sigma \in (t_n, \frac{1}{2})\big) \longrightarrow 0,$ as $n \to \infty.$
\end{proof}


\begin{thebibliography}{10}
\bibitem{BY13}
E.~Bayraktar and X.~Yu.
\newblock {On the Market Viability under Proportional Transaction Costs}.
\newblock {\em Preprint}, 2013.

\bibitem{BCKMK11}
G.~Benedetti, L.~Campi, J.~Kallsen, and J.~Muhle-Karbe.
\newblock On the existence of shadow prices.
\newblock {\em Finance and Stochastics}, 17(4):801--818, 2013.

\bibitem{CO11}
L.~Campi and M.~P. Owen.
\newblock Multivariate utility maximization with proportional transaction
  costs.
\newblock {\em Finance Stoch.}, 15(3):461--499, 2011.

\bibitem{CS06}
L.~Campi and W.~Schachermayer.
\newblock A super-replication theorem in {K}abanov's model of transaction
  costs.
\newblock {\em Finance Stoch.}, 10(4):579--596, 2006.

\bibitem{CG79}
K.~L. Chung and J.~Glover.
\newblock Left continuous moderate {M}arkov processes.
\newblock {\em Z. Wahrsch. Verw. Gebiete}, 49(3):237--248, 1979.

\bibitem{CK96}
J.~Cvitani{\'c} and I.~Karatzas.
\newblock Hedging and portfolio optimization under transaction costs: a
  martingale approach.
\newblock {\em Math. Fin.}, 6(2):113--165, 1996.

\bibitem{CW01}
J.~Cvitani{\'c} and H.~Wang.
\newblock On optimal terminal wealth under transaction costs.
\newblock {\em J. Math. Econom.}, 35(2):223--231, 2001.

\bibitem{CMKS14}
C.~Czichowsky, J.~Muhle-Karbe, and W.~Schachermayer.
\newblock Transaction costs, shadow prices, and duality in discrete time.
\newblock {\em SIAM Journal on Financial Mathematics}, 5(1):258--277, 2014.

\bibitem{CS13}
C.~Czichowsky and W.~Schachermayer.
\newblock Strong supermartingales and limits of non-negative martingales.
\newblock {\em Preprint}, 2013.

\bibitem{CS14}
C.~Czichowsky, W.~Schachermayer, and J.~Yang.
\newblock Shadow prices for continuous price processes.
\newblock {\em Preprint}, 2014.

\bibitem{DPT01}
G.~Deelstra, H.~Pham, and N.~Touzi.
\newblock Dual formulation of the utility maximization problem under
  transaction costs.
\newblock {\em Ann. Appl. Probab.}, 11(4):1353--1383, 2001.

\bibitem{DM78}
C.~Dellacherie and P.-A. Meyer.
\newblock {\em Probabilities and potential}, volume~29 of {\em North-Holland
  Mathematics Studies}.
\newblock North-Holland Publishing Co., Amsterdam, 1978.

\bibitem{DM82}
C.~Dellacherie and P.~A. Meyer.
\newblock {\em Probabilities and {P}otential {B}. {T}heory of {M}artingales}.
\newblock North-Holland, 1982.

\bibitem{HP91}
H.~He and N.~D. Pearson.
\newblock Consumption and portfolio policies with incomplete markets and
  short-sale constraints: the infinite-dimensional case.
\newblock {\em J. Econom. Theory}, 54(2):259--304, 1991.

\bibitem{JK95}
E.~Jouini and H.~Kallal.
\newblock Martingales and arbitrage in securities markets with transaction
  costs.
\newblock {\em J. Econom. Theory}, 66(1):178--197, 1995.

\bibitem{KMK10}
J.~Kallsen and J.~Muhle-Karbe.
\newblock On using shadow prices in portfolio optimization with transaction
  costs.
\newblock {\em The Annals of Applied Probability}, 20(4):1341--1358, 2010.

\bibitem{KMK11}
J.~Kallsen and J.~Muhle-Karbe.
\newblock Existence of shadow prices in finite probability spaces.
\newblock {\em Math. Methods Oper. Res.}, 73(2):251--262, 2011.

\bibitem{KK07}
I.~Karatzas and C.~Kardaras.
\newblock The num\'eraire portfolio in semimartingale financial models.
\newblock {\em Finance Stoch.}, 11(4):447--493, 2007.

\bibitem{KLSX91}
I.~Karatzas, J.~P. Lehoczky, S.~E. Shreve, and G.-L. Xu.
\newblock Martingale and duality methods for utility maximization in an
  incomplete market.
\newblock {\em SIAM J. Control Optim.}, 29(3):702--730, 1991.

\bibitem{K10}
C.~Kardaras.
\newblock Market viability via absence of arbitrage of the first kind.
\newblock {\em Finance and Stochastics}, 16(4):651--667, 2012.

\bibitem{KS99}
D.~Kramkov and W.~Schachermayer.
\newblock The asymptotic elasticity of utility functions and optimal investment
  in incomplete markets.
\newblock {\em Ann. Appl. Probab.}, 9(3):904--950, 1999.

\bibitem{L00}
M.~Loewenstein.
\newblock On optimal portfolio trading strategies for an investor facing
  transactions costs in a continuous trading market.
\newblock {\em J. Math. Econom.}, 33(2):209 -- 228, 2000.

\bibitem{M72}
J.-F. Mertens.
\newblock Th\'eorie des processus stochastiques g\'en\'eraux applications aux
  surmartingales.
\newblock {\em Z. Wahrscheinlichkeitstheorie und Verw. Gebiete}, 22:45--68,
  1972.

\bibitem{Pliska}
S.~R. Pliska.
\newblock A stochastic calculus model of continuous trading: optimal
  portfolios.
\newblock {\em Math. Oper. Res.}, 11(2):370--382, 1986.

\bibitem{S04b}
W.~Schachermayer.
\newblock Portfolio optimization in incomplete financial markets.
\newblock {\em Notes of the Scuola Normale Superiore Cattedra Galileiana,
  Pisa}, 2004.

\bibitem{S13}
W.~Schachermayer.
\newblock {Admissible trading strategies under transaction costs}.
\newblock {\em Preprint, to appear in Seminaire de Probabilite}, 2014.

\bibitem{S14}
W.~Schachermayer.
\newblock {The super-replication theorem under proportional transaction costs
  revisited}.
\newblock {\em Preprint, to appear in Mathematics and Financial Economics},
  2014.

\end{thebibliography}
\bibliographystyle{abbrv}

\end{document}